\documentclass{CSML}

\def\dOi{11(4:16)2015}
\lmcsheading%
{\dOi}
{1--40}
{}
{}
{Sep.~16, 2015}
{Dec.~22, 2015}
{}

\ACMCCS{[{\bf Theory of computation}]: Semantics and
  reasoning---Program semantics; [{\bf Software
    and its engineering}]: Software organization and
  properties---Formal methods} 

\usepackage{hyperref}

\usepackage{amsmath}
\usepackage{dsfont}

\usepackage{mathrsfs}
\usepackage{color}
\usepackage{latexsym}

\def\SET#1{\lbrace #1 \rbrace}
\def\MSET#1#2{\mlbrace \, #1 \mid #2 \, \mrbrace}
\def\ZSET#1#2{\lbrace \, #1 \mid #2 \, \rbrace}

\def\calA{\mathcal{A}}

\def\calF{\mathcal{F}}

\def\calL{\mathcal{L}}
\def\calM{\mathcal{M}}

\def\calP{\mathcal{P}}

\def\calR{\mathcal{R}}
\def\calS{\mathcal{S}}

\def\calU{\mathcal{U}}
\def\calV{\mathcal{V}}
\def\calW{\mathcal{W}}
\def\calX{\mathcal{X}}
\def\calY{\mathcal{Y}}

\newcommand{\blankline}{\vspace*{1.0\baselineskip}}
\newcommand{\halflineup}{\vspace*{-0.5\baselineskip}}

\def\beq#1#2{\approx^{{#2}}_{{#1}}}	
		
\def\FuTSbis#1{\simeq_{{#1}}}
\def\twoFuTSbis#1{\simeq_{2 \mkern1mu {#1}}}
\def\tFuTSbis#1{\simeq_{2 \mkern1mu {#1}}}
\def\twoFuTSbis#1{\simeq_{{#1}}}
\def\tFuTSbis#1{\simeq_{{#1}}}
\def\notFuTSbis#1{\not\simeq_{{#1}}}

\def\IMLsbis{\sim_{\mathit{iml}}}
\def\TPCsbis{\sim_{\mathit{tpc}}}
\def\MALsbis{\sim_{\textsl{ma}}}
\def\PEPAseq{\sim_{\mathit{pepa}}}

\newcommand{\sync}[1]{
  \ensuremath{%
    \mathop{%
      \mkern3mu
  \ooalign{%
    \mbox{\large $\rhd\!\!\!\lhd$}
    \cr
    \raise-4pt\hbox{$\mkern3mu \scriptstyle \mathrm{#1}$}
  }
  \mkern-2mu
}}}

\def\vcalL{\underline{{\cal L}}}
\def\vcalR{\underline{{\cal R}}}
\def\calVvLR{\calV^{\vcalL}_{\!\vcalR}}

\def\bools{\mathbb{B}}
\def\nats{\mathbb{N}}	
\def\reals{\mathbb{R}}

\newcommand{\Set}{\textbf{Set}}
\newcommand{\nil}{\textbf{nil}}
\newcommand{\bfI}{\textbf{I}}
\newcommand{\bfM}{\textbf{M}}
\newcommand{\bfR}{\textbf{R}}
\newcommand{\bfT}{\textbf{T}}

\def\min{\mathrm{\textrm{min}}}

\newcommand{\CTMC}{\textsl{CTMC}}

\newcommand{\Distr}{\textsl{Distr} \mkern2mu}
\newcommand{\DTMC}{\textsl{DTMC}}
\newcommand{\FuTS}{\textsl{Fu\hspace*{-0.5pt}TS} \mkern1mu}
\newcommand{\pFuTS}{\textsl{pFu\hspace*{-0.5pt}TS} \mkern1mu}

\mathchardef\mhyphen="2D

\newcommand{\IMC}{\textsl{IMC} \mkern1mu}
\newcommand{\IML}{\textsl{IML} \mkern1mu}
\newcommand{\MAL}{\textsl{MAL} \mkern1mu}
\newcommand{\LTS}{\textsl{LTS} \mkern1mu}

\newcommand{\MA}{\textsl{MA} \mkern1mu}
\newcommand{\PEPA}{\textsl{PEPA}}
\newcommand{\RTS}{\textsl{RTS}}

\newcommand{\TPC}{\textsl{TPC}}

\newcommand{\arf}{\textsl{arf} \mkern1mu}
\newcommand{\cnt}{\textsl{cnt} \mkern2mu}

\newcommand{\spt}{\textsl{spt\hspace{1pt}}}

\newcommand{\FALSE}{\textsl{false}}
\newcommand{\TRUE}{\textsl{true}}

\newcommand{\iml}{\mathit{iml}}
\newcommand{\mal}{\mathit{mal}}

\newcommand{\pepa}{\mathit{pepa}}
\newcommand{\tpc}{\mathit{tpc}}

\newcommand{\finv}{f^{\mkern0mu -1}}

\newcommand{\amset}[1]{\mathscr{#1}}
\newcommand{\amsetP}{\amset{P}}
\newcommand{\amsetQ}{\amset{Q}}

\newcommand{\mycal}[1]{\mathcal{#1}}
\newcommand{\mycalP}{\mycal{P}}
\newcommand{\mycalQ}{\mycal{Q}}

\newcommand{\Edelay}{\delta}
\newcommand{\Edelaya}{\: \delta_a}
\newcommand{\Edelayb}{\: \delta_b}
\def\IMPL{\Rightarrow}
\def\IFF{\Leftrightarrow}
\def\Mtrans#1{\stackrel{#1}{\dashrightarrow}}
\def\Ttrans#1{\stackrel{#1}{\leadsto}}
\renewcommand{\Ttrans}[1]{%
  \mathrel{\vbox{\offinterlineskip\ialign{%
    \hfil##\hfil\cr
    \raisebox{0.75pt}{$\scriptstyle#1$}\cr
    $\leadsto$\cr
}}}}

\newcommand{\Rclass}[1]{[{#1}]_R}
\newcommand{\RSclass}[1]{[{#1}]_{R_{\calS}}}
\newcommand{\lftRclass}[1]{[{#1}]_{{\equiv_R}}}
\renewcommand{\lftRclass}[1]{[{#1}]_{R}}

\newcommand{\alambda}{(a,\lambda)}

\newcommand{\bnfeq}{\mathrel{\; \raisebox{0.35pt}{\textrm{::}} {=} \;}}

\newcommand{\calLi}{\calL_{\mkern2mu i}}

\newcommand{\calRi}{\calR_{\mkern2mu i}}

\newcommand{\calSiml}{\calS_\iml}

\newcommand{\calSmal}{\calS_\mal}
\newcommand{\calSpepa}{\calS_\pepa}
\newcommand{\calStpc}{\calS_\tpc}

\newcommand{\calULR}{\calU^{\calL \mkern2mu}_{\mkern-4mu \calR \mkern0mu}}

\def\chut{\mathbf{D}}
\def\chutP{\mathbf{D}_P}
\def\chutQ{\mathbf{D}_Q}
\def\cho{\, + \, }
\newcommand{\compose}{\mathop{\raisebox{0.5pt}{\scriptsize $\circ$}}}
\def\dfas{:=}

\newcommand{\fivetuple}[5]{( \mkern1mu {#1}, \mkern1mu {#2}, \mkern1mu {#3}, \mkern1mu {#4}, \mkern1mu {#5} \mkern1mu )}
\newcommand{\fmorph}[3]{[\![ {#3} ]\!]^{#2}_{#1}}

\newcommand{\four}[2]{{#1} {:}{:} {#2}}
\renewcommand{\four}[2]{{#1} {\ast} {#2}}
\renewcommand{\four}[2]{{#1} {\cdot} {#2}}
\renewcommand{\four}[2]{{#1} \mkern1mu {\cdot} \mkern1mu {#2}}

\newcommand{\fsfn}[2]{\mkern1mu \mathrm{\mathcal{F} \mkern-2.5mu \mathcal{S}}( \mkern1mu #1,#2 \mkern2mu )}

\def\fsum{\mathopen{\oplus \mkern2mu}}

\newcommand{\la}{\mathop{\langle \,}}

\newcommand{\mlbrace}{\mathopen{\lbrace \mkern-2.25mu | \mkern2.25mu}}
\newcommand{\mrbrace}{\mathclose{\mkern2mu | \! \rbrace}}
\newcommand{\mdelay}{\mathit{md} \mkern0.5mu}

\def\mtrans#1{\stackrel{\vphantom{\frac{}{}}#1}{\rightarrow}}

\newcommand{\myboxed}[1]{\mathbin{\, \lbrack \mkern-2mu {#1} \mkern-2mu \rbrack \,}}

\newcommand{\myell}{{\ell \mkern2mu}}
\newcommand{\myf}{{f \mkern2mu}}

\newcommand{\myi}{{\mkern1mu i \mkern1mu}}
\newcommand{\myin}{\, \in \,}
\newcommand{\myop}{\mathbin{\mkern1mu | \mkern1mu}}

\newcommand{\smalltimes}{\mbox{\raisebox{0.5pt}{\small $\mkern2mu \times \mkern2mu$}}}
\newcommand{\mytheta}{\theta \mkern1mu}

\def\nnreals{\reals_{\geq 0}}

\newcommand\pfun{\mathrel{\ooalign{\hfil$\mapstochar\mkern5mu$\hfil\cr$\to$\cr}}}

\newcommand{\prlA}{\mathbin{\mkern2mu \parallel_{\mkern1mu A} \mkern0mu}}

\def\pfx#1#2{#1.#2}
\def\prc#1{{\calP}_{#1}}

\newcommand{\prcIML}{\prc{\iml}}
\newcommand{\prcMAL}{\prc{\mal}}

\newcommand{\prcPEPA}{\prc{\pepa}}
\newcommand{\prcTPC}{\prc{\mkern-2mu \tpc}}

\newcommand{\ra}{\mathclose{\, \rangle}}

\def\poreals{\reals_{> 0}}

\def\sosrule#1#2{\frac{\,\,\,\,#1\,\,\,\,}{\,\,\,\,#2\,\,\,\,}}
\def\sosrule#1#2{%
  \def\arraystretch{0.50}
  \begin{array}{c} {#1}\rule{0pt}{10pt} \\ \hline
    \rule{0pt}{10pt}{#2} \end{array}
  \def\arraystretch{1.0}}
\def\sosrn#1#2{\mbox{{\normalsize (#1$_{\mbox{{\scriptsize #2}}}$)\ }}}

\newcommand{\thetai}{\theta_\myi}

\newcommand{\thetapepa}{\theta_\pepa}
\newcommand{\threetuple}[3]{( \mkern1mu {#1}, \mkern1mu {#2}, \mkern1mu {#3} \mkern1mu )}
\newcommand{\trans}[1]{\stackrel{#1}{\rightarrow}}
\renewcommand{\trans}[1]{\xrightarrow{\ #1\ }}

\newcommand{\transalambda}{\trans{a,\lambda}}
\newcommand{\tssum}{\textstyle{\sum \,}}
\newcommand{\twotuple}[2]{( \mkern1mu {#1}, \mkern1mu {#2}\mkern1mu )}

\newcommand{\varphibar}{\bar{\varphi}}

\newcommand{\zerofn}{\underline{\mathrm{0}}}
\renewcommand{\zerofn}{[ \mkern1mu ]}
\def\zerof{\zerofn}

\renewcommand{\calULR}{\calU^{\calL \mkern2mu}_{\mkern-0mu \calR
    \mkern0mu}}
\def\calVvLR{\calV^{\mkern1mu \vcalL}_{\mkern-7mu\vcalR}}
\def\twocalVvcalLR{\raisebox{-3pt}{\scriptsize 2} \mkern-4mu \calV^{{\calL}}_{\mkern-7mu\vcalR}}
\def\twocalVvcalLR{\raisebox{-3pt}{\scriptsize 2} \mkern-4mu \calW^{{\calL}}_{\mkern-7mu\vcalR}}
\def\twocalVvcalLR{\calW^{{\calL}}_{\mkern-7mu\vcalR}}
\def\vcalL{\textit{L}}
\def\vcalR{\textit{R}}

\newcommand{\SPC}{\textsl{SPC} \mkern1mu}
\newcommand{\WTS}{\textsl{WTS} \mkern1mu}

\newcommand{\pCponePonephPh}{[ \, \four{p_1}{P_1} 
  \mkern1mu \oplus \mkern1mu \cdots \mkern1mu \oplus \mkern1mu 
  \four{p_h}{P_h} \, ]}

\newcommand{\SETponePonephPh}{\SET{ \, \four{p_1}{P_1} \, 
  \Box \cdots \Box \, \four{p_h}{P_h} \,}}

\begin{document}

\title[Bisimulation of Labelled FuTS Coalgebraically]{Bisimulation of \\
  Labelled State-to-Function Transition Systems \\
  Coalgebraically}

\author[D.~Latella]{Diego Latella\rsuper a}
\address{{\lsuper{a,b}}CNR -- Istituto di Scienza e Tecnologie dell'Informazione `A. Faedo',
  Pisa}
\email{\{Diego.Latella, Mieke.Massink\}@isti.cnr.it}

\author[M.~Massink]{Mieke Massink\rsuper b}
\address{\vspace{-18 pt}}

\author[E.~P.~de Vink]{Erik P.~de Vink\rsuper c}
\address{{\lsuper c}Department of Mathematics and Computer Science,
  Eindhoven University of Technology
  \hspace*{5cm} \\
  Centrum voor Wiskunde en Informatica, Amsterdam}
\email{evink@win.tue.nl}



\begin{abstract}
  \noindent 
  Labeled state-to-function transition systems, $\FuTS$ for short, are
  characterized by transitions which relate states to
  \emph{functions} of states over general semirings, equipped with a
  rich set of higher-order operators.  As such, $\FuTS$ constitute a
  convenient modeling instrument to deal with process languages and
  their quantitative extensions in particular. In this paper, the
  notion of bisimulation induced by a $\FuTS$ is addressed from a
  coalgebraic point of view.  A correspondence result is established
  stating that $\FuTS$-bisimilarity coincides with behavioural
  equivalence of the associated functor.  As generic examples, the
  equivalences underlying substantial fragments of major examples of
  quantitative process algebras are related to the bisimilarity of
  specific~$\FuTS$. The examples range from a stochastic process
  language, $\PEPA$, to a language for Interactive Markov Chains,
  $\IML$, a (discrete) timed process language, $\TPC$, and a language
  for Markov Automata, $\MAL$.  The equivalences underlying these
  languages are related to the bisimilarity of their specific~$\FuTS$.
  By the correspondence result coalgebraic justification of the
  equivalences of these calculi is obtained. The specific selection of
  languages, besides covering a large variety of process interaction
  models and modelling choices involving quantities, allows us to show
  different classes of $\FuTS$, namely so-called simple $\FuTS$,
  combined $\FuTS$, nested $\FuTS$, and general $\FuTS$.
\end{abstract}

\maketitle



\section{Introduction}
\label{sec:introduction}

In the last couple of decades, qualitative process languages have been
enriched with quantitative information. In the qualitative case,
process languages equipped with formal operational semantics have
proven to be successful formalisms for the modeling of concurrent
systems and the analysis of their behaviour. Generally, the
operational semantics of a qualitative process language are given by
means of a labeled transition system ($\LTS$), with states being
process terms and actions decorating the transitions between
states. Typically, based on the induced transition system relation, a
notion of process equivalence is defined, providing means to compare
systems and to reduce their representation to enhance subsequent
verification.

Extensions of qualitative process languages allow a deterministic as
well as stochastic representation of time, or the use of discrete
probability distributions for resolving (some) forms of
non-determinism. Among them, languages based on stochastic modeling
of action duration or delays, usually referred to as stochastic
process algebras, or stochastic process calculi ($\SPC$), are one of
the quantitative enrichments of process languages that have received
particular attention. For $\SPC$, the main aim has been the
integration of qualitative descriptions with quantitative ones in a
single mathematical framework, building on the combination of $\LTS$
and continuous-time Markov chains ($\CTMC$). The latter is one of the
most successful approaches to modeling and performance analysis of
(computer) systems and networks. An overview of $\SPC$, equivalences
and related analysis techniques may, for example, be found
in~\cite{HHK02,BHHKS04:voss,Ber07:sfm}. A common feature of many
$\SPC$ is that actions are augmented with the rates of exponentially
distributed random variables that characterize their
\emph{duration}. Alternatively, actions are assumed to be
instantaneous, in which case random variables are used for modeling
\emph{delays}, as in~\cite{Her02:springer}. Although exploiting the
same class of distributions, the models and techniques underlying the
definition of the calculi turn out to be significantly different in
many respects. A prominent difference concerns the modeling, by
means of the choice operator, of the race condition arising from the
$\CTMC$ interpretation of process behaviour, and its relationship to the
issue of transition multiplicity. In the quantitative setting,
multiplicities can make a crucial distinction between processes that
are qualitatively equivalent. Several different approaches have been
proposed for handling transition multiplicity. The proposals range
from multi-relations~\cite{Hil96:phd,Her02:springer}, to proved
transition systems~\cite{Pri95:cj}, to $\LTS$ with numbered
transitions~\cite{GSS95:ic,HHK02}, and to unique rate
names~\cite{DLM05}, just to mention a few.

In~\cite{De+08,De+09,De+09a}, Latella, Massink et al.\ proposed a variant of
$\LTS$, called Rate Transition Systems ($\RTS$). In~$\LTS$, a
transition is a triple $( P, \alpha ,P' \mkern1mu )$ where $P$
and~$\alpha$ are the source state and the label of the transition,
respectively, while $P'$ is the target state reached from~$P$ via a
transition labeled with~$\alpha$. In~$\RTS$, a transition is a triple
of the form $(P,\alpha,\amsetP \mkern2mu )$. The first and second
component are the source state and the label of the transition, as
in~$\LTS$, while the third component~$\amsetP$ is a \emph{continuation
  function} (or simply a {\em continuation} in the sequel), which
associates a non-negative real value with each state~$P'$. A non-zero
value for state~$P'$ represents the rate of the exponential
distribution characterizing the average time for the execution of the
action represented by~$\alpha$, necessary to reach $P'$ from~$P$ via
the transition. If~$\amsetP$~maps $P'$ to~$0$, then state~$P'$ cannot
be reached from~$P$ via this transition. The use of continuations
provides a clean and simple solution to the transition multiplicity
problem and make~$\RTS$ particularly suited for $\SPC$ semantics. In
order to provide a uniform account of the many $\SPC$ proposed in the
literature, in previous joint work of the first two authors,
see~\cite{De+14}, State-to-Function Labeled Transition Systems
($\FuTS$) have been introduced as a natural generalization
of~$\RTS$. In~$\FuTS$ the codomain of the continuations are arbitrary
semirings, rather than just the non-negative reals. This provides
increased flexibility while preserving basic properties of primitive
operations like sum and multiplication. Furthermore, $\FuTS$ are
equipped with a rich set of (generic) operations on continuation
functions, which makes the framework very well suited for a
\emph{compositional} definition of the operational semantics of
process calculi, including $\SPC$ and models where both
non-deterministic behaviour and stochastic delays are model led, like
in the Language of Interactive Markov Chains~\cite{Her02:springer}, or
even in combination with probabilistic distributions over behaviours,
as in languages for Markov Automata~\cite{Ti+12}, besides calculi for
deterministic (discrete) timed systems~\cite{ABC10:springer}.

In this paper we extend the work presented in~\cite{De+14} in two
directions.  The first contribution concerns the extension of the
$\FuTS$ framework by introducing the notions of combined $\FuTS$ and
nested $\FuTS$.  Given label sets~$\calLi$ and semirings~$\calRi$, a
combined~$\FuTS$ takes the general format $\calS = ( \, S ,\, \la
\mkern-2mu {\mtrans{}_i} \ra^n_{i = 1} \, )$ with transition relations
${\mtrans{}_i} \, \subseteq \, S \times \calLi \times \fsfn{\mkern2mu
  S}{\calRi}$.  In the degenerated case of $n=1$, we speak of a simple
$\FuTS$, which coincides with the definition of $\FuTS$ proposed
in~\cite{De+14}. Here, $\fsfn{\mkern2mu S}{\calRi}$ is the set of
total functions from~$S$ to~$\calRi$ with finite support, a
sub-collection of functions also occurring in other work combining
coalgebra and quantitative modeling (see,
e.g.~\cite{KS08:fossacs,BBBRS12}). So, a combined $\FuTS$ is
characterized by the presence of \emph{multiple} transition relations
which allow for a clean definition of the $\FuTS$ semantics of
languages which integrate different aspects of behaviour, such as
non-determinism vs.\ stochastic time, as is characteristic for
Interactive Markov Chains.  Using a single transition relation in such
a setting requires additional proof obligations ensuring type
correctness of transition elements, in particular the continuations,
as can be seen in~\cite{De+14}, for example. Instead, for combined $\FuTS$
this is ensured by construction. The general format of a so-called
nested $\FuTS$ over the label set~$\calL$ and semirings $\calR_1,
\ldots \calR_n$, for $n >1$, is a tuple $\calS = ( \, S ,\, \mtrans{}
\, )$ with ${\mtrans{}} \, \subseteq \, {S \times \calL \times
  \fsfn{\, (\ldots \fsfn{\mkern2mu S}{\calR_1} \ldots) \,}{\calR_n}}$.
For the purposes of the present paper, $n=2$ suffices; the nested
$\FuTS$ we consider here are of the form $\calS = ( \, S ,\, \mtrans{}
\, )$ with ${\mtrans{}} \, \subseteq \, {S \times \calL \times
  \fsfn{\fsfn{\mkern2mu S}{\calR_1}}{\calR_2}}$. For nested $\FuTS$
the transition relation relates \emph{functions over states}, instead
of just states, to continuations.  This makes it easy, for instance,
to represent non-deterministic choices between probabilistic
distributions over behaviours, as it is the case for (the non-timed
fragment of languages for) Markov Automata.
Finally, product construction for combined $\FuTS$ and sequencing
construction for nested $\FuTS$) can easily be combined giving rise to
what one may call \emph{general} $\FuTS$ (or just $\FuTS$, in the
sequel), which prove useful for a concise definition of the
operational semantics of Markov Automata languages.

We will briefly show how the various types of $\FuTS$ can be used
conveniently for a clean and compact definition of the fragments of
interest of major process languages (more details on this can be found
in~\cite{De+14}, which the interested reader is referred to).  For
combined $\FuTS$, as well as nested $\FuTS$ and general $\FuTS$, we
also present $\FuTS$ bisimilarity, a general notion of bisimilarity,
which will also be shown to coincide with the standard bisimilarity of
the relevant process languages.

The second direction of investigation presented in this paper consists
of a coalgebraic treatment of the various type of~$\FuTS$.  We will
see that a combined $\FuTS$ $( \, S ,\, \la \mkern-2mu {\mtrans{}_i}
\ra^n_{i = 1} \, )$ is a coalgebra of the product of the functors
$\fsfn{ {\cdot}}{\calRi}^{\mkern1mu \calLi}$. For this to work, we
need the relations~$\mtrans{}_i$ to be total and deterministic for the
coalgebraic modeling as a function. This is not a severe restriction
at all in the presence of continuation functions: as we will see, the
zero-continuation function, which maps every state~$s'$ to~$0$ will be
associated to a state~$s$ and a transition, in order to indicate that
no state~$s'$ is reachable from~$s$ via that transition, in the usual
$\LTS$-sense; if $s$~allows a transition to some state~$s_1$ as well
as to a state~$s_2$, then the continuation function will simply yield
a non-zero value for $s_1$ and for~$s_2$. Therefore, it is no
essential limitation to restrict our investigations to total and
deterministic $\FuTS$. For example, by using Boolean functions, we can
model non-deterministic behaviour, as done in Section~\ref{sec-iml}.
Similarly, we see that a (two-level) nested $\FuTS$ $(\, S, \mtrans{}
\,)$ is a coalgebra of functor
$\fsfn{{\fsfn{{\cdot}}{\calR_1}}}{\calR_2}^{\mkern1mu \calL}$.

Next, the notion of $\calS$-bisimilarity that arises from
a~$\FuTS$~$\calS$ is compared to the coalgebraic notion of behavioral
equivalence. Following a familiar argument, we first establish that
the functor associated with a~$\FuTS$ $\calS$ possesses a final
coalgebra and therefore has an associated notion of behavioural
equivalence. Then it is shown that behavioural equivalence of the
functor coincides with $\calS$-bisimilarity, bisimilarity for
$\FuTS$. Pivotal for the proof is the absence of multiplicities in the
$\FuTS$ treatment of quantities at the level of the transitions. In
fact, quantities are accumulated in the function values of the
continuations and hidden at the higher level of abstraction. It is
noted, in the presence of a final coalgebra for $\FuTS$ a more general
definition of behavioural equivalence based on cospans coincides with
the one given here, cf.~\cite{Kur00:phd,Sta11:lmcs}.
Finally, it is worth noting that for the coalgebraic treatment itself
of $\FuTS$ we propose here it is not necessary for the co-domain of
continuations to be semirings; working with monoids would be
sufficient. However, the richer structure of semirings is convenient,
if not essential, when using continuations and their operators in the
formal definition of the $\FuTS$ semantics of~$\SPC$.

Using the bridge established by the $\FuTS$ bisimulation vs.\
coalgebraic behavioral equivalence correspondence results, we continue
by showing for two well-known stochastic process algebras, viz.\
Hillston's $\PEPA$~\cite{Hil96:phd} and Hermanns's
$\IML$~\cite{Her02:springer}, that the standard notions of $\PEPA$
strong equivalence and $\IML$ strong bisimilarity coincide with
bisimilarity of the associated proper simple and combined $\FuTS$,
respectively. In turn, this means that the standard notions of strong
equivalence and strong bisimilarity coincide with behavioural
equivalence when cast in a coalgebraic framework.

$\PEPA$ stands out as one of the prominent Markovian process algebras,
and $\IML$ specifically provides separate prefix constructions for
actions and for delays. 
In passing, the issue of transition multiplicity has to be dealt with.
Appropriate lemmas are provided relating the relation-based cumulative
treatment with $\FuTS$ to the multi-relation-based explicit treatment
of $\PEPA$ and~$\IML$. It is noted that in our treatment below we
restrict to the key-fragment of these two~$\SPC$. 
We furthermore provide a combined $\FuTS$ semantics for a simple
language of deterministically-timed processes, viz.\
$\TPC$~\cite{ABC10:springer} and we show the coincidence between
$\FuTS$ bisimilarity and the standard equivalence of timed
bisimilarity for the language.  
Finally, we provide a general $\FuTS$ semantics for a process language
which incorporates non-determinism, discrete probabilities and
Markovian randomized delays, i.e.\ a language for Markov
Automata~\cite{EHZ10:concur,EHZ10:lics}. Also in this case we prove
that $\FuTS$ bisimulation and Markov Automata bisimulation coincide,
adding to the claim that $\FuTS$ bisimulation is a natural notion of
process identification for~$\SPC$.

\blankline

\noindent
Related work on coalgebra includes the papers
\cite{DeR99,KS08:fossacs,Sok11:tcs}. Additionally, these papers cover
measures and congruence formats, a topic not touched upon in the
present paper. For what concerns the discrete parts, regarding the
correspondence of bisimulations, our work is in line with the approach
of the papers mentioned. In the treatment below the bi-algebraic
perspective of SOS and bisimulation~\cite{TP97:lics} is left implicit.
In~\cite{MiP14} an approach similar to ours has been applied to the
ULTraS model, a model which shares some features with simple
$\FuTS$. In ULTraS posets are used instead of semirings, although a
monoidal structure is then implicitly assumed when process
equivalences are taken into consideration~\cite{Be+13b}.  Furthermore,
in~\cite{MiP14} a general GSOS specification format is presented which
allows for a `syntactic' treatment of continuations involving
so-called \emph{weight functions}.
An interesting direction of research combining coalgebra and
quantities studies various types of weighted automata, including
linear weighted automata, and associated notions of bisimulation and
languages, as well as algorithms for these
notions~\cite{Bor09:concur,Kli09:mosses,SBBR10:ic,BBBRS12}. Klin
considers weighted transition systems, labelled transition systems
that assign a weight to each transition and develops Weighted GSOS, a
(meta-)syntactic framework for defining well-behaved weighted
transition systems. For commutative monoids the notion of a weighted
transition system compares with our notion of a~$\FuTS$, and, when
cast in the coalgebraic setting, the associated concept of
bisimulation coincides with behavioral equivalence.
Weights of transitions of weighted transition systems are computed by
induction on the syntax of process terms and by taking into account
the contribution of all those GSOS rules that are triggered by the
relevant (apparent) weights. Note that such a set of rules is finite.
So, in a sense, the computation of the weights is {\em distributed among
(the instantiations of) the relevant rules} with intermediate results
collected (and integrated) in the final weight. 
In comparison, as mentioned before, 
in the~$\FuTS$ approach, the relevant values are manipulated in a more direct way,
using the higher-order operators on continuation functions, applying them directly
to the continuations in the transitions {\em within the same}  the semantics definition rules.
So, in a sense, the~$\FuTS$ approach is better suited
for a \emph{compositional} definition of the operational semantics of
a wide range of process calculi due to the suitable choice of a rich
set of generic operations on continuation functions.
In~\cite{LMV13} the investigation on the relationship for nested
$\FuTS$ between $\FuTS$ bisimilarity, and behavioural equivalence, and
also coalgebraic bisimilarity is presented. In particular, it is shown
that the functor type involved preserves weak pullbacks when the
underlying semiring satisfies the zero-sum property.

The process languages with stochastic delays we consider in the
sequel, involve a multi-way CSP-like parallel operator; components
proceed simultaneously when synchronization on an action from the
synchronization alphabet that indexes the parallel operator is
possible. However, here we do not distinguish between internal and
external non-determinism, cf.~\cite{Hoa85:ph}, since an explicit
representation of such a distinction is not relevant for the subject
of this paper. A coalgebraic treatment of this distinction is proposed
in~\cite{Wol02:tcs}, which uses a functor for so-called
non-deterministic filter automata, viz.\ $\calP(\calP(\calA)) \times [
  \, \calA \pfun \calP_{\!\!f} \mkern1mu ( {\cdot} ) \,]$ involving
partial functions from a set of actions~$\calA$ to a finite
power-set. Via currying, this can be brought into the form $\fsfn{
  {\cdot} }{\bools}^{\mkern1mu \calL}$ for $\calL =
\calP(\calP(\calA)) \times \calA$, fitting the format of the functor
for the (simple) $\FuTS$ considered here. In~\cite{BG06:tcs} processes
are interpreted as formal power-series over a semiring in the style
of~\cite{Rut03:tcs}. This allows to compare testing equivalence for a
CSP-style language and bisimulation in a Moore automaton. Note that
the notions of equivalence addressed in this paper, as often in
coalgebraic treatments of process relations, are all strong
bisimilarities.

An extended abstract of part of this paper has appeared
as~\cite{LMV12:accat} where the coalgebraic view of the $\FuTS$
approach and its application to~$\PEPA$ and~$\IML$ was originally
presented. The workshop contribution~\cite{LMV15:qapl} gives an
account of bisimulation of $\FuTS$ of specific type and provides a
general correspondence result with of $\FuTS$-bisimulation and
behavioral equivalence. The present paper covers these ideas in a
structured way, going from simple $\FuTS$ to combined $\FuTS$ and
nested $\FuTS$. It includes the presentation of the use of combined
$\FuTS$ for the definition of the semantics of a language of
deterministically timed processes and the treatment of nested $\FuTS$
for the integration of stochastically timed, non-deterministic and
probabilistic processes, as in Markov Automata.

\blankline

\noindent
For the present paper we assume the reader to have some familiarity
with $\SPC$ and the application of $\FuTS$ for the definition of their
semantics. The reader is referred to~\cite{De+14} for an introduction
on the subject. Furthermore, in~\cite{LMV13} an illustrative
definition of a simple, qualitative, process calculus in the $\FuTS$
framework is shown.
Section~\ref{sec-preliminaries} provides basic concepts and
notation. Simple $\FuTS$ are introduced in Section~\ref{sec-futs},
followed by their coalgebraic treatment in
Section~\ref{sec-coalgebra}. Simple $\FuTS$ are illustrated by the
case of $\PEPA$ in Section~\ref{sec-pepa} which also covers the
correspondence of the respective notions of
bisimulation. Section~\ref{sec-combined} introduces combined $\FuTS$
as well as their coalgebraic representation. Sections \ref{sec-iml}
and~\ref{sec-tpc} treat $\IML$ and~$\TPC$. For both $\SPC$, semantics
based on combined $\FuTS$ are  given, and $\FuTS$ bisimulation is
compared to standard bisimulation. Next, Section~\ref{sec-nested}
introduces nested as well as general $\FuTS$, again tying up with
behavior equivalence. In Section~\ref{sec-mal}, a general $\FuTS$ is
used for the semantics of a Markov Automata language, for which the
notion of bisimulation is related to the standard
one. Section~\ref{sec-conclusions} wraps up and discusses closing
remarks.


\section{Preliminaries}
\label{sec-preliminaries}

\noindent
A tuple $\calR = \fivetuple{R}{+}{0}{\ast}{1}$ is called a semiring if
$\threetuple{R}{+}{0}$ is a commutative monoid with neutral
element~$0$, $\threetuple{R}{\ast}{1}$ is a monoid with neutral
element~$1$, $\ast$~distributes over~$+$, and $0 \ast r = r \ast 0 =
0$ for all~$r \in R$. 
As examples of a semiring we will use the
Booleans $\bools = \SET{ \, \FALSE ,\, \TRUE \, }$ with disjunction as
sum and conjunction as multiplication, the non-negative
reals~$\nnreals$ with the standard operations, and the powerset
construct~$\textbf{2}^X$ for a set~$X$ with intersection and union as
sum and multiplication, respectively.  
We will consider, for a
semiring~$\calR$ and a function $\varphi : X \to \calR$, (countable)
sums $\tssum_{x \myin X'} \; \varphi(x)$ in~$\calR$, for~$X' \subseteq
X$. For such a sum to exist we require~$\varphi$ to be of finite
support, i.e.\ the support set $\spt(\varphi) = \ZSET{ x \in X }{
  \varphi(x) \neq 0 }$ is finite. We use the notation $\fsum \varphi$
to denote the value $\tssum_{x \in X} \; {\varphi(x)}$ in~$\calR$.
    
We use the notation $\fsfn{X}{\calR}$ for the collection of all
functions of finite support from the set~$X$ to the semiring~$\calR$.
A construct $[ \, x_1 \mapsto r_1 ,\, \ldots ,\, x_n \mapsto r_n \,]$,
or more compactly $[ \, x_i \mapsto r_i ]_{i=1}^n$, with $x_i \in X$
all distinct and $r_i \in \calR$, denotes the mapping that assigns
$r_i$ to~$x_i$, $i = 1, \ldots, n$, and assigns $0$ to all $x \in X
\setminus \SET{x_1, \ldots, x_n}$.  In particular~$\zerof \mkern1mu$,
or more precisely~$\zerof_{\mkern1mu \calR}$, is the constant
function~$x \mapsto 0$ and $\chut^{\mkern1mu \calR}_{x} = [ \, x
  \mapsto 1 \, ]$ is the Dirac function on~$\calR$ for $x \in X$; in
the sequel we will often drop the subscript or superscript~$\calR$
from $\zerof_{\mkern1mu\calR}$ and~$\chut^{\mkern1mu \calR}_{x}$, when
the semiring~$\calR$ is clear from the context.

For $\varphi, \psi \in \fsfn{X}{\calR}$, the function $\varphi \cho
\psi$ is the pointwise sum of $\varphi$ and~$\psi$, i.e.\ $(\varphi
\cho \psi)(x) = \varphi(x) \cho \psi(x) \in \calR$. Clearly, $\varphi
\cho \psi$ is of finite support as are $\varphi$ and~$\psi$. Slightly
more generally, for functions $\varphi_i \in \fsfn{X}{\calR}$ where $i
= 1 , \ldots , n$, we define the function $\tssum_{i=1}^n \:
\varphi_i$ in $\fsfn{X}{\calR}$ by $\big( \mkern1mu \tssum_{i=1}^n \:
\varphi_i \mkern1mu \big)(x) = \tssum_{i=1}^n \: \varphi_i \mkern1mu
(x)$. Given an injective operation~${\myop} \colon X \times X \to X$,
we define $\varphi \myop \psi : X \to \calR$, by $( \varphi \myop \psi
)( x ) = \varphi(x_1) \ast \psi(x_2)$ if $x = x_1 \myop x_2$ for some
$x_1, x_2 \in X$, and $( \varphi \myop \psi )( x ) = 0$
otherwise. Injectivity of the operation~$\myop$ guarantees that
$\varphi \myop \psi$ is well-defined. Again, $\varphi \myop \psi$ is
of finite support as are $\varphi$ and~$\psi$. Such an operation is
used in the setting of syntactic processes~$P$ that may have the form
$P_1 \myop P_2$ for two processes $P_1$ and~$P_2$ and a syntactic
operator~$\myop$.  



 
\blankline

\noindent
We recall some basic definitions from coalgebra.  See
e.g.~\cite{Rut00:tcs} for more details.  For a functor $\calF : \Set
\to \Set$ on the category $\Set$ of sets and functions, a
coalgebra~$\calX$ of~$\calF$ is a set~$X$ together with a mapping
$\alpha : X \to \calF(X)$.  A homomorphism between two
$\calF$-coalgebras $\calX = (X,\alpha)$ and~$\calY = (Y,\beta)$ is a
function $f : X \to Y$ such that $\calF(f) \compose \alpha = \beta
\compose f$. An $\calF$-coalgebra $(\Omega_{\calF},\omega_{\calF})$ is
called final or terminal, if there exists, for every $\calF$-coalgebra
$\calX=(X,\alpha)$, a unique homomorphism
$\fmorph{\calF}{\calX}{\cdot} : (X,\alpha) \to
(\Omega_{\calF},\omega_{\calF})$.  Two elements~$x_1, x_2$ of an
$\calF$-coalgebra~$\calX$ are called behavioural equivalent with
respect to~$\calF$ if $\fmorph{\calF}{\calX}{x_1} =
\fmorph{\calF}{\calX}{x_2}$, denoted $x_1 \beq{\calF}{\calS} x_2$.  In
the notation $\fmorph{\calF}{\calX}{\cdot}$ as well as
$\beq{\calF}{\calX}$, the indication of the specific coalgebra~$\calX$
will be omitted when clear from the context.

A functor~$\calF$ is called accessible if it preserves
$\kappa$-filtered colimits for some cardinal number~$\kappa$. However,
in the category~$\Set$, we have the following characterization of
accessibility: for every set~$X$ and any element $\xi \in \calF X$,
there exists a subset $Y \subseteq X$ with ${|} Y {|} < \kappa$, such
that $\xi \in \calF Y$. It holds that a functor has a final coalgebra
if it is $\kappa$-accessible for some cardinal
number~$\kappa$. See~\cite{AP04:tcs,AMM10:esslli}.

\blankline

\noindent
A number of proofs of results on process languages~$\prc{}$ in this
paper rely on so-called guarded induction~\cite{KR90:tcs}. Typically,
constants~$X$, also called process variables, are a syntactical
ingredient in these languages. As usual, if $X \dfas P$, i.e.\ the
constant~$X$ is declared to have the process~$P$ as its body, we
require $P$ to be prefix-guarded, i.e. any occurrence of a constant in
the body~$P$ is in the scope of a prefix-construct of the language.
Guarded induction assumes the existence of a `complexity' function $c
: \prc{} \to \nats$ such that $c(P) = 1$ if $P$~is a prefix construct,
$c(P_1 \mathbin{\raisebox{0.5pt}{\fontsize{9}{0}{$\bullet$}}} P_2) >
\textrm{max} \SET{ \, c(P_1) ,\, c(P_2) \,}$ for all other syntactic
operators~\raisebox{0.5pt}{\fontsize{9}{0}{$\bullet$}} of~$\prc{}$,
and, in particular, $c(X) > c(P)$ if $X \dfas P$. For all concrete
process languages treated in this paper such a complexity function can
be given straightforwardly. See~\cite{BV96:mit} for more detail.

  For convenience we collect here a number of abbreviations used in
  the sequel: $\CTMC$ and $\DTMC$ for the standard notions of
  Continuous-Time Markov Chains and Discrete-Time Markov Chains,
  respectively; $\LTS$ for Labelled Transition System, 
  $\RTS$ for Rate Transition System, and $\FuTS$ for Labelled
  State-to-Function Transition System, the extension of $\LTS$ we
  focus on in this paper; $\SPC$ for Stochastic Process Calculus,
  referring to the class of process algebras featuring a choice
  construct based on a non-negative exponential distribution; for
  specific process calculi and semantic models, viz.\ $\PEPA$ for
  Performance Evaluation Process Algebra, $\IMC$ for Interactive
  Markov Chains and $\IML$ for the $\IMC$-based language used in this paper, $\TPC$
  for an example Timed Process Calculus, $\MA$ for Markov Automata
  and $\MAL$ for the $\MA$-based language used in this paper. 


\section{Simple State-to-Function Labelled Transition Systems}
\label{sec-futs}

Below we introduce \emph{simple} $\FuTS$,
i.e.\ $\FuTS$ with a single transition relation, which are sufficient
for the definition of the semantics of many of the relevant stochastic
process languages proposed in the literature (see~\cite{De+14} for
details). 

\begin{defi}
  \label{df-ltfs}
  A \emph{simple} $\FuTS$~$\calS$, in full `a simple state-to-function
  labelled transition system', over label set~$\calL$ and
  semiring~$\calR$, is a tuple $\calS = ( \, S ,\, \ \mkern-2mu
  {\mtrans{}} \, )$ where ${\mtrans{}}$ is the transition relation,
  with ${\mtrans{}} \, \subseteq \, S \times \calL
  \times \fsfn{\mkern2mu S}{\calR}$.
\end{defi}

\noindent
In the sequel we omit the word `simple' when this cannot give rise to
confusion. Similar as for state-to-state transitions of $\LTS$, for
state-to-function transitions of $\FuTS$ we write $s \mtrans{\ell} v$
for $(s,\ell,v) \in {\mtrans{}}$.  
Note that a (simple) $\FuTS$ over a label set~$\calL$ and a
semiring~$\calR$ is reminiscent of a weighted
automaton~\cite{DKV09}. However, for a FuTS no output function is
given, as is for a weighted automaton. To stress the relationship
between $\LTS$ and~$\FuTS$ we stick to the terminology and notion
stemming from~$\LTS$.

For a $\FuTS$~$\calS = ( \, S ,
\mkern-2mu {\mtrans{}} \, )$ the set~$S$ is called the set of states
or the carrier set. We refer to $\mtrans{}$ as the state-to-function
transition relation of~$\calS$ or just as the transition relation.  A
$\FuTS$~$\calS$ is called total and deterministic if, for all $s \in
S$ and $\ell \in \calL \mkern1mu$, we have $s \mtrans{\ell} v$ for
exactly one $v \in \fsfn{S}{\calR}$. In such a situation, the
state-to-function relation~$\mtrans{}$ corresponds to a function $S
\to \calL \to \fsfn{S}{\calR}$.  For the remainder of the paper, all
$\FuTS$ we consider will be total and deterministic. It is noted that
Definition~\ref{df-ltfs} slightly differs in formulation from the one
provided in~\cite{De+14}.


As an example, Figure~\ref{fig-examples} displays a simple $\FuTS$
over the action set~$\calA$ and the semiring~$\nnreals$ of the
non-negative real numbers with standard sum and multiplication.  The functions $v_0$ to~$v_3$ used in the
example have the property that $\fsum v_\myi(s) \: = \: 1$, for $i =
0, \ldots , 3$.  More explicitly, we have
\begin{displaymath}
  \def\arraystretch{1.5}
  \begin{array}{c}
    s_0 \mtrans{a} [ s_0 \mapsto \frac12 ,\, s_1 \mapsto \frac12 ]
    \qquad
    s_2 \mtrans{a} [ s_2 \mapsto \frac12 ,\, s_3 \mapsto \frac12 ]
    \qquad
    s_3 \mtrans{a} [ s_0 \mapsto \frac12 ,\, s_3 \mapsto \frac12 ]
    \\
    s_1 \mtrans{a} [ s_1 \mapsto \frac12 ,\, s_2 \mapsto \frac12 ]
    \qquad
    s_1 \mtrans{b} [ s_0 \mapsto \frac16 ,\, s_2 \mapsto \frac12 ,\,
    s_3 \mapsto \frac13 ]
    \\
    \text{
      $s_i \mtrans{b}
      \zerofn_{\bools} \; $ for $i = 0, 2, 3$}
    \qquad
  \end{array}
  \def\arraystretch{1.0}
\end{displaymath}
Usually, such a $\FuTS$ over~$\nnreals$, with its weights adding up
to~$1$, is called a (reactive) probabilistic transition
system~\cite{GSS95:ic}. 

  

%
\begin{figure}
  \centering
  \scalebox{0.90}{%
  \raisebox{0.3in}{\includegraphics[scale=0.50]{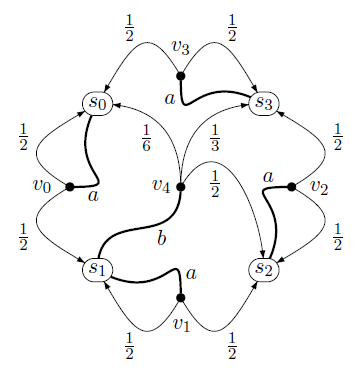}}
  } 
  \halflineup
  \halflineup
  \caption{Simple $\FuTS$ for 
  a probabilistic process.}
  \label{fig-examples}
\end{figure}


\blankline

\noindent
Below it will be notationally convenient to consider a (total,
deterministic and simple) $\FuTS$ as a tuple $( \, S ,\, \mkern-2mu
{\theta} \, )$ with transition function $\theta : S \to \calL \to
\fsfn{\mkern2mu S}{\calR}$, rather than using the form $( \, S ,\,
\mkern-2mu {\mtrans{}} \, )$ that occurs more frequently for concrete
examples in the literature. We will use the notation with transition
functions $\theta : S \to \calL \to \fsfn{\mkern2mu S}{\calR}$ to
introduce the notion of bisimilarity for a simple $\FuTS$.

\blankline

\begin{defi}
  \label{df-ltfs-bisim}
  Let $\calS = ( \, S ,\, \mkern-2mu {\theta} \, )$ be a simple
  $\FuTS$ over label set $\calL$ and semiring~$\calR$. An equivalence
  relation $R \subseteq S \times S$ is called an $\calS$-bisimulation
  if $R(s_1,s_2)$ implies
  \begin{equation}
    \tssum_{t' \in \Rclass{t}} \; \theta \mkern2mu (s_1)(\ell
    \mkern2mu )(t' \mkern1mu ) 
    =
    \tssum_{t' \in \Rclass{t}} \; \theta \mkern2mu (s_2)(\ell
    \mkern2mu )(t' \mkern1mu ) 
    \label{eq-ltfs-bisim} 
  \end{equation}%
  for all $t \in S$ and $\ell \in \calL$, where we use the notation
  $\Rclass{t}$ to denote the equivalence class of~$t \in S$ with
  respect to~$R$. Two elements $s_1, s_2 \in S$ are called
  $\calS$-bisimilar if $R(s_1,s_2)$ for some $\calS$-bisimulation~$R$
  for~$\calS$. Notation $x_1 \FuTSbis{\calS} x_2$.
\end{defi}

\noindent
Note that the sums in equation~(\ref{eq-ltfs-bisim}) exist since the
functions $\theta \mkern2mu (s_1)(\ell \mkern2mu ), \theta \mkern2mu
(s_2)(\ell \mkern2mu ) \in \fsfn{\mkern2mu S}{\calR}$ are of finite
support.

\section{Simple \texorpdfstring{$\FuTS$}{FuTS} coalgebraically}
\label{sec-coalgebra}

In this section we will cast \emph{simple} $\FuTS$ in the framework of
coalgebras and prove a correspondence result of $\FuTS$ bisimilarity
and behavioural equivalence for functors of the form
$\fsfn{{\cdot}}{\calR}^{\mkern1mu \calL}$ on~$\Set$, with~$\calR$ a
semiring and $\calL$ a set of labels. 
%
\begin{defi}
  \label{df-v-functor}
  Let $\calL$ be a set of labels and $\calR$ a semiring.  Functor
  $\calULR : \Set \to \Set$ assigns to a set~$X$ the function space
  $\fsfn{X}{\calR}^{\mkern1mu \calL}$ of all functions $\varphi :
  \calL \to \fsfn{X}{\calR}$ and assigns to a mapping $f : X \to Y$
  the mapping $\calULR(\myf) : \fsfn{X}{\calR}^{\mkern1mu \calL} \to
  \fsfn{Y}{\calR}^{\mkern1mu \calL}$ where
  \begin{displaymath}
    \calULR(\myf)(\varphi)(\myell)(y)
    =
    \tssum_{x \myin \finv(y)} \  \varphi(\myell)(x)
  \end{displaymath}
  for all $\varphi \in \fsfn{X}{\calR}^{\mkern1mu \calL}$, $\ell \in
  \calL$ and~$y \in Y$.
\end{defi}

\noindent
Working in the context of $\FuTS$ we include the label set~$\calL$ in
the notation for the functor~$\calULR$. The functor
$\fsfn{{\cdot}}{\calX}$ itself, for $\calX$ not necessarily a
semiring, but a commutative monoid or field instead, have been studied
frequently in the literature, see
e.g.~\cite{GS01:entcs,Kli09:mosses,BBBRS12}.

Again we rely on~$\varphi(\myell) \in \fsfn{X}{\calR}$ having a finite
support for the sum to exist and for $\calULR$ to be well-defined. We
observe that for any simple $\FuTS$ $(S, \theta)$ over $\calL$ and
$\calR$ we have $\theta : S \to \calL \to \fsfn{S}{\calR}$. Thus
$(S,\theta)$ can be interpreted as a $\calULR$-coalgebra. In the
sequel, we will abbreviate $\calULR$ with $\calU$ whenever $\calL$ and
$\calR$ are clear from the context.
  
\blankline

\noindent
As we aim at comparing our notion of bisimilarity for simple~$\FuTS$
with behavioural equivalence for the functor~$\calULR$, $\calU$~for
short, given a set of labels~$\calL$ and a semiring~$\calR$, we need
to check that $\calU$ possesses a final coalgebra.
For this, one may adapt the proof for the functor~$\fsfn{\cdot}{\calM}
: \Set \to \Set$ where $\calM$~is a monoid (rather than a semiring) as
sketched in~\cite{Sil10:phd,SBBR10:ic} to the setting here. 
An alternative route to showing the existence of a final coalgebra is
to verify accessibility. We directly apply the results
of~\cite[Section~5]{AP04:tcs}.

\begin{lem}
  \label{lm-v-is-bounded}
  For a set of labels~$\calL$ and a semiring~$\calR$, the
  functor~$\calU$ has a final coalgebra.
\end{lem}

\begin{proof}
  It suffices to show that the $\Set$-functor $\calU$ is accessible
  for some suitable cardinal number.  In fact, $\calU$ is ${|} \calL
  \mkern1mu {|} \mkern2mu {\times} \mkern2mu \omega \mkern2mu
  $-accessible: Consider $\varphi : \calL \to \fsfn{X}{\calR}$ in the
  image of the set~$X$. Let $Y_\ell \subseteq X$ be the support
  of~$\varphi \mkern1mu (\ell) \in \fsfn{X}{\calR}$ and $Y =
  \bigcup_{\ell \in \calL} \: Y_\ell \subseteq X$. Then $\varphi$ can
  be seen as an element of $\calL \to \fsfn{Y}{\calR}$, since outside
  of~$Y$ it holds that $\varphi$ equals~$0 \in \calR$.
\end{proof}

\noindent
Next we establish, for a given simple $\FuTS$~$\calS$, the
correspondence of $\calS$-bisimulation as given by
Definition~\ref{df-ltfs-bisim} and behavioural equivalence induced
by~$\calU$. The proof is similar to~\cite[Theorem~1]{BBBRS12}.
  
\begin{thm}
  \label{th-correspondence}
  Let $\calS = ( \, S ,\, \mkern-2mu {\theta} \mkern1mu)$ be a simple
  $\FuTS$ over the label set $\calL$ and semiring~$\calR \mkern1mu$,
  and $\calU$ as in Definition~\ref{df-v-functor}.  Then $s_1
  \FuTSbis{\calS} s_2 \IFF s_1 \beq{\calU}{} s_2$, for all $s_1, s_2
  \in S$.
\end{thm}

\begin{proof}
  Let $s_1, s_2 \in S$.  We first prove ${s_1 \FuTSbis{\calS} s_2} \,
  \IMPL \, {s_1 \beq{\calU}{} s_2}$.  So, assume $s_1 \FuTSbis{\calS}
  s_2$.  Let $R \subseteq S \times S$ be an $\calS$-bisimulation with
  $R(s_1,s_2)$.   Note $\twotuple{S}{\mytheta}$ is a
  $\calU$-coalgebra.  We turn the collection of equivalence
  classes~$S/R$ into a $\calU$-coalgebra $\calS_R=(S/R,\varrho_R)$ 
  where
  \begin{displaymath}
  \varrho_R ( \, \Rclass{s} \, )(\myell)( \,
    \Rclass{t} \, ) = \textstyle{\sum}_{t' \myin \Rclass{t}} \;
    \theta(s)(\myell)(t')
  \end{displaymath}
  for $s, t \in S$, and $\ell \in \calL$.  This is
  well-defined since $R$ is an $\calS$-bisimulation: if $R(s,s')$ then
  we have $\textstyle{\sum}_{t' \in \Rclass{t}} \;
  \theta(s)(\myell)(t') = \textstyle{\sum}_{t' \in \Rclass{t}} \;
  \theta(s')(\myell)(t')$.  The canonical mapping $\varepsilon_R
  : S \to S/R$ is a $\calU$-homomorphism: For $\ell
  \in \calL$ and $t \in S$, we have
  \begin{displaymath}
    \def\arraystretch{1.2}
    \begin{array}{rcll}
      \multicolumn{4}{l}{ \calU \mkern1mu
        (\varepsilon_R)( \, \theta(s)\,)(\myell)(\Rclass{t})}
      \\ & = &
      \tssum_{t' \myin \varepsilon^{-1}_R( \Rclass{t} )} \; \theta( s )( \myell )(t' )
      & \text{by definition of~$\calU$}
      \\ & = &
      \tssum_{t' \myin \Rclass{t}} \; \theta( s )( \myell )(t' )
      & \text{by definition of~$\varepsilon_R$}
      \\ & = &
       \varrho_R \mkern1mu ( \, \Rclass{s} \, )(\myell)( \,\Rclass{t} \, )
      & \text{by definition of~$\varrho_R$}
      \\ & = &
       \varrho_R \mkern1mu ( \, \varepsilon_R(s) \, )(\myell)(
       \,\Rclass{t} \, ) 
      & \text{by definition of~$\varepsilon_R$}
    \end{array}
  \end{displaymath}
  Thus, $\calU(\varepsilon_R) \compose \theta = \varrho
  \compose \varepsilon_R$, i.e.  
  $\varepsilon_R$ is a $\calU$-homo\-morphism. Therefore, by
  uniqueness of a final morphism, we have
  $\fmorph{\calU}{\calS}{\cdot} = \fmorph{\calU}{{\calS_R}}{\cdot}
  \compose \, \varepsilon_R$.  In particular, with respect to~$\calS$,
  this implies $\fmorph{\calU}{}{s_1} = \fmorph{\calU}{}{s_2}$ since
  $\varepsilon_R(s_1) = \varepsilon_R(s_2)$.  Thus, $s_1 \beq{\calU}{}
  s_2$.
  
  For the reverse, ${s_1 \beq{\calU}{} s_2 } \, \IMPL \, { s_1
    \FuTSbis{\calS} s_2 }$, assume $s_1 \beq{\calU}{} s_2$, i.e.
  $\fmorph{\calU}{}{s_1} = \fmorph{\calU}{}{s_2}$, for $s_1,s_2 \myin
  S$.  Since the map $\fmorph{\calU}{}{\cdot} : \twotuple{S}{\mytheta}
  \to \twotuple{\Omega}{\omega}$ is a $\calU$-homomorphism, the
  equivalence relation~$R_{\calS}$ with $R_{\calS} \mkern1mu (s',s'')
  \IFF \fmorph{\calU}{}{s'} = \fmorph{\calU}{}{s''}$ is an
  $\calS$-bisimulation: Suppose $R_{\calS} \mkern1mu (s',s'')$, i.e.\
  $s' \beq{\calU}{} s''$, for some $s',s'' \in S$.  
  Pick $\ell \in \calL$, $t \in S$ and assume 
  $\fmorph{\calU}{}{t} = w \in \Omega$.  Since
  $\fmorph{\calU}{}{\cdot} : \twotuple{S}{\mytheta} \to
  \twotuple{\Omega}{\omega}$ is a $\calU$-homomorphism we have that 
  $\omega \, \compose \,  \fmorph{\calU}{}{\cdot} \, = \,
  \calU(\fmorph{\calU}{}{\cdot})\,  \compose \, \theta$.  Hence, for
  $s \in S$, it holds that
  \begin{equation}
    \omega \mkern1mu ( \,\fmorph{\calU}{}{s} \,)( \myell )(w ) 
    = 
    \calU(\fmorph{\calU}{}{\cdot}) (\theta(s) )( \myell )(w) 
    = 
    \tssum_{t' \in \,  {\fmorph{\calU}{-1}{\cdot}} (w)} \;
    \theta(s)( \myell )(t')
    \label{eq-thetai-OmegaS}
  \end{equation}
  Therefore we have
  \begin{displaymath}
    \def\arraystretch{1.2}
    \begin{array}{rcll}
      \multicolumn{4}{l}{\tssum_{t' \myin [t]_{R_{\calS}}} \;
        \theta(s')(\myell)(t')} 
      \\ & = &
      \tssum_{t' \myin  \, \fmorph{\calU}{-1}{\cdot}\!(w)} \; \theta (s')(\myell)(t') 
      & \text{by definition of~$R_{\calS}$ and~$w$}
      \\ & = &
      \omega \mkern1mu ( \,\fmorph{\calU}{}{s'} \,)( \myell )(w )
      & \text{by equation~(\ref{eq-thetai-OmegaS})}
      \\ & = &
      \omega \mkern1mu ( \,\fmorph{\calU}{}{s''} \,)( \myell )(w )
      & \text{$s' \beq{\calU}{} s''$ by assumption}
      \\ & = &
      \tssum_{t' \myin \, \fmorph{\calU}{-1}{\cdot}\!(w)} \; \theta (s'')(\myell)(t')  
      & \text{by equation~(\ref{eq-thetai-OmegaS})}
      \\ & = &
      \tssum_{t' \myin [t]_{R_{\calS}}} \; \theta(s'')(\myell)(t') 
      & \text{by definition of~$R_{\calS}$ and~$w$}
    \end{array}
    \def\arraystretch{1.0}
  \end{displaymath}
  Thus, if $R_{\calS} \mkern1mu (s',s'')$ then $\tssum_{t' \myin
    [t]_{R_{\calS}}} \; \theta(s')(\myell)(t') = \tssum_{t' \myin
    [t]_{R_{\calS}}} \; \theta(s'')(\myell)(t')$ for all $t \in
  S$ and  $\ell \in \calL$, and therefore $R_{\calS}$
  is an $\calS$-bisimulation.  Since $\fmorph{\calU}{}{s_1} =
  \fmorph{\calU}{}{s_2}$, it follows that $R_{\calS} \mkern1mu
  (s_1,s_2)$.  Thus $R_{\calS}$ is an $\calS$-bisimulation relating
  $s_1$ and~$s_2$.  Conclusion, it holds that $s_1 \FuTSbis{\calS}
  s_2$.
\end{proof}

\noindent
In the next section we will provide $\FuTS$ semantics for a fragment of $\PEPA$, a representative process language.
For this language we will establish that its
standard notion of strong equivalence as known in the literature
coincides with the notion of strong bisimulation as induced by the
$\FuTS$ semantics. The results of this section form the basis for
showing that the standard notions of strong equivalence on the one
hand, and behavioural equivalence on the other hand, are all the
same. The notion of bisimulation for $\FuTS$ plays an intermediary
role: it bridges between the standard notion of concrete equivalence
and the abstraction notions from coalgebra.


\section{\texorpdfstring{$\FuTS$}{FuTS} Semantics of \texorpdfstring{$\PEPA$}{PEPA}}
\label{sec-pepa}

In this section we consider a significant fragment of the {\em
  Performance Evaluation Process Algebra}, $\PEPA$,~\cite{Hil96:phd}
--which we still call $\PEPA$ for simplicity-- including the parallel
operator implementing the scheme of so-called minimal apparent rates,
and provide a $\FuTS$ semantics for it. We point out that there is no
technical difficulty in extending the $\FuTS$ approach to the full
language; we do not do so here since its treatment does not yield a
conceptual benefit for this paper. We present a $\FuTS$ semantics for
$\PEPA$ in line with~\cite{De+14} and show that $\PEPA$'s notion of
equivalence~$\PEPAseq \mkern1mu$, called strong equivalence
in~\cite{Hil96:phd}, fits with the bisimilarity induced by the $\FuTS$
semantics.

\begin{defi}
  The set $\prcPEPA$ of $\PEPA$ processes is given by the grammar below:
  \begin{displaymath}
    P \bnfeq \nil \mid \alambda.P \mid P + P \mid P \sync{A} P \mid X
  \end{displaymath}
  where $a$~ranges over the set of actions~$\calA$, $\lambda$
  over~$\poreals$, $A$~over the set of finite subsets of~$\calA$, and
  $X$~over the set of constants~$\calX$. 
\end{defi}

\noindent
For $X \in \calX$, the notation $X \dfas P$ indicates that the
process~$P$ is associated with the process constant~$X$.  It is required that
each occurrence of a process constant in the body~$P$ of the
definition $X \dfas P$ is guarded by a prefix.

$\PEPA$, like many other $\SPC$, e.g.\ \cite{HHM98:cnis,BG98:tcs},
couples actions and rates.  The prefix $\alambda$ of the process
$\alambda.P$ expresses that the duration of the execution of the
action~$a \in \calA$ is sampled from a random variable with an
exponential distribution of rate~$\lambda$. The CSP-like parallel
composition $P \sync{A} Q$ of a process~$P$ and a process~$Q$ for a
set of actions~$A \subseteq \calA$ allows for the independent,
asynchronous execution of actions of $P$ or~$Q$ not occurring in the
subset~$A$, on the one hand, and requires the simultaneous,
synchronized execution of $P$ and~$Q$ for the actions occurring
in~$A$, on the other hand.  The transition rules of the
$\FuTS$-semantics of the fragment of $\PEPA$ we consider here is given
in Figure~\ref{fig-pepa-rules}, on which we comment below.
  
Characteristic for the $\PEPA$ language is the choice to model
parallel composition, or cooperation in the terminology of $\PEPA$,
scaled by the minimum of the so-called apparent rates.  By doing so,
$\PEPA$'s strong equivalence becomes a congruence~\cite{Hil96:phd}.
Informally, the apparent rate $r_a(P)$ of an action~$a$ for a
process~$P$ is the sum of the rates of all possible $a$-executions
for~$P$. The apparent rate $r_a(P)$ can easily be defined recursively
on the structure of~$P$ (see~\cite[Definition~3.3.1]{Hil96:phd} for
details). Accordingly, in the sequel we will refer to $r_a(P)$ as the
`syntactic' apparent rate. When considering the parallel composition
$P \sync{A} Q$, with cooperation set~$A$, an action~$a$ occurring
in~$A$ has to be performed by both $P$ and~$Q$. The rate of such an
execution is governed by the slowest of the two processes, on average,
in this respect. (One cannot take the slowest process per sample,
because such an operation cannot be expressed as an exponential
distribution in general.) Thus $r_a( \mkern1mu P \sync{A} Q \mkern1mu
)$ for $a \in A$ is the minimum $\min \SET{ \, r_a(P), \, r_a(Q) \,
}$.  Now, if $P$~schedules an execution of~$a$ with rate~$r_1$ and
$Q$~schedules a transition of~$a$ with rate~$r_2$, in the minimal
apparent rate scheme the combined execution yields the action~$a$ with
rate $r_1 \cdot r_2 \cdot \arf ( P , Q )$.  Here, the `syntactic'
scaling factor $\arf(P,Q)$, the apparent rate factor, is defined by
\begin{displaymath}
\def\arraystretch{1.0}
  \arf(P,Q) =
\begin{array}{c}
  \min \SET{ \, r_a(P), \, r_a(Q) \, }
  \\ \hline 
  r_a(P) \cdot r_a(Q)
\end{array}
\def\arraystretch{1.0}
\end{displaymath}
assuming $r_a(P), r_a(Q) > 0$, otherwise $\arf(P,Q) = 0$. Organizing
the product $r_1 \cdot r_2 \cdot \arf(P,Q)$ differently as $r_1/r_a(P)
\cdot r_2/r_a(Q) \cdot \min \SET{ \, r_a(P), \, r_a(Q) \, }$ we see
that for $P \sync{A} Q$ the minimum of the apparent rates $\min \SET{ \,
  r_a(P), \, r_a(Q) \, }$ is adjusted by the relative probabilities
$r_1/r_a(P)$ and~$r_2/r_a(Q)$ for executing~$a$ by~$P$ and~$Q$,
respectively.  

\begin{figure}
\begin{displaymath}
\scalebox{0.85}{$
\begin{array}{c}
\sosrn{NIL}{}
\sosrule
  {\phantom{\mtrans{\alpha}_\pepa}}
  {\nil \, \mtrans{\Edelaya}_\pepa \, \zerof_{\nnreals}} 
\qquad
\sosrn{RAPF1}{}
\sosrule
  {\phantom{\mtrans{\alpha}_\pepa}}
  {\alambda.P \, \mtrans{\Edelaya}_\pepa \, [P \mapsto \lambda]}
\qquad
\sosrn{RAPF2}{}
\sosrule
  {b \neq a}
  {\alambda.P \, \mtrans{\Edelayb}_\pepa \, \zerof_{\nnreals}} 
\bigskip \\
\sosrn{CHO}{}
\sosrule
  {P \, \mtrans{\Edelaya}_\pepa \, \amset{P} \quad 
   Q \, \mtrans{\Edelaya}_\pepa \, \amset{Q}}
  {P \cho Q \  \mtrans{\Edelaya}_\pepa \  \amset{P} \cho \amset{Q}}
\qquad
\sosrn{CNS}{}
\sosrule
  {P \, \mtrans{\Edelaya}_\pepa \, \amset{P} \quad X \dfas P}
  {X \, \mtrans{\Edelaya}_\pepa \, \amset{P}}
  \bigskip \\
\!\!\!\sosrn{PAR1}{}\!\!
\sosrule
  {P \, \mtrans{\Edelaya}_\pepa \, \amset{P} \quad
   Q \, \mtrans{\Edelaya}_\pepa \, \amset{Q} \quad 
   a \, \notin \, A}
  {P \sync{A} Q \  \mtrans{\Edelaya}_\pepa \ 
   ( \, \amset{P} \sync{A} \chut_Q \, ) 
   \, + \,
   ( \, \chut_P \sync{A} \amset{Q} \, )}
\quad\!\!\!\!
\sosrn{PAR2}{}
\sosrule
  {P \, \mtrans{\Edelaya}_\pepa \, \amset{P} \quad 
   Q \, \mtrans{\Edelaya}_\pepa \, \amset{Q} \quad 
   a \, \in \, A}
  {P \sync{A} Q \  \mtrans{\Edelaya}_\pepa \ 
   \arf( \mkern1mu {\amsetP}, {\amsetQ} \mkern1mu )
   \, \cdot \, ( \, \amset{P} \sync{A} \amset{Q} \, )}
\end{array}
$} 
\end{displaymath}
\halflineup
\halflineup
\caption{$\FuTS$ Transition Deduction System for $\PEPA$.}
\label{fig-pepa-rules}
\end{figure}

\blankline

\noindent
The $\FuTS$ we consider for the semantics of $\PEPA$ has been proposed
originally in~\cite{De+14}. The transition relation is given by the
rules in Figure~\ref{fig-pepa-rules}. The set of labels involved
is~$\Delta_{\calA}$ defined by $\Delta_{\calA} = \ZSET{\Edelaya}{a \in
  \calA }$. In the context of the $\FuTS$ semantics considered in this
paper, we conventionally use the special symbol $\delta$ for denoting
that there is
a random \emph{delay},  with an negative exponential
distribution, associated with the action. 
The underlying semiring for the $\FuTS$
for~$\PEPA$ is the semiring~$\nnreals$ of non-negative reals.
    
\begin{defi}
  \label{df-ltfs-pepa}
  The simple $\FuTS$ $\calSpepa =
  \twotuple{\prcPEPA}{\mtrans{}_\pepa}$ over $\Delta_{\calA}$ and
  $\nnreals$ has as transition relation the smallest relation
  satisfying the axioms and rules of Figure~\ref{fig-pepa-rules}.
\end{defi}

\noindent
We discuss the rules of~$\calSpepa$. The $\FuTS$ semantics provides
$\nil \mtrans{\Edelaya}_\pepa \zerof_{\nnreals}$, for every
action~$a$, with $\zerof_{\nnreals}$ the 0-function of~$\nnreals$.
Therefore we have $\thetapepa(\nil)(\Edelaya)(P') = 0$ for every $a
\in \calA$ and $P' \in \prcPEPA$, or, in standard terminology, $\nil$
has no transition.  For the rated action prefix $\alambda$ we
distinguish two cases: (i)~execution of the prefix in rule~(RAPF1);
(ii)~no~execution of the prefix in rule~(RAPF2).  In the case of
rule~(RAPF1) the label~$\Edelaya$ signifies that the transition
involves the execution of the action~$a$.  The continuation $[ \, P
\mapsto \lambda \, ]$ is the function that assigns the rate~$\lambda$
to the process~$P$.  All other processes are assigned~$0$, i.e.\ the
zero-element of the semiring~$\nnreals$.  In the second case,
rule~(RAPF2), for labels~$\Edelayb$ with $b \neq a$, we do have a
state-to-function transition, but it is a degenerate one.  The two
rules for the prefix, in particular having the `null-continuation'
rule (RAPF2), support the unified treatment of the choice operator in
rule (CHO) and the parallel operator in rules (PAR1) and~(PAR2). The
treatment of constants is as usual.

The semantics of the choice operator is defined by rule (CHO), where
the continuation of  process  $P \cho Q$
is given by direct composition---using pointwise sum---of the continuation 
$\amset{P}$ of $P$ and the continuation $\amset{Q}$ of $Q$.
 
Regarding the parallel operator~$\sync{A}$, with respect to some
cooperation set $A \subseteq \calA$ there are two rules. Now the
distinction is between interleaving and synchronization. In the case
of a label~$\Edelaya$ involving an action~$a$ not in the subset~$A$,
either the $P$-operand or the $Q$-operand of $P \sync{A} Q$ makes
progress.  For example, the effect of the pattern $\amset{P} \sync{A}
\chut_Q$ is that the value $\amset{P}(P') \cdot 1$ is assigned to a
process~$P' \sync{A} Q$, the value $\amset{P}(P') \cdot 0 = 0$ to a
process $P' \sync{A} Q'$ for all $Q' \neq Q$, and the value~$0$ for a
process not of the form $P' \sync{A} Q'$.  Note that the
  syntactic constructor $\sync{A}: \prcPEPA \times \prcPEPA \to
  \prcPEPA$ is clearly injective; so, for all functions $\amset{P}$
  and $\amset{Q}$ in $\fsfn{\prcPEPA}{\nnreals}$, we can define
  $\amset{P} \sync{A} \amset{Q}$, as described in
  Section~\ref{sec-preliminaries}.  Here, as in all other rules, the
right-hand sides of the transitions only involve functions
in~$\fsfn{\prcPEPA}{\nnreals}$ and operators on them.
  
For the synchronization case of the parallel construct, assuming $P
\mtrans{\Edelaya}_\pepa \amsetP$ and $Q \mtrans{\Edelaya}_\pepa
\amsetQ$, the `semantic' scaling factor $\arf(\amsetP,\amsetQ)$ is
applied to~$\amsetP \sync{A} \amsetQ $. This scaling factor for
continuation in~$\fsfn{\prcPEPA}{\nnreals}$, is, very much
similar to its `syntactic' counterpart, given by
\halflineup
\begin{displaymath}
  \arf( \mkern1mu \amsetP, \, \amsetQ \mkern1mu ) = 
  \begin{array}{c}
    \min \mkern2mu \SET{ \, \fsum \amsetP ,\, \fsum \amsetQ \, }
    \\ \hline
	\fsum \amsetP \cdot \fsum \amsetQ
  \end{array}
\end{displaymath}
provided $\fsum \amsetP, \fsum \amsetQ > 0$, and $\arf( \mkern1mu
\amsetP, \, \amsetQ \mkern1mu ) = 0$ otherwise. For a process~$R = R_1
\sync{A} R_2$ we obtain the value $\arf ( \mkern1mu \amsetP, \, \amsetQ
\mkern1mu ) \cdot ( \, \amsetP \sync{A} \amsetQ \, )( R_1 \sync{A} R_2) =
\arf ( \mkern1mu \amsetP, \, \amsetQ \mkern1mu ) \cdot \amsetP(R_1)
\cdot \amsetQ(R_2)$.
  
The following lemma establishes the relationship between the
`syntactic' and `semantic' apparent rate factors defined on processes
and on continuation functions, respectively.

\begin{lem}
\label{lm-sem-syn-arf}
  Let $P \in \prcPEPA$ and $a \in \calA$. 
  Suppose $P \mtrans{\Edelaya}_\pepa \amsetP$.
  Then $ r_a(P) = \fsum \amsetP$.
  \qed
\end{lem}

\noindent
The proof of the lemma is straightforward (relying on the obvious
definition of~$r_a(P)$, omitted above, which can be found in~\cite{Hil96:phd}).  It is also easy to prove, by
guarded induction, that the $\FuTS$ $\calSpepa$ given by
Definition~\ref{df-ltfs-pepa} is total and deterministic.  

\begin{lem}
\label{lm-pepa-total-det}
  The $\FuTS$ $\calSpepa$ is total and deterministic.
  \qed
\end{lem}

\noindent
In view of the lemma it is justified to write $\calSpepa
=\twotuple{\prcPEPA}{\thetapepa}$. We use the abbreviated
notation~$\FuTSbis{\pepa}$ for denoting $\FuTSbis{\calSpepa}$, the
bisimulation equivalence induced by~$\calSpepa$.
  
\blankline

\begin{exa}
  To illustrate the ease to deal with multiplicities in the $\FuTS$
  semantics, consider the $\PEPA$ processes $P_1 = \alambda.P$ and
  $P_2 = \alambda.P + \alambda.P$ for some~$P \in \prcPEPA$.  We have that 
  $P_1 \mtrans{\Edelaya}_\pepa [ \, P \mapsto \lambda \, ]$ by rule
  (RAPF1), but $P_2 \mtrans{\Edelaya}_\pepa [ \, P \mapsto 2 \lambda \, ]$
  by rule (RAPF1) and rule~(CHO).  The latter makes us to compute $[
  \, P \mapsto \lambda \, ] + [ \, P \mapsto \lambda \, ]$, which
  equals $[ \, P \mapsto 2 \lambda \, ]$.  Thus, in particular we have
  $P_1 \notFuTSbis{\calSpepa} P_2$.  Intuitively it is clear
  that, in general we cannot have $P + P \sim P$ for any reasonable
  quantitative process equivalence~$\sim$ in the Markovian setting.
  Having twice as many $a$-labelled transitions, the average number
  for $\alambda.P + \alambda.P$ of executing the action~$a$ per time
  unit is double the average of executing~$a$ for~$\alambda.P$.
\end{exa}

\begin{figure}
\begin{displaymath}
\scalebox{0.825}{$
\begin{array}{c}
\sosrn{RAPF}{}
\sosrule
  {\phantom{\mtrans{\alpha}_\pepa}}
  {\alambda.P \, \trans{a,\lambda}_\pepa P}
\qquad
\sosrn{CHO1}{}
\sosrule
  {P \, \trans{a,\lambda}_\pepa \, P'} 
  {P \cho Q \  \trans{a,\lambda}_\pepa P'} 
\qquad
\sosrn{CHO2}{}
\sosrule
  {Q \, \trans{a,\lambda}_\pepa \, Q'} 
  {P \cho Q \  \trans{a,\lambda}_\pepa P'} 
\bigskip \\
\sosrn{PAR1a}{}
\sosrule
  {P \, \trans{a,\lambda}_\pepa \, P' \quad 
   a \, \notin \, A}
  {P \sync{A} Q \  \trans{a,\lambda}_\pepa P' \sync{A} Q}
\quad
\sosrn{PAR1b}{}
\sosrule
  {Q \, \trans{a,\lambda}_\pepa Q' \quad 
   a \, \notin \, A}
  {P \sync{A} Q \  \trans{a,\lambda}_\pepa \  P \sync{A} Q'}
\bigskip \\
\sosrn{PAR2}{}
\sosrule
  {P \, \trans{a,\lambda_1} \, P' \quad 
   Q \, \trans{a,\lambda_2} \, Q' \quad 
   a \, \in \, A}
  {P \sync{A} Q \  \trans{a,\lambda}_\pepa \ P' \sync{A} Q'}
   \quad \text{$\lambda = \arf( \mkern1mu {P}, {Q} \mkern1mu ) {\cdot} \lambda_1 {\cdot} \lambda_2$}
\bigskip\\
\sosrn{CNS}{}
\sosrule
  {P \, \trans{a,\lambda}_\pepa \, P' \quad X \dfas P}
  {X \, \trans{a,\lambda}_\pepa \, P'}
\end{array}
$} 
\end{displaymath}
\halflineup
\halflineup
\caption{Standard Transition Deduction System for $\PEPA$.}
\label{fig-standard-pepa-rules}
\end{figure}

\blankline

\noindent
The standard operational semantics
of~$\PEPA$~\cite{Hil96:phd,Hil05:lics} is given in
Figure~\ref{fig-standard-pepa-rules}.  The transition relation
${\trans{}_\pepa} \subseteq \prcPEPA \times ( \, \calA \times \poreals
\, ) \times \prcPEPA$ is the least relation satisfying the rules.  For
an appropriate treatment of the rates, the transition relation is
considered as a multi-transition system, where also the number of
possible derivations of a transition $P \trans{a,\lambda}_\pepa P'$
matters.  We stress that such bookkeeping is not needed in the
$\FuTS$-approach.  In rule~(PAR2) we use the `syntactic' apparent rate
factor for $\PEPA$ processes.

The so-called total conditional transition rate $q[P,C,a]$ of a
$\PEPA$-process~\cite{Hil96:phd,Hil05:lics} for a subset of processes
$C \subseteq \prcPEPA$ and~$a \in \calA$ is given by
\begin{displaymath}
  q[P,C,a] 
  = 
  \tssum_{Q \myin C} \; \sum \MSET{ \lambda }{P \trans{a,\lambda}_\pepa Q}.
\end{displaymath}
Here, $\mlbrace \, P \trans{a,\lambda}_\pepa Q \, \mrbrace$ is the
multiset of transitions $P \trans{a,\lambda}_\pepa Q$ and $\MSET{
  \lambda }{P \trans{a,\lambda}_\pepa Q}$ is the multiset of
all~$\lambda$'s involved.  The multiplicity of $P
\trans{a,\lambda}_\pepa Q$ is to be interpreted as the number of
different ways the transition can be derived using the rules of
Figure~\ref{fig-standard-pepa-rules}.  We are now ready to define
$\PEPA$'s notion of strong equivalence.
  
\begin{defi}
  \label{df-strong-equiv}
  An equivalence relation $R \subseteq \prcPEPA \times \prcPEPA$ is
  called a strong equivalence if 
  \[
  q[ P_1 , \Rclass{Q} , a ] = q[ P_2 , \Rclass{Q} , a ]
  \]
   for all $P_1, P_2 \in \prcPEPA$ such that
  $R(P_1,P_2)$, all~$Q \in \prcPEPA$ and all~$a \in \calA$.  Two
  processes $P_1, P_2 \in \prcPEPA$ are strongly equivalent if
  $R(P_1,P_2)$ for a strong equivalence~$R$, notation 
  $P_1 \PEPAseq P_2$.  
\end{defi}

\noindent
The next lemma couples, for a $\PEPA$-process~$P$, an action~$a$ and a continuation
function $\amsetP \in \fsfn{\prcPEPA}{\nnreals}$, the
evaluation~$\amsetP(P')$ with respect to the $\FuTS$-semantics to the
cumulative rate for~$P$ of reaching~$P'$ by a transition involving the
label~$a$ in the standard operational semantics. The lemma is pivotal
in relating $\FuTS$ bisimulation and standard bisimulation for~$\PEPA$
in Theorem~\ref{th-pepa-strong-equiv-lfts-bisim} below.

\begin{lem}
\label{lm-cnt-ltfs-match}
  Let $P \in \prcPEPA$ and $a \in \calA$. 
  Suppose $P \mtrans{\Edelaya}_\pepa \amsetP$.
  The following holds: $\amsetP (P') = \tssum \; \MSET{ \lambda }{ P
    \trans{a,\lambda}_\pepa P' }$ for all~$P' \in \prcPEPA$. 
\end{lem}
\begin{proof}
  Guarded induction on~$P$. We only treat the cases for the parallel
  composition.  Note, the operation $ {\sync{A}} : \prcPEPA \times
  \prcPEPA \to \prcPEPA$ with ${\sync{A}}\twotuple{P_1}{P_2} = P_1
  \sync{A} P_2$ is injective.  Recall, for $\amsetP_1, \amsetP_2 \in
  \fsfn{\prcPEPA}{\nnreals}$, we have $(\amsetP_1 \sync{A}
  \amsetP_2)(P_1 \sync{A} P_2) = \amsetP_1(P_1) \cdot \amsetP_2(P_2)$.
 
  Suppose $a \notin \calA$.  Assume $P_1 \mtrans{\Edelaya}_\pepa
  \amsetP_1$, $P_2 \mtrans{\Edelaya}_\pepa \amsetP_2$, $P _1 \sync{A}
  P_2 \mtrans{\Edelaya}_\pepa \amsetP$.  We distinguish three cases.
  
\noindent
Case~(I), $P' = P'_1 \sync{A} P_2$, $P'_1 \neq P_1$.  Then we have
\begin{displaymath}
\def\arraystretch{1.2}
\begin{array}{rcl@{\qquad\qquad}l}
  \multicolumn{4}{l}{\tssum \MSET{ \lambda }{ P_1 \sync{A} P_2 \,
      \transalambda_\pepa \, P' }} 
  \\ & = & 
  \tssum \MSET{ \lambda }{ P_1 \, \transalambda_\pepa \, P'_1 }
  & \text{by rule (PAR1a)}
  \\ & = &
  \amsetP_1( P'_1 )
  & \text{by the induction hypothesis}
  \\ & = &
  \amsetP_1( P'_1 ) \cdot \chut_{P_2}( P_2 )
  & \text{since $\chut_{P_2}( P_2 ) = 1$}
  \\ & = & \multicolumn{2}{l}{%
    ( \amsetP_1 \sync{A} \chut_{P_2} )( P'_1\sync{A} P_2 )  +
    ( \chut_{P_1} \sync{A} \amsetP_2 )( P'_1 \sync{A} P_2 )
  } 
  \\ & & &
  \text{definition $\sync{A}$ on~$\fsfn{\prcPEPA}{\nnreals}$,
    $\chut_{P_1}(P'_1) = 0$} 
  \\ & = & 
  \amsetP( P' )
  & \text{by rule (PAR1)}
\end{array}
\def\arraystretch{1.0}
\end{displaymath}
Case~(II), $P' = P_1 \sync{A} P'_2$, $P'_2 \neq P_2$: similar.

\noindent
Case~(III), $P' = P_1 \sync{A} P_2$.
Then we have:
\begin{displaymath}
  \def\arraystretch{1.2}
  \begin{array}{rcl@{\quad}l}
    \multicolumn{4}{l}{\tssum \MSET{ \lambda }{ P_1 \sync{A} P_2 \,
        \transalambda_\pepa \, P' }} 
   \\ & = & \multicolumn{2}{l}{%
      \bigl( \, \tssum \MSET{ \lambda }{ P_1 \, \transalambda_\pepa \,
      P_1 } \, \bigr)
      +
      \bigl( \, \tssum \MSET{ \lambda }{ P_2 \, \transalambda_\pepa \,
      P_2 } \, \bigr)
    } 
    \\ & & \multicolumn{2}{r}{%
      \text{by rules (PAR1a) and (PAR1b)}
    } 
    \\ & = &
    \amsetP_1( P_1 )
    +
    \amsetP_2( P_2 )
    & \text{by the induction hypothesis}
    \\ & = & \multicolumn{2}{l}{%
      ( \amsetP_1 \sync{A} \chut_{P_2} )( P_1 \sync{A} P_2 )
      +
      ( \chut_{P_1} \sync{A} \amsetP_2 )( P_1 \sync{A} P_2 )
    } 
    \\ & & \multicolumn{2}{r}{%
      \text{definition $\sync{A}$ on~$\fsfn{\prcPEPA}{\nnreals}$,
        $\chut_{P_1}(P_1),\ \chut_{P_2}( P_2 ) = 1$} 
    } 
    \\ & = & 
    \amsetP( P' )
    & \text{again by rule (PAR1)}
  \end{array}
  \def\arraystretch{1.0}
\end{displaymath}

Suppose $a \in A$.  Assume $P_1 \mtrans{\Edelaya}_\pepa \amsetP_1$,
$P_2 \mtrans{\Edelaya}_\pepa \amsetP_2$, $P _1 \sync{A} P_2
\mtrans{\Edelaya}_\pepa \amsetP$.  Without loss of generality, $P' =
P'_1 \sync{A} P'_2$ for suitable $P'_1, P'_2 \in \prcPEPA$.
 
\begin{displaymath}
\def\arraystretch{1.2}
\begin{array}{rcl@{\quad}l}
 \multicolumn{4}{l}{\tssum \MSET{ \lambda }{ P_1 \sync{A} P_2 \,
     \transalambda_\pepa \, P' }}
\\ & = & \multicolumn{2}{l}{%
 \tssum \MSET{ \arf(P_1,P_2) \cdot \lambda_1 \cdot \lambda_2 }{ P_1 \,
   \trans{a,\lambda_1}_\pepa \, P'_1 ,\, P_2 \,
   \trans{a,\lambda_2}_\pepa \, P'_2 } 
   } 
   \\ & & \multicolumn{2}{r}{
     \text{by rule (PAR2)}
     } 
    \\ & = & \multicolumn{2}{l}{%
 \arf(P_1,P_2) \cdot 
 \bigl( \mkern1mu 
 \tssum \MSET{ \lambda_1 }{ P_1 \, \trans{a,\lambda_1}_\pepa \, P'_1 } 
 \mkern1mu \bigr) \cdot \bigl( \mkern1mu 
 \tssum \MSET{ \lambda_2 }{ P_2 \, \trans{a,\lambda_2}_\pepa \, P'_2 }
 \mkern1mu \bigr)
     } 
     \\ & & \multicolumn{2}{r}{%
      \text{by distributivity}
    } 
   \\ & = & \multicolumn{2}{l}{
 \arf(P_1,P_2) \cdot 
 \amsetP_1( P'_1 )
 \cdot
 \amsetP_2( P'_2 )
   }  
   \text{by the induction hypothesis}
     \\ & = & \multicolumn{2}{l}{%
 \arf( \amsetP_1, \amsetP_2) \cdot 
 \amsetP_1( P'_1 )
 \cdot
 \amsetP_2( P'_2 )
 } 
 \text{by Lemma~\ref{lm-sem-syn-arf}}
    \\ & = &  \multicolumn{2}{l}{%
 \arf(\amsetP_1,\amsetP_2) \cdot 
 ( \amsetP_1 \sync{A} \amsetP_2 )( P'_1 \sync{A} P'_2 )
 } 
\text{definition $\sync{A}$ on~$\fsfn{\prcPEPA}{\nnreals}$}%
 \\ & = & \multicolumn{2}{l}{%
 \amsetP( P' )
  } 
 \text{by rule (PAR2)}
\end{array}
\def\arraystretch{1.0}
\end{displaymath}
The other cases are simpler and omitted here.
\end{proof}

\noindent
With the lemma in place we can prove the following correspondence
result for $\calSpepa$-bisimilarity and strong equivalence as given by
Definition~\ref{df-strong-equiv}.
\newpage

\begin{thm}
\label{th-pepa-strong-equiv-lfts-bisim}
For $\PEPA$-processes $P_1, P_2 \in \prcPEPA$, it holds that
$P_1 \FuTSbis{\pepa} P_2$ iff $P_1 \PEPAseq P_2$.
\end{thm}
\begin{proof}
 Let~$R$ be an equivalence relation on~$\prcPEPA$.
 Choose $P,Q \in \prcPEPA$ and~$a \in \calA$.
 Suppose $P \mtrans{\Edelaya}_\pepa \amsetP$.
 Thus $\theta_\pepa (P)(\Edelaya) = \amsetP$.
 We have 
\begin{displaymath}
\def\arraystretch{1.2}
\begin{array}[t]{rcl@{\qquad}l}
 q[ P, \Rclass{Q}, a ]
 & = &
 \tssum_{Q' \myin \Rclass{Q}} \; \tssum\MSET{ \lambda }{P
   \trans{a,\lambda}_\pepa Q'} 
 & \text{by definition $q[ P, \Rclass{Q}, a ]$}
 \\ & = &
 \tssum_{Q' \myin \Rclass{Q}} \; \amsetP( Q' ) 
 & \text{by Lemma~\ref{lm-cnt-ltfs-match}}
 \\ & = &
 \tssum_{Q' \myin \Rclass{Q}} \; \thetapepa(P)(a)(Q')
 & \text{by definition $\thetapepa$}
\end{array}
\def\arraystretch{1.0}
\end{displaymath} 
Therefore, for $\PEPA$-processes $P_1$ and~$P_2$ it holds that $q[
P_1, \Rclass{Q}, a ] = q[ P_2, \Rclass{Q}, a ]$ for all $Q \in
\prcPEPA$, $a \in \calA$ iff $\tssum_{Q' \in \Rclass{Q}} \;
\thetapepa(P_1)(a)(Q') = \tssum_{Q' \in \Rclass{Q}} \;
\thetapepa(P_2)(a)(Q')$ for all $Q \in \prcPEPA$, $a \in \calA$.
Thus, the equivalence relation~$R$ is a strong equivalence
(Definition~\ref{df-strong-equiv}) iff $R$~is an
$\calSpepa$-bisimulation (Definition~\ref{df-ltfs-bisim}), from which
the theorem follows.
\end{proof}

\blankline

\noindent
By the theorem the $\FuTS$ semantics for~$\PEPA$ of
Definition~\ref{df-ltfs-pepa} is correct with respect to $\PEPA$'s
standard semantics of Figure~\ref{fig-standard-pepa-rules}. However,
because of the use of continuation functions, the former does not
involve implicit counting, decorations or multisets. From the general
results on $\FuTS$ of the previous section, we also obtain a
coalgebraic semantics for~$\PEPA$ for which behavioral equivalence
coincides with strong equivalence as defined in~\cite{Hil96:phd}.

\section{Combined \texorpdfstring{$\FuTS$}{FuTS}}
\label{sec-combined}

In the sequel of this article we will deal with a number of calculi
and models that mix non-deterministic behaviour with stochastic or
deterministic time or with probabilistic behaviour. In this section,
we introduce the notion of a \emph{combined} $\FuTS$, which allows for
a clean definition of the semantics of calculi where different aspects
of behaviour are integrated in an orthogonal way. Prominent
examples of such calculi are $\IML$, a language for~$\IMC$ where
non-determinism is integrated with stochastic continuous delays (see
Section~\ref{sec-iml}) and $\TPC$, a language where where
non-determinism is integrated with deterministic discrete delays (see
Section~\ref{sec-tpc}). 

\begin{defi}
  \label{df-gltfs}
  A combined $\FuTS$~$\calS$, in full `a combined state-to-function
  labeled transition system', over a number of label sets~$\calLi$
  and semirings~$\calRi$, $i = 1 , \ldots , n$, is a tuple $\calS = ( \, S
  ,\, \la \mkern-2mu {\mtrans{}_i} \ra^{\mkern1mu n}_{i = 1} \, )$
  with set of states~$S$ and such that ${\mtrans{}_i} \, \subseteq \,
  S \times \calLi \times \fsfn{\mkern2mu S}{\calRi}$, for $i = 1
  , \ldots , n$.
\end{defi}

\noindent
Combined $\FuTS$ of Definition~\ref{df-gltfs} extend the simple ones
of Definition~\ref{df-ltfs}. Note, a combined $\FuTS$ is defined over
a number of label sets and semirings, and, accordingly, gives rise to
the same number of transition relations. Thus, a combined $\FuTS$ can
be seen as a multi-dimensional simple $\FuTS$. The underlying idea is
that the behaviour model given by a combined $\FuTS$ is such that one can
identify different \emph{types} of labels, assuming disjoint label
sets $\calL_1 ,\, \ldots ,\, \calL_n$. Then, the continuation function
of a transition labeled with an element of~$\calLi$ is taken from
$\fsfn{\mkern2mu S}{\calRi}$, expressing the association of the label
set~$\calL_i$ with the semiring~$\calR_i$.

For example, in the case of $\IML$, with set of processes~$\prcIML$,
both non-deterministic behaviour and stochastically-timed behaviour
are treated. Furthermore, action execution is intended to be
instantaneous, while stochastic time is characterized by the rates
of negative exponential distributions. Consequently, it is convenient
to use two label sets, namely a set of actions~$\calA$ and a singleton
set~$\Delta = \SET{ \mkern2mu \delta \mkern2mu }$ where the
symbol~$\delta$ is used as label to indicate that the transition
involves some exponentially distributed delay. The relevant semirings
will be $\bools$, used for modeling the purely non-deterministic
aspects of behaviour, and~$\nnreals$, used for the rates
characterizing the stochastic aspects of behaviour, as in the case
of~$\PEPA$, but here without any association of delay and
actions. Consequently, for~$\IML$ there will be two transition
relations: ${\mtrans{}_1} \subseteq \prcIML \times \calA \times
\fsfn{\prcIML}{\bools}$ modeling non-deterministic behaviour, and
${\mtrans{}_2}\subseteq \prcIML \times \Delta \times
\fsfn{\prcIML}{\nnreals}$ modeling stochastic-time behaviour.

It is worth pointing out here that one could use an alternative
approach instead of taking resort to combined $\FuTS$, namely one
based on disjoint unions of label sets, and respectively, continuation
functions. Letting $\bigoplus_{i=1}^n X_i$ denote the disjoint union
of sets $X_i$, $i=1 , \ldots , n$, one could use a \emph{single}
transition relation 
\begin{displaymath}
  {\mtrans{}} \; \subseteq \; 
  S \, \times \,
  \bigoplus_{i=1}^n \calLi \, \times \, 
  \bigoplus_{i=1}^n \: \fsfn{S}{\calRi}
\end{displaymath}
satisfying the additional property that $v \in \fsfn{S}{\calRi}$ if
$\myell \in \calLi$, for all transitions $s \mtrans{\ell} v$.  As a
matter of fact, this approach based on disjoint unions and a single
transition relation has been used in~\cite{De+14}. Technically, the
two approaches are equivalent. On the other hand, in the definition
with a single transition relation, type compatibility between labels
and continuation functions yields an additional proof obligation for
the well-definedness the definition of the operational semantics for
every specific process calculus (the interested reader is referred
to~\cite{De+14} for details). The use of an approach with multiple
transition relations instead, automatically guarantees type
compatibility, viz.\ \emph{by definition}. Furthermore, the approach
based on disjoint unions appears less amenable to a
category-theoretical treatment. For the reasons mentioned we stick to
the format of Definition~\ref{df-gltfs} in this paper.

As we will see, for the purposes of the present paper it is sufficient
to consider only total and deterministic combined $\FuTS$, i.e.\ those
where every transition relation ${\mtrans{}_i}$ is a total
function. Consequently, it will be notationally convenient to consider
a combined $\FuTS$ $\calS = ( \, S ,\, \la \mkern-2mu {\mtrans{}_i}
\ra^{\mkern1mu n}_{i = 1} \, )$ as a tuple $( \, S ,\, \la \mkern-2mu
{\theta_i} \ra^{\mkern1mu n}_{i = 1} \, )$ with transition functions
$\theta_i : S \to \calLi \to \fsfn{\mkern2mu S}{\calRi}$, for $i = 1
, \ldots , n$, rather than using the form $( \, S ,\, \la \mkern-2mu
{\mtrans{}_i} \ra^{\mkern1mu n}_{i = 1} \, )$ that occurs more
frequently for concrete examples in the literature. In the sequel, we
occasionally omit the qualification `combined' for a combined $\FuTS$
when this cannot cause confusion. All relevant definitions and results
presented in Sections~\ref{sec-futs} and~\ref{sec-coalgebra} can be
extended straightforwardly to combined $\FuTS$. We refer
to~\cite{LMV13} for details on the extension of definitions, results
and their proofs. Here we recall the most important ones.

\begin{defi}
  \label{df-gltfs-bisim}
  For a combined $\FuTS$ $\calS = ( \, S ,\, \la \mkern-2mu {\theta_i}
  \ra^n_{i = 1} \,)$, an $\calS$-bisimulation is an equivalence
  relation $R \subseteq S \times S$ such that $R(s_1,s_2)$ implies
  \begin{displaymath}
    \tssum_{t' \in \Rclass{t}} \; 
    \theta_i \mkern2mu (s_1)(\ell \mkern2mu )(t' \mkern1mu ) 
    = 
    \tssum_{t' \in \Rclass{t}} \; 
    \theta_i \mkern2mu (s_2)(\ell \mkern2mu )(t' \mkern1mu )
  \end{displaymath}
  for all $t \in S$ and $\ell \in \calLi$, $i = 1 , \ldots , n$.  Two
  elements $s_1, s_2 \in S$ are called $\calS$-bisimilar for the
  combined $\FuTS$~$\calS$ if $R(s_1,s_2)$ for some
  $\calS$-bisimulation~$R$ for~$\calS$.  Notation $s_1 \FuTSbis{\calS}
  s_2$.
\end{defi}

\noindent
Working with total and deterministic $\FuTS$, we can interpret a
combined $\FuTS$ $\calS = ( \, S ,\, \la \mkern-2mu {\theta_i}
\ra^n_{i = 1} \, )$ over the label sets $\calLi$ and
semirings~$\calRi$, $i = 1 , \ldots , n$, as a product $\theta_1
\smalltimes \cdots \smalltimes \mkern2mu \theta_n : S \to
\prod_{i=1}^n \: ( \, \calLi \to \fsfn{S}{\calRi} \, )$ of functions
$\thetai : S \to \calLi \to \fsfn{S}{\calRi}$.  To push this idea a
bit further, we want to consider the combined $\FuTS$~$\calS = ( \,
\calS ,\, \la \mkern-2mu {\thetai} \ra^n_{i = 1} \,)$ as a coalgebra
of a suitable product functor on~sets.
  
\begin{defi}
  \label{df-futs-functor}
  Let $\vcalL = \langle\calL_1,\ldots,\calL_n\rangle$ be an $n$-tuple
  of label sets and $\vcalR = \langle \calR_1,\ldots,\calR_n\rangle$
  be an $n$-tuple of semirings. The functor $\calVvLR$ on~$\Set$ is
  defined by $\calVvLR = \prod_{i=1}^n \: \fsfn{ {\, \cdot \,}
  }{\calRi}^{\mkern1mu \calLi}$.
\end{defi}

\noindent
Referring to Definition~\ref{df-v-functor}, we have $\fsfn{ {\, \cdot
    \,} }{\calRi}^{\mkern1mu \calLi} = \calU_{\calRi}^{\mkern1mu
  \calLi}$, for $i = 1 , \ldots , n$. Therefore, $\calVvLR = \prod_{i=1}^n
\: \calU_{\calRi}^{\mkern1mu \calLi} $.  We note that any combined
$\FuTS$ $\calS = ( \, S ,\, \la \mkern-2mu {\theta_i} \ra^{\, n}_{i =
  1} \, )$ over label sets $\calL_i$ and semirings~$\calR_i$, for $i =
1 , \ldots , n$, is in fact a $\calVvLR$-coalgebra. Reversely, every
$\calVvLR$-coalgebra, for $\vcalL = \langle \calL_1, \ldots, \calL_n
\rangle$ and $\vcalR = \langle \calR_1, \ldots, \calR_n \rangle$,
corresponds to a combined $\FuTS$ over the label sets $\calL_i$ and
semirings $\calR_i$, for $i = 1 , \ldots , n$.
Below we shall use $\calV$ as an abbreviation for $\calVvLR$ whenever
$\vcalL = \langle \calL_1, \ldots, \calL_n \rangle$ and $\vcalR =
\langle \calR_1, \ldots, \calR_n \rangle$ are clear from the context.
Similarly, for the sake of readability, we shall often abbreviate
$\calU_{\calRi}^{\mkern1mu \calLi}$ by~$\calU_i$.



\blankline

\noindent
As product of accessible functors, the functor~$\calV$ of
Definition~\ref{df-futs-functor} is accessible and possesses a final
coalgebra, $(\Omega, \omega)$ say. So, we can speak of the behavioural
equivalence~$\beq{\calV}{}$ on any $\calV$-coalgebra or, equivalently,
of any combined $\FuTS$~$\calS$. Moreover, writing
$\fmorph{\calV}{}{\cdot}$ for the final morphism of a
$\calV$-coalgebra~$\calS$ into~$(\Omega, \omega)$, we have
\begin{displaymath}
  \fmorph{\calV}{}{\cdot} = 
  \fmorph{\calU_1}{}{\cdot} \times \cdots \times \:
  \fmorph{\calU_n}{}{\cdot} 
\end{displaymath}
Next we establish for a given $\FuTS$~$\calS$ over $\calL_1, \ldots,
\calL_n$ and $\calR_1\ldots,\calR_n$ the correspondence of
$\calS$-bisimulation~$\FuTSbis{\calS}$ and the behavioural
equivalence~$\beq{\calV}{}$ for the functor~$\calV$. Thus, one may
argue, Definition~\ref{df-gltfs-bisim} provides an explicit
description of behavioral equivalence. The proof of the theorem below
for combined $\FuTS$ is an adaptation of the proof of
Theorem~\ref{th-correspondence} for simple ones (see~\cite{LMV13} for
details).

\begin{thm}
  \label{th-combined-correspondence}
  Let $\calS = ( \, S ,\, \la \mkern-2mu {\theta_i} \ra^{\, n}_{i =
    1}\, )$ be a $\FuTS$ over the label sets $\calLi$ and
  semirings~$\calRi \mkern1mu$, $i = 1 , \ldots , n$, and $\calV$ as in
  Definition~\ref{df-futs-functor}.  Then $s_1 \FuTSbis{\calS} s_2
  \IFF s_1 \beq{\calV}{} s_2$, for all $s_1, s_2 \in S$.
  \qed
\end{thm}

\noindent
In the sequel of the paper we will consider combined $\FuTS$, as well
as a so-called general $\FuTS$, for concrete process languages. We
will show for each process language that the notion of bisimulation of
its $\FuTS$ coincides with the notion of strong bisimulation that is
associated in the literature with the language. Consequently, as a
corollary of Theorem~\ref{th-combined-correspondence}, we obtain that
the notions of strong bisimulations align with behavioral equivalence.


\section{\texorpdfstring{$\FuTS$}{FuTS} Semantics of \texorpdfstring{$\IML$}{IML}}
\label{sec-iml}

In this section we provide a $\FuTS$ semantics for a relevant part of
$\IML$, the language of Interactive Markov
Chains~\cite{Her02:springer}, $\IMC$ for short, and compare the
notion of bisimulation induced by its $\FuTS$ to the standard notion
of bisimulation based on the SOS-semantics as reported in the
literature.

$\IMC$~are automata that combine two types of transitions:
interactive transitions that involve the execution of actions, and
Markovian transitions that represent the progress of time governed by
exponential distributions.  As a consequence, $\IMC$ embody both
non-deterministic behaviour and stochastic, i.e.\ stochastically
timed, behaviour. System analysis using $\IMC$ proves to be a
powerful approach because of the orthogonality of qualitative and
quantitative dynamics, their logical underpinning and tool support,
cf.\ \cite{Boh06:tosen,HK10:fmco,Boz12:data}.  Such orthogonality
makes it natural to use a \emph{combined} $\FuTS$ for the
semantics of $\IML$.  A number of behavioural equivalences, both
strong and weak, are available for $\IMC$~\cite{EHZ10:concur}. In our
treatment here, we discuss a sublanguage of $\IML$, which we still
call $\IML$ for simplicity. In particular we do not deal with internal
$\tau$-steps, since we focus on strong bisimilarity here. The
$\FuTS$ semantics we consider in the sequel has been originally
proposed in~\cite{De+14}.

\begin{defi}
  The set $\prcIML$ of $\IML$ processes is given by the grammar 
  \begin{displaymath}
    P \bnfeq \nil \mid a.P \mid \lambda.P \mid P + P \mid P \prlA P
    \mid X 
  \end{displaymath}
  where $a$~ranges over the set of actions~$\calA$,
  $\lambda$~over~$\poreals$, $A$~over the set of finite subsets
  of~$\calA$ and $X$~over the set of constants~$\calX$.
\end{defi}

\noindent
We assume the same notation and (action) guardedness requirements for
constant definitions and usage as in Section~\ref{sec-pepa} for
$\PEPA$.

In line with the discussion above, in~$\IML$ there are separate prefix
constructions for actions~$a.P$ (meaning that the process
\emph{instantaneously performs} action~$a$ and then behaves like~$P$)
and for time-delays~$\lambda.P$ (meaning that the process is
\emph{delayed} for a period of time governed by a random variable with
negative exponential distribution with rate $\lambda$, and then
behaves like~$P$). No restriction is imposed on the alternative and
parallel composition of processes.  For example, in~$\IML$, we have
the process $a.\lambda.\nil \cho \mu.b.\nil$. With respect to the
$\FuTS$ semantics to be defined below, we will see that this
process admits both a non-trivial interactive transition and a
non-trivial Markovian transition,
\begin{displaymath}
  \begin{array}{rcccl}
    a.\lambda.\nil \cho \mu.b.\nil
    & \mtrans{a}_1 & 
    [ \, \lambda.\nil \mapsto \TRUE \, ] + \zerof_\bools 
    & = &
    [ \, \lambda.\nil \mapsto \TRUE \, ] \\
    a.\lambda.\nil \cho \mu.b.\nil
    & \mtrans{\delta}_2 &
    \zerof_{\nnreals} + [ \, b.\nil \mapsto \mu \, ] 
    & = &
    [ \, b.\nil \mapsto \mu \, ]
  \end{array}
\end{displaymath}
leading to an interactive continuation and a Markovian continuation,
respectively. 

\begin{defi}
\label{df-iml-sem}
The $\FuTS$ semantics of $\prcIML$ is given by the $\FuTS$ $\calSiml =
\threetuple{\, \prcIML}{\mtrans{}_1}{\mtrans{}_2}$, a combined $\FuTS$
over the label sets $\calA$ and~$\Delta = \SET{ \mkern2mu \delta
  \mkern2mu }$ and the semirings $\bools$ and~$\nnreals$ with
transition relations ${\mtrans{}_1} \subseteq \prcIML \times \calA
\times \fsfn{\prcIML}{\bools}$ and ${\mtrans{}_2}\subseteq \prcIML
\times \Delta \times \fsfn{\prcIML}{\nnreals}$ defined as the least
relations satisfying the rules of Figure~\ref{fig-iml-rules}.
\end{defi}

\begin{figure}
\begin{displaymath}
\scalebox{0.825}{$
\begin{array}{c}

\sosrn{NIL1}{}
\sosrule
  {a \in \calA}
  {\nil \, \mtrans{a}_1 \, \zerof_\bools} 
  
\qquad

\sosrn{NIL2}{}
\sosrule
  {}
  {\nil \, \mtrans{\Edelay}_2 \, \zerof_{\nnreals}}
  
\bigskip \\
  


\sosrn{RPF1}{}
\sosrule
  {a \in \calA}
  {\lambda.P \, \mtrans{a}_1 \, \zerof_\bools}
  
\qquad

\sosrn{RPF2}{}
\sosrule
  {}
  {\lambda.P \, \mtrans{\Edelay}_2 \, [P \mapsto \lambda]}
  
\bigskip \\

\sosrn{APF1}{}
\sosrule
  {}
  {a.P \, \mtrans{a}_1 \, [P \mapsto \TRUE]}
  
\quad

\sosrn{APF2}{}
\sosrule
  {b \neq a}
  {a.P \, \mtrans{b}_1 \, \zerof_\bools}
  
\quad


\sosrn{APF3}{}
\sosrule
  {}
  {a.P \, \mtrans{\Edelay}_2 \, \zerof_{\nnreals}} 

\quad


\bigskip \\

\sosrn{CHO1}{}
\sosrule
  {P \, \mtrans{a}_1 \, \amset{P} \quad 
   Q \, \mtrans{a}_1 \, \amset{Q}}
  {P \cho Q \  \mtrans{a}_1 \  \amset{P} \cho \amset{Q}} 

\qquad

\sosrn{CHO2}{}
\sosrule
  {P \, \mtrans{\delta}_2 \, \mycal{P} \quad 
   Q \, \mtrans{\delta}_2 \, \mycal{Q}}
  {P \cho Q \  \mtrans{\delta}_2 \  \mycal{P} \cho \mycal{Q}} 

\bigskip \\

\sosrn{PAR1}{}
\sosrule
  {P \, \mtrans{a}_1 \, \amset{P} \quad
   Q \, \mtrans{a}_1 \, \amset{Q} \quad 
   a \notin A}
  {P \prlA Q \  \mtrans{a}_1 \ 
   ( \, \amset{P} \prlA \chut_{\mkern-2mu Q} \, ) 
   \, + \,
   ( \, \chut_{\mkern-2mu P} \prlA \amset{Q} \, )}
  
\qquad

\sosrn{PAR2}{}
\sosrule
  {P \, \mtrans{a}_1 \, \amset{P} \quad 
   Q \, \mtrans{a}_1 \, \amset{Q} \quad 
   a \in A}
  {P \prlA Q \  \mtrans{a}_1 \ 
   \amset{P} \prlA \amset{Q}}
   
\bigskip \\

\sosrn{PAR3}{}
\sosrule
  {P \, \mtrans{\delta}_2 \, \mycal{P} \quad
   Q \, \mtrans{\delta}_2 \, \mycal{Q} \quad 
   \delta \notin A}
  {P \prlA Q \  \mtrans{\delta}_2 \ 
   ( \, \mycal{P} \prlA \chut_{\mkern-2mu Q} \, ) 
   \, + \,
   ( \, \chut_{\mkern-2mu P} \prlA \mycal{Q} \, )}

\bigskip \\  

\sosrn{CON1}{}
\sosrule
  {P \, \mtrans{a}_1 \, \amset{P} \quad 
   X \dfas P}
  {X \, \mtrans{a}_1 \, \amset{P}}

\qquad

\sosrn{CON2}{}
\sosrule
  {P \, \mtrans{\delta}_2 \, \mycal{P} \quad 
   X \dfas P}
  {X \, \mtrans{\delta}_2 \, \mycal{P}}

\end{array}
$} 
\end{displaymath}
\halflineup
\halflineup
\caption{$\FuTS$ Transition Deduction System for $\IML$.}
\label{fig-iml-rules}
\end{figure}

\noindent
Actions~$a \in \calA$ decorate~$\mtrans{}_1$, the
special symbol~$\delta$, with $\delta$ for delay,
decorates~$\mtrans{}_2$. Note that rule~(APF3) and rule~(RPF1) involve
the null-functions of~$\nnreals$ and of~$\bools$, respectively, to
express that a process~$a.P$ does not trigger a delay and a
process~$\lambda.P$ does not execute any action. In
Figure~\ref{fig-iml-rules} and in the rest of this section we use
$\amsetP, \amsetQ \in \fsfn{\prcIML}{\bools}$ as typical interactive
continuations, and $\mycalP, \mycalQ \in \fsfn{\prcIML}{\nnreals}$ as
typical Markovian continuations.

For the parallel construct~$\prlA$, interleaving applies both for
non-synchronized actions~$a \notin A$ as well as for delays.
Therefore we have rule~(PAR1) pertaining to $\mtrans{}_1$ and
rule~(PAR3) pertaining to~$\mtrans{}_2$. The same holds for
non-deterministic choice, rules (CHO1) and~(CHO2), and constants,
rules (CON1) and~(CON2). Finally, $\IML$ does not provide
synchronization of delays in the parallel construct. Hence,
rule~(PAR2) only concerns the transition relation~$\mtrans{}_1$
capturing synchronization on actions. We recall that for all $R \in
\prcIML$, on the one hand,
\begin{displaymath}
  (\amset{P} \prlA  \amset{Q})(R) \: = \:
  \left\{
    \begin{array}{cl}
      \amset{P}(R_1) \land \amset{Q}(R_2) &
      \mbox{if } R= R_1 \prlA R_2 \mbox{ for some } R_1, R_2 \in
      \prcIML \\
      \FALSE & \mbox{otherwise}
    \end{array}
  \right.
\end{displaymath}
and, on the other hand,
\begin{displaymath}
  (\mycal{P} \prlA  \mycal{Q})(R) \: = \:
  \left\{
    \begin{array}{cl}
      \mycal{P}(R_1) \cdot \mycal{Q}(R_2) &
      \mbox{if } R= R_1 \prlA R_2 \mbox{ for some } R_1, R_2 \in
      \prcIML \\
      0 & \mbox{otherwise}
    \end{array}
  \right.
\end{displaymath}
where $\cdot$ is the product in $\nnreals$.
  
\blankline

\begin{exa}
  For $a.(\lambda.\nil + b.\nil),\ \mu.a.\nil \in \prcIML$ and $A =
  \SET{a}$ we have
  \begin{displaymath}
    \def\arraystretch{1.3}
    \begin{array}{rcl}
      \multicolumn{3}{l}{a.(\lambda.\nil + b.\nil) \prlA \mu.a.\nil}
      \\ \qquad & 
      \mtrans{\Edelay}_2 &
      \zerof_{\nnreals} \prlA \chut^{}_{\mkern-2mu \mu.a.\nil}
      \; + \;
      \chut^{}_{\mkern-2mu a.(\lambda.\nil + b.\nil)} \prlA [
      \, a.\nil \mapsto \mu \,  ]
      \\ & = & 
      \zerof_{\nnreals} \prlA [ \, \mu.a.\nil \mapsto 1 \, ]
      \; + \;
      [ a.(\lambda.\nil + b.\nil) \mapsto 1 \, ] \prlA [
      \, a.\nil \mapsto \mu \,  ]
      \\ & = &
      [ \, a.(\lambda.\nil + b.\nil) \prlA a.\nil
      \: \mapsto \: \mu \,  ] 
      \\
    \end{array}
    \def\arraystretch{1.0}
  \end{displaymath}
  For $X \dfas a.\lambda.b.X$ and $Y \dfas a.\mu.b.Y$,
  and $A = \SET{a,b}$
  we have
\begin{displaymath}
\begin{array}{r@{\, \prlA \,}l@{\;}c@{\;}l@{\,}r@{\, \prlA
      \,}l@{\,}lcr@{\, \prlA \,}l@{\;}c@{\;}l@{\,}r@{\, \prlA
      \,}l@{\;}l} 
  X & Y & \mtrans{a}_1 & 
  \, [ & \lambda.b.X & \mu.b.Y \: & \mapsto \TRUE \, ] 
  & \  &
  \lambda.b.X & b.Y & \mtrans{\Edelay}_2 &
  \, [ & b.X & b.Y & \, \mapsto \, \lambda \, ]
  \\
  b.X & b.Y \; & \mtrans{b}_1 &
  \, [ & X & Y & \mapsto \TRUE \, ]
  && b.X & \mu.b.Y \: & \mtrans{\Edelay}_2 &
  \, [ & b.X & b.Y & \, \mapsto \, \mu \, ]
  \\ \multicolumn{15}{c}{%
  \lambda.b.X \, \prlA \, \mu.b.Y \mtrans{\Edelay}_2 
  [  b.X \, \prlA \, \mu.b.Y \mapsto \lambda ,\ 
  \lambda.b.X \prlA b.Y \mapsto \mu \, ] }
\end{array}
\end{displaymath}
\end{exa}

\blankline

\noindent
It is not difficult to verify that $\calSiml$ is a total and
deterministic combined $\FuTS$.

\begin{lem}
\label{lm-iml-total-det}
  The $\FuTS$ $\calSiml$ is total and deterministic.
  \qed
\end{lem}

\noindent
Below we use $\calSiml = \threetuple {\prcIML}{\theta_1}{\theta_2}$
and write $\FuTSbis{\iml}$ rather than $\FuTSbis{\calSiml}$, the
bisimulation equivalence induced by~$\calSiml$.
 
\blankline

\noindent
The standard SOS semantics of~$\IML$~\cite{Her02:springer} is given in
Figure~\ref{fig-standard-iml} involving the transition relations
\begin{displaymath}
  {\trans{}} \subseteq \prcIML \times \calA \times \prcIML 
  \qquad \text{and} \qquad 
  {\Mtrans{}} \subseteq \prcIML \times \poreals \times \prcIML
\end{displaymath}
Below we will use functions $\bfT$ and~$\bfR$ based on $\trans{}$
and~$\Mtrans{}$, cf.~\cite{HK10:fmco}.  We have $\bfT \colon \prcIML
\times \calA \times {\textbf{2}}^{\prcIML} \to \bools$ given by $\bfT(
P, a, C ) = \TRUE$ if the set $\ZSET{ P' \in C }{ P \trans{a} P' }$ is
non-empty, for all $P \in \prcIML$, $a \in \calA$ and any subset~$C
\subseteq \prcIML$.  For $\bfR \colon \prcIML \times \prcIML \to
\nnreals$ we put $\bfR(P,P') = \tssum \MSET{ \lambda }{ P
  \Mtrans{\lambda} P' }$.  Here, as common for probabilistic and
stochastic process algebras, the comprehension is over the multiset of
transitions leading from~$P$ to~$P'$ with label~$\lambda$.
Alternatively, one could define an explicit $\cnt$-function, $\cnt :
\prcIML \times \poreals \times \prcIML \to \nnreals$ returning the
number of multiplicities of a transition $P \Mtrans{\lambda} P'$, or
other means of decorations. We extend $\bfR$ to $\prcIML \times
{\textbf{2}}^{\prcIML}$ by $\bfR(P,C) = \tssum_{P' \myin C} \; \tssum
\MSET{ \lambda }{ P \Mtrans{\lambda} P' }$, for $P \in \prcIML$, $C
\subseteq \prcIML$ .
  
%
%
%
%
%
%
%
%
%
%
%
%
%
%
%
%
%
%
%
%
%
%
%
%
%
%

\begin{figure}
\begin{displaymath}
\scalebox{0.825}{$
\begin{array}{c}

\sosrn{APF}{}
\sosrule
  {}
  {\pfx{a}{P} \, \trans{a} \, P}
  
\quad

\sosrn{RPF}{}
\sosrule{}
  {\pfx{\lambda}{P} \, \Mtrans{\lambda} \, P}

\bigskip \\

\sosrn{CHO1}{}
\sosrule
  {P \, \trans{a} \, R}
  {P \cho Q  \, \trans{a} \, R}

\qquad

\sosrn{CHO2}{}
\sosrule
  {Q \, \trans{a} \, R}
  {P \cho Q  \, \trans{a} \, R}

\bigskip \\
 
\sosrn{CHO3}{}
\sosrule
  {P \, \Mtrans{\lambda} \, R}
  {P \cho Q  \, \Mtrans{\lambda} \, R}

\qquad

\sosrn{CHO4}{}
\sosrule
  {Q \, \Mtrans{\lambda} \, R}
  {P \cho Q \, \Mtrans{\lambda} \, R}
 
\bigskip \\ 

\sosrn{PAR1a}{}
\sosrule
  {P \, \trans{a} \, P' \quad a \notin A}
  {P \prlA Q \, \trans{a}\ , P' \prlA Q}

\qquad

\sosrn{PAR1b}{}
\sosrule
  {Q \, \trans{a} \, Q' \quad a \notin A}
  {P \prlA Q \, \trans{a} \, P \prlA Q'}

\qquad

\bigskip \\

\qquad

\qquad

\sosrn{PAR1c}{}
\sosrule
  {P \, \Mtrans{\lambda} \, P'}
  {P \prlA Q \, \Mtrans{\lambda} \, P' \prlA Q}

\qquad

\sosrn{PAR1d}{}
\sosrule
  {Q \, \Mtrans{\lambda} \, Q'}
  {P \prlA Q \, \Mtrans{\lambda} \, P \prlA Q'}

\bigskip\\

\sosrn{PAR2}{}
\sosrule
  {P \, \trans{a} \, P' \quad 
   Q \, \trans{a} \, Q' \quad a \in A}
  {P \prlA Q \, \trans{a} \, P' \prlA Q'}
  
\bigskip\\

\sosrn{CON1}{}
\sosrule
  {P \, \trans{a} \, Q \quad X \dfas P}
  {X \, \trans{a} \, Q}

\qquad

\sosrn{CON2}{}
\sosrule
  {P \,  \Mtrans{\lambda} \, Q \quad X \dfas P}
  {X \, \Mtrans{\lambda} \, Q}
  
\end{array}
$} 
\end{displaymath}

\caption{Standard Transition Deduction System for $\IML$.}
\label{fig-standard-iml}
\end{figure}

For $\IML$ we have the following notion of strong
bisimulation~\cite{Her02:springer,HK10:fmco} that we will compare with
the notion of bisimulation associated with the $\FuTS$~$\calSiml$.
 
\begin{defi}
\label{df-iml-strong-bisimulation}
  An equivalence relation $R \subseteq \prcIML \times \prcIML$ is
  called a strong bisimulation for~$\IML$ if, for all $P_1, P_2 \in
  \prcIML$ such that $R(P_1,P_2)$, it holds that
  \begin{itemize}
  \item for all $a \in \calA$ and $Q \in \prcIML$: $\bfT
    \threetuple {P_1} a {\Rclass{Q}} \iff \bfT \threetuple {P_2} a
    {\Rclass{Q}}$
  \smallskip
  \item for all $Q \in \prcIML$: $\bfR \twotuple
    {P_1}{\Rclass{Q}} \; = \; \bfR \twotuple {P_2}{\Rclass{Q}}$.
\end{itemize}
  Two processes $P_1, P_2 \in \prcIML$ are called strongly bisimilar
  if $R(P_1,P_2)$ for a strong bisimulation~$R$ for~$\IML$, notation
  $P_1 \IMLsbis P_2$.  
\end{defi}

\noindent
To establish the correspondence of $\FuTS$
bisimilarity~$\FuTSbis{\iml}$ for~$\calSiml$ as given by
Definition~\ref{df-iml-sem} and strong bisimilarity~$\IMLsbis$
for~$\IML$ as given by Definition~\ref{df-iml-strong-bisimulation}, we
need to connect the state-to-function relation~$\mtrans{}_1$ and the
transition relation~$\trans{}$ as well as the state-to-function
relation~$\mtrans{}_2$ and the transition relation~$\Mtrans{} \,$.

\pagebreak[3]
\begin{lem}
\label{lm-mtrans-trans}
\hfill
\begin{enumerate}[label=\({\alph*}]
\item Let $P \in \prcIML$ and $a \in \calA$.  If $P \mtrans{a}_1
  \amsetP$ then $P \trans{a} P' \iff \amsetP( P' ) = \TRUE$.
\item Let $P \in \prcIML$ .  If $P \mtrans{\Edelay}_2 \mycalP$
  then $\tssum \MSET{ \lambda }{ P \Mtrans{\lambda} P' } = \mycalP( P'
  )$.  
\end{enumerate}
\end{lem}\newpage

\begin{proof}\hfill\smallskip

\noindent \(a\
 Guarded induction. Let $a \in \calA$. We treat two typical cases,
 viz.\ $\lambda.P$ and $P_1 \prlA P_2$ for $a \notin A$.
 
 Case~$\lambda.P$.  Suppose $\lambda.P \mtrans{a}_1 \amsetP$.  Then we
 have $\amsetP = \zerof_\bools$. We have $\lambda.P \trans{a} P'$ for
 no~$P' \in \prcIML$, as no transition is provided in~$\trans{}$, and
 we have $\amsetP( P' ) = \FALSE$ by definition of~$\zerof_\bools$,
 for all~$P' \in \prcIML$.
 
 Case~$P_1 \prlA P_2$, $a \notin A$.  Suppose $P_1 \mtrans{a}_1
 \amsetP_1$, $P_2 \mtrans{a}_1 \amsetP_2$ and $P_1 \prlA P_2
 \mtrans{a}_1 \amsetP$.  Then it holds that $\amsetP = ( \amsetP_1
 \prlA \chut_{P_2} ) \cho ( \chut_{P_1} \prlA \amsetP_2 )$.  Recall,
 for~$Q \in \prcIML$, by definition of $\chut_Q \in
 \fsfn{\prcIML}{\bools}$, $\chut_Q(Q') = \TRUE$ iff $Q' = Q$, for~$Q'
 \in \prcIML$.  We have
\begin{displaymath}
\def\arraystretch{1.2}
\begin{array}{rcl}
 \multicolumn{3}{l}{%
 P_1 \prlA P_2 \trans{a} P'
 } 
 \\ & \Leftrightarrow &
 ( \, P_1 \trans{a} P'_1 \land P' = P'_1 \prlA P_2 \, ) 
 \; \lor \;
 ( \, P_2 \trans{a} P'_2 \land P' = P_1 \prlA P'_2 \, ) 
 \\ &&
 \text{by analysis of $\trans{}$}
 \\ & \Leftrightarrow &
 ( \, \amsetP_1( P'_1 ) = \TRUE \land P' = P'_1 \prlA P_2 \, )
 \; \lor \;
 ( \, \amsetP_2( P'_2 ) = \TRUE \land P' = P_1 \prlA P'_2 \, )
 \\ && \text{by the induction hypothesis}
 \displaybreak[2] \\ & \Leftrightarrow &
 ( \, \amsetP_1( P'_1 ) \cdot \chut_{P_2}(P_2) = \TRUE \land P' = P'_1 \prlA P_2 \, )
 \; \lor \;
 {} \\ && \qquad \qquad \qquad
 ( \, \chut_{P_1}(P_1) \cdot \amsetP_2( P'_2 ) = \TRUE \land P' = P_1 \prlA P'_2 \, )
 \\ && \text{by definition of $\chut_{P_1}$ and~$\chut_{P_2}$}
 \displaybreak[2] \\ & \Leftrightarrow &
 ( \, ( \amsetP_1 \prlA \chut_{P_2} )(P'_1 \prlA P_2) = \TRUE \land P' = P'_1 \prlA P_2 \, )
 \; \lor \;
 {} \\ && \qquad \qquad \qquad
 ( \, ( \chut_{P_1} \prlA \amsetP_2 )( P_1 \prlA P'_2 ) = \TRUE \land P' = P_1 \prlA P'_2 \, )
 \\ && \text{by definition of~$\prlA$}
 \displaybreak[2] \\ & \Leftrightarrow &
 ( \amsetP_1 \prlA \chut_{P_2} )(P') = \TRUE 
 \; \lor \;
 ( \chut_{P_1} \prlA \amsetP_2 )(P') = \TRUE 
 \\ && \text{by definition of~$\prlA$, $\chut_{P_1}$ and~$\chut_{P_2}$}
 \displaybreak[2] \\ & \Leftrightarrow &
 ( \, ( \amsetP_1 \prlA \chut_{P_2} ) \cho ( \chut_{P_1} \prlA \amsetP_2 ) \, )( P' ) = \TRUE
 \\ && \text{by definition of~$\cho$ on~$\fsfn{\prcIML}{\bools}$}
 \\ & \Leftrightarrow &
 \amsetP( P' ) = \TRUE
\end{array}
\def\arraystretch{1.0}
\end{displaymath}
The other cases are standard, or similar and easier.\medskip

\noindent\(b\ Guarded induction. We treat the cases for~$\mu.P$ and $P_1 \prlA
 P_2$.

Case~$\mu.P$.  Assume $\mu.P \mtrans{\Edelay}_2 \mycalP$, then
$\mycalP = [ \, P \mapsto \mu \, ]$. Moreover, it holds that $\mu.P$
admits a single $\Mtrans{} \mkern3mu$-transition, viz.\ $\mu.P
\Mtrans{\mu} P$.  \\Thus we have $\tssum \MSET{ \lambda }{ \mu.P
  \Mtrans{\lambda} P' } = \mu = [ \, P \mapsto \mu \, ](P) = \mycalP(
P )$.
 
 Case~$P_1 \prlA P_2$.  Assume $P_1 \mtrans{\Edelay}_2 \mycalP_1$,
 $P_2 \mtrans{\Edelay}_2 \mycalP_2$ and $P_1 \prlA P_2
 \mtrans{\Edelay}_2 \mycalP$.  It holds that $\mycalP = ( \mycalP_1
 \prlA \chut_{P_2} ) \cho ( \chut_{P_1} \prlA \mycalP_2 )$.  We
 calculate
\allowdisplaybreaks
\def\arraystretch{1.2}
\begin{align*}
 \tssum \MSET{ \lambda & }{ P_1 \prlA P_2 \Mtrans{\lambda} P' }
 \\ & = 
  \tssum \MSET{ \lambda }{ P_1 \Mtrans{\lambda} P'_1 ,\ P' = P'_1
    \prlA P_2 } \cho \tssum \MSET{ \lambda }{ P_2 \Mtrans{\lambda}
    P'_2 ,\ P' = P_1 \prlA P'_2 } 
 \\ & \qquad \text{by analysis of $\Mtrans{}$}
 \\ & = 
 ( \, \texttt{if} \; P' = P'_1 \prlA P_2 \; \texttt{then} \  \tssum
 \MSET{ \lambda }{ P_1 \Mtrans{\lambda} P'_1 } \ \texttt{else} \ 0
 \ \texttt{end} \, )  
 \cho 
 {} \\ & \qquad \qquad
 ( \, \texttt{if} \; P' = P_1 \prlA P'_2 \; \texttt{then} \  \tssum
 \MSET{ \lambda }{ P_2 \Mtrans{\lambda} P'_2 } \  \texttt{else} \ 0
 \ \texttt{end} \, ) 
 \\ & = 
 ( \, \texttt{if} \; P' = P'_1 \prlA P_2 \; \texttt{then}
 \  \mycalP_1(P'_1) \ \texttt{else} \ 0 \ \texttt{end} \, )  
 \cho 
 {} \\ & \qquad \qquad
 ( \, \texttt{if} \; P' = P_1 \prlA P'_2 \; \texttt{then}
 \  \mycalP_2(P'_2) \  \texttt{else} \ 0 \ \texttt{end} \, ) 
 \\ & \qquad \text{by induction hypothesis for $P_1$ and~$P_2$}
 \\ & = 
 ( \, \mycalP_1 \prlA \chut_{P_2} \, )(P') \cho ( \, \chut_{P_1} \prlA \mycalP_2 \, )(P')
 \\ & \qquad \text{by definition of~$\prlA$, $\chut_{P_1}$, $\chut_{P_2}$
   and~$\cho$ on~$\fsfn{\prcIML}{\nnreals}$} 
 \\ & = 
 \mycalP(P')
\end{align*}
\def\arraystretch{1.0}
 The remaining cases are left to the reader.
\end{proof}

\noindent
We are now in a position to relate $\FuTS$ bisimilarity and standard
strong bisimilarity for~$\IML$. In essence,
Lemma~\ref{lm-mtrans-trans} is all we need.

\begin{thm}
\label{th-correpondence-iml}
  For any two processes $P_1, P_2 \in \prcIML$ it holds that 
  $P_1  \FuTSbis{\iml} P_2$ iff $P_1 \IMLsbis P_2$.
\end{thm} 

\begin{proof}
 Let~$R$ be an equivalence relation on~$\prcIML$.
 Pick $P \in \prcIML$, $a \in \calA$ and choose any~$Q \in \prcIML$.
 Suppose $P \mtrans{a}_1 \amsetP$. 
 Thus $\theta_1(P)(a) = \amsetP$.
 Then we have
\begin{displaymath} 
\def\arraystretch{1.2}
\begin{array}{rcll}
 \bfT \threetuple {P} a {\Rclass{Q}}  
 & \Leftrightarrow &
 \exists Q' \in \Rclass{Q} \colon P \trans{a} Q'
 & \text{by definition of $\bfT$}
 \\ & \Leftrightarrow &
 \exists Q' \in \Rclass{Q} \colon \amsetP(Q') = \TRUE
 & \text{by Lemma~\ref{lm-mtrans-trans}a}
 \\ & \Leftrightarrow &
 \tssum_{Q' \myin \Rclass{Q}} \; \theta_1(P)(a)(Q') = \TRUE
 & \text{by definition of $\theta_1$}
\end{array}
\def\arraystretch{1.0}
\end{displaymath}
Note, summation in~$\bools$ is disjunction.  Likewise, on the
Markovian side, we have
\begin{displaymath} 
\def\arraystretch{1.2}
\begin{array}{rcll}
 \bfR \twotuple {P}{\Rclass{Q}}  
 & = &
 \tssum_{Q' \myin \Rclass{Q}} \; \tssum \MSET{ \lambda }{ P \Mtrans{\lambda} Q' }
 & \text{by definition of $\bfR$}
 \\ & = &
 \tssum_{Q' \myin \Rclass{Q}} \; \mycalP(Q' \mkern1mu ) 
 & \text{by Lemma~\ref{lm-mtrans-trans}b}
 \\ & = &
 \tssum_{Q' \myin \Rclass{Q}} \; \theta_2(P)(\Edelay)(Q)
 & \text{by definition of $\theta_2$}
\end{array}
\def\arraystretch{1.0}
\end{displaymath}
We conclude that a strong bisimulation
for~$\IML$ is also an $\calSiml$-bisimulation for the
$\pFuTS$~$\calSiml$, and vice versa.  From this the theorem follows.
\end{proof}

\noindent
From the theorem we conclude that also for~$\IML$ the concrete notion
of strong bisimilarity $\IMLsbis$ is coalgebraically underpinned, as
it coincides with the behavioral equivalence $\FuTSbis{\iml}$ that
comes with the corresponding $\FuTS$~$\calSiml$.

\section{\texorpdfstring{$\FuTS$}{FuTS} Semantics of \texorpdfstring{$\TPC$}{TPC}}
\label{sec-tpc}

In this section we consider a simple language of timed processes for
which we provide a combined~$\FuTS$. The language is a relevant
fragment of the timed process algebra $\TPC$ presented
in~\cite{ABC10:springer}. The model of time under consideration is
discrete and deterministic. The relevant construct is the
time-prefix~$(n).P$, with $n \in \nats$, $n > 0$, expressing that the
process~$P$ is to be executed after $n$~time steps. We will provide a
$\FuTS$ semantics and compare the induced notion of bisimulation to
the notion of timed bisimulation underlying the operational semantics
reported in~\cite{ABC10:springer}.

To the best of our knowledge, this is the first time a
deterministically timed model is dealt with in the coalgebraic
framework. As we will see, we resort to~$\textbf{2}^{\nats}$ as
co-domain for the time continuations, instead of just $\nats$, as one
may expect.  In particular, we use the semiring ~$\textbf{2}^{\nats}$
with set union as sum and intersection as multiplication.
The reason of this choice is mainly technical and is
connected to the proof of the bisimulation correspondence theorem
(Theorem~\ref{th-correpondence-tpc} below). Furthermore, the
appropriate treatment of delays requires the extension of the set of
operators on continuations.

\begin{defi}
  The set $\prcTPC$ of $\TPC$ processes is given by the grammar below:
  \begin{displaymath}
    P \bnfeq \nil \mid a.P \mid (n) \,.\, P \mid P + P \mid P \prlA P
    \mid X 
  \end{displaymath}
  where $a$~ranges over the set of actions~$\calA$, $n$~over $\nats$
  with~$n > 0$, $A$~over the set of finite subsets of~$\calA$, and
  $X$~over the set of constants~$\calX$.
\end{defi}

\noindent
We assume the same notation and guardedness requirements for constant
definition and usage as for $\PEPA$ or~$\IML$.

\begin{defi}
  \label{df-tpc-sem}
  The formal semantics of $\prcTPC$ is given by the $\FuTS$ $\calStpc
  = \threetuple{\, \prcTPC}{\mtrans{}_1}{\mtrans{}_2}$, a combined
  $\FuTS$ over the label sets $\calA$ and~$\Theta$ with $\Theta =
  \SET{ \mkern2mu \surd \mkern4mu }$ and the semirings $\bools$
  and~$\textbf{2}^\nats$ with transition relations ${\mtrans{}_1}
  \subseteq \prcTPC \times \calA \times \fsfn{\prcTPC}{\bools}$ and
  ${\mtrans{}_2}\subseteq \prcTPC \times \Theta \times
  \fsfn{\prcTPC}{\textbf{2}^\nats}$ defined as the least relations
  satisfying the rules of Figure~\ref{fig-tpc-rules}.
\end{defi}

\begin{figure}
\begin{displaymath}
\scalebox{0.90}{$
\begin{array}{c}
\sosrn{NIL1}{}
\sosrule
  {a \in \calA}
  {\nil \, \mtrans{a}_1 \, \zerof_\bools} 
\qquad
\sosrn{NIL2}{}
\sosrule
  {}
  {\nil \, \mtrans{\surd}_2 \, \zerof_{\textbf{2}^\nats}}
\bigskip \\
\sosrn{APF1}{}
\sosrule
  {}
  {a.P \, \mtrans{a}_1 \, [P \mapsto \TRUE]}
\qquad
\sosrn{APF2}{}
\sosrule
  {b \neq a}
  {a.P \, \mtrans{b}_1 \, \zerof_\bools}
\qquad
\sosrn{APF3}{}
\sosrule
  {}
  {a.P \, \mtrans{\surd}_2 \, \zerof_{\textbf{2}^\nats}} 
\qquad
\bigskip \\
\sosrn{TPF1}{}
\sosrule
  {a \in \calA}
  {(n).P \, \mtrans{a}_1 \, \zerof_\bools}
\quad
\sosrn{TPF2}{}
\sosrule
  {P \, \mtrans{\surd}_2 \, \mycalP}
  {(n).P \, \mtrans{\surd}_2 \, [n;P] + [P \mapsto \SET{n}] + (n + \mycalP)}
\bigskip \\
\sosrn{CHO1}{}
\sosrule
  {P \, \mtrans{a}_1 \amset{P} \quad 
   Q \, \mtrans{a}_1 \amset{Q}}
  {P \cho Q \  \mtrans{a}_1 \, \amset{P} \! \cho \! \amset{Q}}
\qquad
\sosrn{CHO2}{}
\sosrule
  {P \, \mtrans{\surd}_2 \, \mycal{P} \quad 
   Q \, \mtrans{\surd}_2 \, \mycal{Q}}
  {P \cho Q \  \mtrans{\surd}_2 \  \mycal{P} \myboxed{\cho} \mycal{Q}}
\bigskip \\
\sosrn{PAR1}{}
\sosrule
  {P \, \mtrans{a}_1  \amset{P} \quad
   Q \, \mtrans{a}_1  \amset{Q} \quad 
   a \notin A}
  {P \prlA Q \  \mtrans{a}_1 \ 
   ( \, \amset{P} \prlA \chut^{}_{\mkern-2mu Q} \, ) 
   \, + \,
   ( \, \chut^{}_{\mkern-2mu P} \prlA \! \amset{Q} \, )}
 \qquad
\sosrn{PAR2}{}
\sosrule
  {P \, \mtrans{a}_1  \amset{P} \quad 
   Q \, \mtrans{a}_1  \amset{Q} \quad 
   a \in A}
  {P \prlA Q \  \mtrans{a}_2 \ 
   \amset{P} \prlA \! \amset{Q}}
\bigskip \\
\sosrn{PAR3}{}
\sosrule
  {P \, \mtrans{\surd}_2 \, \mycal{P} \quad 
   Q \, \mtrans{\surd}_2 \, \mycal{Q} }
  {P \prlA Q \  \mtrans{\surd}_2 \ 
   \mycal{P} \myboxed{\,\prlA\,} \mycal{Q}}
\bigskip \\
\sosrn{CON1}{}
\sosrule
  {P \, \mtrans{\alpha}_1  \amset{P} \quad 
   X \dfas P}
  {X \, \mtrans{\alpha}_1  \amset{P}}
\qquad
\sosrn{CON2}{}
\sosrule
  {P \, \mtrans{\alpha}_2 \, \mycal{P} \quad 
   X \dfas P}
  {X \, \mtrans{\alpha}_2 \, \mycal{P}}
\end{array}
$} 
\end{displaymath}
\caption{$\FuTS$ Transition Deduction System for $\TPC$.}
\label{fig-tpc-rules}
\end{figure}

\noindent
Also $\calStpc$ is a \emph{combined} $\FuTS$, having the two
state-to-function relations $\mtrans{}_1$ and~$\mtrans{}_2$.
Actions~$a \in \calA$ decorate~$\mtrans{}_1$, the special
symbol~$\surd$ decorates~$\mtrans{}_2$ (with a similar role
as~$\delta$ for~$\IML$). As for~$\mtrans{}_2$ the label is always the
same, we occasionally suppress it. Note rule~(APF3) and rule~(TPF1)
involve the null-functions of~$\textbf{2}^\nats$ and of~$\bools$,
respectively, to express that a process~$a.P$ does not trigger a delay
and a process~$(n).P$ does not execute an action. In
Figure~\ref{fig-tpc-rules} and in the rest of this section we use
$\amsetP, \amsetQ \in \fsfn{\prcTPC}{\bools}$ as typical action
continuations, and $\mycalP, \mycalQ \in
\fsfn{\prcTPC}{\textbf{2}^\nats}$ as typical time continuations.

The second time prefix rule {(TPF2)} combines a possible evolution
over time of the process~$P$ into its continuation~$\mycalP$ with the
elapse of the prefix. Note, the continuation in the conclusion of rule
{(TPF2)} is a sum of three parts, viz.\ $[n;P]$, $[ \, P \mapsto
\SET{n} \,]$, and~$(n + \mycalP)$. The auxiliary mappings $[n;P]$ and
$(n + \mycalP)$, for timed continuations, are given by
\begin{displaymath}
  [n;P](Q) =
  \left \lbrace
    \begin{array}{@{\mkern1mu}cl}
      \SET{m} & \text{if $Q = (n-m).P$, $0<m<n$} \\
      \emptyset & \text{otherwise}
    \end{array}
  \right .
  \quad\ 
  (n + \mycalP)(Q) =
  \ZSET{n+m}{m \in \mycalP(Q)}
\end{displaymath}
It is easy to see that, for $n \in \nats$, $Q \in \prcTPC$, and
$\mycalP \in \fsfn{\prcTPC}{\textbf{2}^\nats}$, $[n;Q] = [ \, (n-i).Q
\mapsto \SET{i} \, ]_{i=1}^{n-1}$, and if $\mycalP(Q)= \emptyset$,
then also $(n + \mycalP)(Q) = \emptyset$.  Time progress taking fewer
steps than~$n$ is covered by the continuation~$[n;P]$. For~$m$
strictly between $0$ and~$n$, after $m$~time steps there remains
$(n-m).P$ to be executed. After exactly $n$~time steps, $P$~is to be
executed, i.e. the component $[P \mapsto \SET{n}]$ is used). After
more than~$n$ time steps, say $n+m$~time steps, process~$Q$ is to be
executed if $m \in \mycalP(Q)$. Thus, if no such~$m$ exist, i.e.\ if
$\mycalP(Q) = \emptyset$, this yields an empty set too.

The rules for the choice and parallel construct of~$\TPC$ make use of
corresponding operations on $\fsfn{\prcTPC}{\bools}$ and
$\fsfn{\prcTPC}{\textbf{2}^\nats}$. For $\amsetP, \amsetQ \in
\fsfn{\prcTPC}{\bools}$, the functions $\amsetP + \amsetQ$ and
$\amsetP \prlA \amsetQ$ are as before. For
$\fsfn{\prcTPC}{\textbf{2}^\nats}$ the following operators are used:
\begin{displaymath}
  (\mycalP \myboxed{+} \mycalQ)(R) =
  \left \lbrace
    \begin{array}{cl}
      \mycalP(P) \cap \mycalQ(Q) & \text{if $R = P+Q$ for $P, Q \in
        \prcTPC$} \\  
      \emptyset & \text{otherwise}
    \end{array}
  \right .  
\end{displaymath}
and, likewise
\begin{displaymath}
  (\mycalP \myboxed{\, \prlA \,} \mycalQ)(R) =
  \left \lbrace
    \begin{array}{cl}
      \mycalP(P) \cap \mycalQ(Q) & \text{if $R = P \prlA Q$, for $P,Q \in
        \prcTPC$} \\ 
      \emptyset & \text{otherwise}
    \end{array}
  \right .
\end{displaymath}
We have that for $P \in \prcTPC$ there exists a unique $\mycalP \in
\fsfn{\prcTPC}{\textbf{2}^\nats}$ such that $P \mtrans{}_2
\mycalP$. Moreover, given the rules for~$\calStpc$ and the definition
of the operators above, it can verified that, for $P, Q \in \prcTPC$
and $\mycalP \in \fsfn{\prcTPC}{\textbf{2}^\nats}$ such that $P
\mtrans{}_2 \mycalP$ it holds that $\mycalP(Q)$ is either a singleton
or the empty set. 
See Lemma~\ref{lm-tpc-total-det} below.

In order to prove the lemma we introduce an auxiliary function
$\mdelay : \prcTPC \to \nats$, establishing the so-called maximum
delay of a process, given by
\begin{displaymath}
  \begin{array}{rclcrcl}
    \mdelay(\nil) & = & 0 & \quad & 
    \mdelay(P_1 + P_2) & = & \min \lbrace \, \mdelay(P_1) ,\,
    \mdelay(P_2) \, \rbrace \\
    \mdelay(a.P) & = & 0 && 
    \mdelay(P_1 \prlA P_2) & = & \min \lbrace \, \mdelay(P_1) ,\,
    \mdelay(P_2) \, \rbrace \\
    \mdelay((n).P) & = & n + \mdelay(P) &&
    \mdelay((X) & = & \mdelay(P) \quad \text{if $X \dfas P$}
  \end{array}
\end{displaymath}
By guarded induction, one straightforwardly verifies the property that
$\mdelay(Q') < \mdelay(Q)$ for $Q,Q' \in \prcTPC$ and $\mycalQ \in
\fsfn{\prcTPC}{\nats}$ such that $Q \mtrans{}_2 \mycalQ$ and
$\mycalQ(Q') \neq \emptyset$. From this observation is follows that
$[n;P]$, $[P \mapsto \SET{n}]$ and $(n + \mycalP)$ have disjoint
supports: We have that (i)~if $[n;P](P') \neq \emptyset$ then $P' =
(n-m).P$ for $0 < m < n$, hence $\mdelay(P') = (m-n) + \mdelay(P) >
\mdelay(P)$; (ii)~if $[ \, P \mapsto \SET{n} \, ](P') \neq \emptyset$
then $P' = P$, hence $\mdelay(P') = \mdelay(P)$; (iii)~if $(n +
\mycalP)(P') \neq \emptyset$ then $\mycalP(P') \neq \emptyset$ hence,
using the property above, $\mdelay(P') < \mdelay(P)$.


%

\begin{lem}
  \label{lm-tpc-total-det}
  \hfill
  \begin{enumerate}[label=\({\alph*}]
  \item  The $\FuTS$ $\calStpc$ is total and deterministic.
  \item  If $P \mtrans{\surd}_2 \mycalP$ then either $\mycalP(Q)
    = \SET{n}$ for some $n > 0$ or $\mycalP(Q) = \emptyset$.
  \end{enumerate}
\end{lem}

\begin{proof}
  Part~\(a\ goes by guarded induction on~$P$, both for $\mtrans{}_1$
  and~$\mtrans{}_2$. 
  Part~\(b\ follows by guarded induction. For the time prefix~$(n).P$
  we use that $[n;P]$, $[P \mapsto \SET{n}]$ and $(n + \mycalP)$ have
  disjoint supports, as noted above. For the constructs $P \cho Q$
  and~$P \prlA Q$ we observe that the operations $\myboxed{+}$
  and~$\myboxed{\, \prlA \,}$ preserve the property mentioned, as the
  intersection of two singletons holding a positive number is either a
  singleton with a positive number or the empty set.
\end{proof}

\noindent
Below we have $\calStpc = \threetuple {\prcTPC}{\theta_1}{\theta_2}$
and use $\FuTSbis{\tpc}$ to denote the bisimulation equivalence
induced by~$\calStpc$.

\blankline

\noindent
The standard SOS semantics of the $\TPC$ fragment of interest is given
in Figure~\ref{fig-standard-tpc}, involving the transition relations
\begin{displaymath}
  {\trans{}} \: \subseteq \: { \prcTPC \times \calA \times \prcTPC } 
  \qquad \text{and} \qquad 
  {\Ttrans{}} \: \subseteq \: { \prcTPC \times \nats_{{>}0} \times
    \prcTPC }
\end{displaymath}
Note that for timed transitions $P \Ttrans{n} P'$ it is
required that $n > 0$. Therefore, regarding rule (DEC), a
process~$(n).P$ for example with a timed prefix will \emph{not} yield
a zero-time step $(n).P \Ttrans{0} (n).P$ for which time does not
progress. The case for $(n).P$ where $n$~time step elapse, is covered
by rule~(PRE).

\begin{figure}
\begin{displaymath}
\scalebox{0.825}{$
\begin{array}{c}

\sosrn{APF}{}
\sosrule
  {}
  {\pfx{a}{P} \, \trans{a} \, P}

\bigskip\\

\sosrn{PRE}{}
\sosrule{}
  {\pfx{(n)}{P} \, \Ttrans{n} \, P}

\qquad
  
\sosrn{DEC}{}
\sosrule{n=m+\ell}
  {\pfx{(n)}{P} \, \Ttrans{m} \, (\ell).P}

\qquad
  
\sosrn{SUM}{}
\sosrule{P \Ttrans{n} P'}
  {\pfx{(m)}{P} \, \Ttrans{n+m} \, P'}

\bigskip \\

\sosrn{CHO1}{}
\sosrule
  {P \, \trans{a} \, R}
  {P \cho Q  \, \trans{a} \, R}

\qquad

\sosrn{CHO2}{}
\sosrule
  {Q \, \trans{a} \, R}
  {P \cho Q  \, \trans{a} \, R}

\qquad
 
\sosrn{ALT}{}
\sosrule
  {P \, \Ttrans{n} \, P' \quad Q \, \Ttrans{n} \, Q'}
  {P \cho Q  \, \Ttrans{n} \, P' + Q'}

\displaybreak[3] 
\bigskip \\

\sosrn{PAR1a}{}
\sosrule
  {P \, \trans{a} \, P' \quad a \notin A}
  {P \prlA Q \, \trans{a} \, P' \prlA Q}

\qquad

\sosrn{PAR1b}{}
\sosrule
  {Q \, \trans{a} \, Q' \quad a \notin A}
  {P \prlA Q \, \trans{a} \, P \prlA Q'}

\bigskip \\

\sosrn{PAR2}{}
\sosrule
  {P \, \trans{a} \, P' \quad 
   Q \, \trans{a} \, Q' \quad a \in A}
  {P \prlA Q \, \trans{a} \, P' \prlA Q'}

\qquad

\sosrn{SYN}{}
\sosrule
  {P \, \Ttrans{n} \, P' \quad Q \, \Ttrans{n} \, Q'}
  {P \prlA Q \, \Ttrans{n} \, P' \prlA Q'}

\bigskip\\

\sosrn{CON1}{}
\sosrule
  {P \, \trans{a} \, Q \quad X \dfas P}
  {X \, \trans{a} \, Q}

\qquad

\sosrn{CON2}{}
\sosrule
  {P \,  \Ttrans{n} \, Q \quad X \dfas P}
  {X \, \Ttrans{n} \, Q}
  
\end{array}
$} 
\end{displaymath}
\caption{Standard Transition Deduction System for $\TPC$.}
\label{fig-standard-tpc}
\end{figure}

The definition of timed bisimilarity for $\TPC$ we give below is a bit
more concise than the one originally introduced
in~\cite{ABC10:springer}, but the two notions can be easily proven to
coincide. We will compare timed bisimilarity with the notion of
bisimulation associated with the combined $\FuTS$~$\calStpc$.

\begin{defi}
  \label{df-tpc-timed-bisimulation}
  An equivalence relation $R \subseteq \prcTPC \times \prcTPC$ is a
  timed bisimulation for~$\TPC$ if, for all $P_1, P_2 \in \prcTPC$
  such that $R(P_1,P_2)$, it holds that for all $a \in \calA$ and
  $n\in \nats$
  \begin{itemize}
  \item whenever $P_1 \trans{a} Q_1$, then $P_2 \trans{a} Q_2$
    for some $Q_2 \in \prcTPC$ with
    $R(Q_1,Q_2)$; 
  \item whenever $P_1 \Ttrans{n} Q_1$, then $P_2 \Ttrans{n} Q_2$ for
    some $Q_2 \in \prcTPC$ with $R(Q_1,Q_2)$.
  \end{itemize}
  Two processes $P_1, P_2 \in \prcTPC$ are called timed bisimilar,
  notation $P_1\TPCsbis P_2$ if $R(P_1,P_2)$ for some timed
  bisimulation for~$\prcTPC$.
\end{defi}

\noindent
To establish the correspondence of $\FuTS$
bisimilarity~$\FuTSbis{\tpc}$ for~$\calStpc$ of
Definition~\ref{df-tpc-sem} and timed bisimilarity~$\TPCsbis$
for~$\TPC$ of Definition~\ref{df-tpc-timed-bisimulation}, we need to
connect the state-to-function relation~$\mtrans{}_1$ and the
transition relation~$\trans{}$ as well as the state-to-function
relation~$\mtrans{}_2$ and the transition relation~$\Ttrans{} \,$.
The connection is established by Lemma~\ref{lm-mtrans-Ttrans}. First
we state an auxiliary result, which is commonly referred to as
time-determinism (cf.~\cite{BM02:springer}) and which can be shown
straightforwardly by guarded induction.

\begin{lem}
  \label{lm-time-determinism}
  If $P \Ttrans{n} P'$ and $P \Ttrans{n} P''$, for $P, P', P'' \in
  \prcTPC$ and $n > 0$, then~$P' = P''$.
  \qed
\end{lem}

\noindent
We use time-determinism of~$\TPC$ in the proof of the following lemma.

\begin{lem}
\label{lm-mtrans-Ttrans}
\hfill
\begin{enumerate}[label=\({\alph*}]
\item  Let $P \in \prcTPC$ and $a \in \calA$.  If $P \mtrans{a}_1
  \! \amsetP$ then $P \trans{a} P' \iff \amsetP( P' ) = \TRUE$.
\item  Let $P \in \prcTPC$.  If $P \mtrans{\surd}_2 \mycalP$
  then $P \Ttrans{n} P' \iff \mycalP( P' ) = \SET{n}$. 
\end{enumerate}
\end{lem} 

\begin{proof}
  Part~\(a\ is similar to the corresponding part of
  Lemma~\ref{lm-mtrans-trans}. Part~\(b\ can be shown by guarded
  induction for which we exhibit two cases (the others being similar or
  straightforward). For readability, we suppress the label~$\surd$
  of~$\mtrans{}_2$.

  Case~$(m).P$. Suppose $(m).P \mtrans{}_2 \mycalP$ and $P \mtrans{}_2
  \mycalP'$. Then, by (TPF2), we have $\mycalP(P') = \SET{\ell}$, for
  $0 < \ell < 
  m$, iff $P' = (m-\ell).P$, $\mycalP(P') = \SET{m}$ iff $P' = P$, and
  $\mycalP(P') = \SET{\ell}$, for $\ell > m$ iff $\mycalP'(P') =
  \SET{\ell-m}$.
  Now, if $(m).P \Ttrans{n} P'$ for $0 < n < m$, then $P' = (m-n).P$,
  because of rules (PRE) and~(DEC) and
  Lemma~\ref{lm-time-determinism}. Therefore, $\mycalP(P') =
  \mycalP((m-n).P) = \SET{n}$. If $(m).P \Ttrans{n} P'$ with $n = m$,
  then $P' = P$, as (PRE) applies (and with an appeal to
  Lemma~\ref{lm-time-determinism}). Therefore, $\mycalP(P') = \mycalP(P) =
  \SET{m} = \SET{n}$. Finally, if $(m).P \Ttrans{n} P'$ for $n > m$,
  then we have $P \Ttrans{n-m} P'$, in view of rule~(SUM) and because
  of time-determinism. By induction hypothesis, we obtain
  $\mycalP'(P') = \SET{n-m}$ and therefore $\mycalP(P') =
  (m+\mycalP')(P') = \ZSET{ m + n }{n \in 
    \mycalP'(P')} = \SET{m + n - m} = \SET{n}$.
  Reversely, by rules (PRE) and~(DEC) we have $(m).P \Ttrans{\ell}
  (m-\ell).P$, for $0 < \ell < m$ and $(m).P \Ttrans{m} P$. Moreover,
  if $\mycalP(P') = \SET{\ell}$, for $\ell > m$, then $\mycalP'(P') =
  \SET{\ell - m}$. By induction hypothesis, $P \Ttrans{\ell - m}
  P'$. Hence, $(m).P \Ttrans{m+\ell-m} P'$, i.e.\ $(m).P \Ttrans{\ell}
  P'$, by~(SUM).
  
  Case $P_1 + P_2$. Suppose $P_1 + P_2 \mtrans{}_2 \mycalP$. Then
  $\mycalP = \mycalP_1 \myboxed{+} \mycalP_2$ for $\mycalP_1,
  \mycalP_2 \in \fsfn{\prcTPC}{\textbf{2}^\nats}$ such that $P_1
  \mtrans{}_2 \mycalP_1$ and $P_2 \mtrans{}_2 \mycalP_2$.
  If $P_1 + P_2 \Ttrans{n} P'$, then exist $P'_1, P'_2 \in \prcTPC$
  such that $P_1 \Ttrans{n} P'_1$, $P_2 \Ttrans{n} P'_2$ and $P' =
  P'_1 + P'_2$, because (ALT) is the only rule applicable. By
  induction hypothesis, $\mycalP_1(P'_1) = \SET{n}$ and
  $\mycalP_2(P'_2) = \SET{n}$. Hence $\mycalP(P') = ( \mycalP_1
  \myboxed{+} \mycalP_2 \mkern1mu )(P'_1 + P'_2) = \SET{n}$.
  In the other direction, if $\mycalP(P') = \SET{n}$, then $P' = P'_1
  + P'_2$ for processes $P'_1, P'_2 \in \prcTPC$ such that
  $\mycalP_1(P'_1) = \SET{n}$ and $\mycalP_2(P'_2) = \SET{n}$. By
  induction hypothesis, $P_1 \Ttrans{n} P'_1$ and $P_2 \Ttrans{n}
  P'_2$, from which it follows that $P_1 + P_2 \Ttrans{n} P'_1 +
  P'_2$, i.e.\ $P_1 + P_2 \Ttrans{n} P'$, by~(SUM).
\end{proof}

\noindent
With Lemma~\ref{lm-mtrans-Ttrans} in place we are ready to show the
correspondence of $\FuTS$ bisimilarity and timed bisimilarity
for~$\TPC$.

\begin{thm}
\label{th-correpondence-tpc}
  For any two processes $P_1, P_2 \in \prcTPC$ it holds that 
  $P_1  \FuTSbis{\tpc} P_2$ iff $P_1 \TPCsbis P_2$.
\end{thm} 

\begin{proof}
  Suppose $P_1 \FuTSbis{\tpc} P_2$, for $P_1, P_2 \in \prcTPC$. Let $R
  \subseteq \prcTPC \times \prcTPC$ be a bisimulation with respect
  to~$\calStpc$ such that~$R(P_1,P_2)$. We verify that $R$ meets the
  two transfer conditions of Definition~\ref{df-tpc-timed-bisimulation}.

  If $P_1 \trans{a} Q_1$, for some $a \in \calA$ and $Q_1 \in
  \prcTPC$, then $\theta_1(P_1)(a)(Q_1) = \TRUE$ by
  Lemma~\ref{lm-mtrans-Ttrans}. From the definition of a $\FuTS$
  bisimulation we obtain
  \begin{equation}
    \label{eq-tpc-bisim-non-det}
    \tssum_{Q' \myin \Rclass{Q}} \; \theta_1(P_1)(a)(Q')
    = 
    \tssum_{Q' \myin \Rclass{Q}} \; \theta_1(P_2)(a)(Q')
  \end{equation}
  for all $Q \in \prcTPC$.  As we have seen before, we argue that
  summation of~$\bools$ is disjunction, and since
  $\theta_1(P_1)(a)(Q_1) = \TRUE$, there must exist $Q_2 \in
  \Rclass{Q_1}$ such that $\theta_1(P_2)(Q_2) =
  \TRUE$. Hence, $R(Q_1,Q_2)$ and, by Lemma~\ref{lm-mtrans-Ttrans}, $P_2
  \trans{a} Q_2$.

  If $P_1 \Ttrans{n} Q_1$, for some~$n > 0$, then, by
  Lemma~\ref{lm-mtrans-Ttrans}, $\theta_2(P_1)(\surd)(Q_1) =
  \SET{n}$. From the definition of   $\FuTS$
  bisimulation we obtain
  \begin{equation}
    \label{eq-tpc-bisim-timed}
    \tssum_{Q' \myin \Rclass{Q}} \; \theta_2(P_1)(\surd)(Q')
    = 
    \tssum_{Q' \myin \Rclass{Q}} \; \theta_2(P_2)(\surd)(Q')
  \end{equation}
  for all $Q \in \prcTPC$. Note, summation of the
  semiring~$\textbf{2}^{\nats}$ is union of sets. So, by picking $Q =
  Q_1$ we have $n \in \tssum_{Q' \myin \Rclass{Q_1}} \;
  \theta_2(P_2)(\surd)(Q')$. Thus, for some $Q_2 \in \prcTPC$ with
  $R(Q_1,Q_2)$ it holds that $n \in \theta_2(P_2)(\surd)(Q_2)$.  It
  follows from Lemma~\ref{lm-tpc-total-det}b that
  $\theta_2(P_2)(\surd)(Q_2) = \SET{n}$, and thus, again by
  Lemma~\ref{lm-mtrans-Ttrans}, $P_2 \Ttrans{n} Q_2$.

  Now suppose $P_1 \TPCsbis P_2$, for $P_1, P_2 \in \prcTPC$. Let $R
  \subseteq \prcTPC \times \prcTPC$ be a timed bisimulation such
  that~$R(P_1,P_2)$. We verify that, with respect to $P_1$
  and~$P_2$, $R$ meets the two summation conditions of
  Definition~\ref{df-gltfs-bisim} for the case of~$\calStpc$, i.e.,
  equations (\ref{eq-tpc-bisim-non-det})
  and~(\ref{eq-tpc-bisim-timed}), for all $Q \in \prcTPC$ and $a \in
  \calA$. We have
  \begin{displaymath} 
    \def\arraystretch{1.2}
    \begin{array}{rcll}
      \multicolumn{4}{l}{\tssum_{Q' \myin \Rclass{Q}} \;
        \theta_1(P_1)(a)(Q')} \\
      & \Leftrightarrow &
      \exists \mkern1mu Q' \in \prcTPC \colon
      R(Q' \mkern-2mu ,Q) \land \theta_1(P_1)(a)(Q') = \TRUE
      & \text{by structure of $\bools$} \\
      & \Leftrightarrow &
      \exists \mkern1mu Q' \in \prcTPC \colon
      R(Q' \mkern-2mu ,Q) \land P_1 \trans{a} Q'
      & \text{by Lemma~\ref{lm-mtrans-Ttrans}} \\
      & \Leftrightarrow &
      \exists \mkern1mu Q'' \in \prcTPC \colon
      R(Q'' \mkern-2mu ,Q) \land P_2 \trans{a} Q''
      & \text{$R(P_1,P_2)$ and $R$ timed bisimulation} \\
      & \Leftrightarrow &
      \exists \mkern1mu Q'' \in \prcTPC \colon
      R(Q'' \mkern-2mu ,Q) \land \theta_1(P_2)(a)(Q'') = \TRUE
      & \text{by Lemma~\ref{lm-mtrans-Ttrans}} \\
      & \Leftrightarrow &
      \tssum_{Q'' \myin \Rclass{Q}} \; \theta_1(P_2)(a)(Q'')
      & \text{by structure of $\bools$} \\
   \end{array}
    \def\arraystretch{1.0}
  \end{displaymath}
  and also
  \begin{displaymath} 
    \def\arraystretch{1.2}
    \begin{array}{rcll}
      \multicolumn{4}{l}{n \in \tssum_{Q' \myin \Rclass{Q}} \;
        \theta_2(P_1)(\surd)(Q')} \\
      & \Leftrightarrow &
      \exists \mkern1mu Q' \in \prcTPC \colon
      R(Q' \mkern-2mu ,Q) \land {n \in \theta_2(P_1)(\surd)(Q')}
      & \text{by structure of $\textbf{2}^\nats$} \\
      & \Leftrightarrow &
      \exists \mkern1mu Q' \in \prcTPC \colon
      R(Q' \mkern-2mu ,Q) \land {\theta_2(P_1)(\surd)(Q') = \SET{n}}
      & \text{by Lemma~\ref{lm-tpc-total-det}} \\
      & \Leftrightarrow &
      \exists \mkern1mu Q' \in \prcTPC \colon
      R(Q' \mkern-2mu ,Q) \land P_1 \Ttrans{n} Q'
      & \text{by Lemma~\ref{lm-mtrans-Ttrans}} \\
      & \Leftrightarrow &
      \exists \mkern1mu Q'' \in \prcTPC \colon
      R(Q'' \mkern-2mu ,Q) \land P_2 \Ttrans{n} Q''
      & \text{$R(P_1,P_2)$ and $R$ timed bisimulation} \\
      & \Leftrightarrow &
      \exists \mkern1mu Q'' \in \prcTPC \colon
      R(Q'' \mkern-2mu ,Q) \land \theta_2(P_2)(\surd)(Q'') = \SET{n}
      & \text{by Lemma~\ref{lm-mtrans-Ttrans}} \\
      & \Leftrightarrow &
      \exists \mkern1mu Q'' \in \prcTPC \colon
      R(Q'' \mkern-2mu ,Q) \land {n \in \theta_2(P_2)(\surd)(Q'')}
      & \text{by Lemma~\ref{lm-tpc-total-det}} \\
      & \Leftrightarrow &
      \tssum_{Q'' \myin \Rclass{Q}} \; \theta_2(P_2)(\surd)(Q'')
      & \text{by structure of $\textbf{2}^\nats$} \\
   \end{array}
    \def\arraystretch{1.0}
  \end{displaymath}
  Thus, $R$ satisfies the conditions for a $\FuTS$ bisimulation
  for~$\calStpc$.
\end{proof}

\noindent
We conclude that also in the setting of a $\FuTS$ for discrete time
involving the semiring~$\textbf{2}^\nats$, we have an example of a
correspondence result of $\FuTS$-bisimilarity and bisimilarity based
on a standard SOS definition.  It is worth pointing out that in the
proof above, the equivalence of $n \in \tssum_{Q' \myin \Rclass{Q}} \;
\theta_2(P_1)(\surd)(Q')$ and $\exists \mkern1mu Q' \in \prcTPC \colon
R(Q' \mkern-2mu ,Q) \land {n \in \theta_2(P_1)(\surd)(Q')}$, holds
because we are working with a semiring of (finite) \emph{sets}
over~$\nats$ with summation to be interpreted as (finite) union. Was
summation to be interpreted as sum over~$\nats$, as it would have been
the case if we would have used the semiring~$\nats$, i.e.\ using
$\fsfn{\prcTPC}{\nats}$ instead of $\fsfn{\prcTPC}{\textbf{2}^\nats}$,
then, from $ n = \tssum_{Q' \myin \Rclass{Q}} \;
\theta_2(P_1)(\surd)(Q')$ we would not have been able to conclude
$\exists \mkern1mu Q' \in \prcTPC \colon R(Q' \mkern-2mu ,Q) \land {n
  = \theta_2(P_1)(\surd)(Q')} $, and vice-versa.

\section{Nested \texorpdfstring{$\FuTS$}{FuTS}}
\label{sec-nested}

In this section we extend the applicability of the $\FuTS$ framework
to more complex models, in particular those in which different aspects
of behaviour are integrated in a non-orthogonal way---as it is the
case for non-deterministic choice of probabilistic distributions over
behaviour in probabilistic and Markov automata. We introduce the
notion of a \emph{nested} $\FuTS$, namely a $\FuTS$ where the
transition relation involves continuation functions that do not act on
the set of states~$S$ directly, but instead on {\em functions} acting
on~$S$ or, in the general case, on functions over the latter and so
on. As mentioned in the introduction, here we restrict our
investigation on nested $\FuTS$s with two levels, namely nested
$\FuTS$s where the domain of the continuation functions is a set of
functions the domain of which is the set~$S$ of states.
In the following, we give the formal
definition of a {\em simple} two-level nested $\FuTS$, i.e.\ a nested
$\FuTS$ involving two levels of continuations that has a single
transition relation.

\begin{defi}
  \label{df-2-FuTS}
  Let $\calL$ be a set of labels and $\calR_1$ and $\calR_2$ be two
  semirings. A (simple) two-level nested $\FuTS$~$\calS$, over $\calL$
  and $\calR_1$ and~$\calR_2$ is a tuple $\calS = ( \, S ,\, \mtrans{}
  \, )$ with set of states~$S$ and transition relation ${\mtrans{}} \,
  \subseteq \, {S \times \calL \times \fsfn{\, \fsfn{\mkern2mu
        S}{\calR_1} \,}{\calR_2}}$.
\end{defi}

\noindent
A two-level nested $\FuTS$ is called total and deterministic if, for
all $s \in S$ and $\ell \in \calL$, there exists exactly one $\psi \in
\fsfn{\, \fsfn{\mkern2mu S}{\calR_1} \,}{\calR_2}$ such that $s
\mtrans{\, \ell} \psi$.  As before, for a total and deterministic
nested $\FuTS$ we use the notation $(S, \theta \mkern1mu )$ where the
function~$\theta$ has type $S \to \calL \to \fsfn{\, \fsfn{\mkern2mu
    S}{\calR_1} \,}{\calR_2}$.  Here, for $s \in S$, $\ell \in \calL$,
$\varphi \in \fsfn{\mkern2mu S}{\calR_1}$, $y \in \calR_2$, we have
$\theta(s)(\myell)(\varphi) = y$ iff $\psi(\varphi) = y$ for the
unique $\psi \in \fsfn{\, \fsfn{\mkern2mu S}{\calR_1} \,}{\calR_2}$
such that $s \mtrans{\, \ell} \psi$.

For a set of states~$S$ and a semiring~$\calR$, an equivalence
relation~$R$ on~$S$ induces an equivalence relation
on~$\fsfn{\mkern2mu S}{\calR}$, referred to as the lifting of~$R$
to~$\fsfn{\mkern2mu S}{\calR}$ and also denoted as~$R$. The induced
relation~$R$ is defined by
\begin{displaymath}
  R( \mkern1mu \varphi_1  , \mkern1mu \varphi_2 \mkern1mu )
  \quad \text{iff } \quad
  {\tssum_{t' \in \Rclass{t}} \; \varphi_1 (t')}
  \ = 
  {\tssum_{t' \in \Rclass{t}} \; \varphi_2(t')}
  \quad \text{for all $t \in S$}
\end{displaymath}
for $\varphi_1, \varphi_2 \in \fsfn{\mkern2mu S}{\calR}$.  It is easy
to see that $R$ on~$\fsfn{\mkern2mu S}{\calR}$ is indeed an
equivalence relation. Therefore, the notion of a two-level
bisimulation for a two-level nested $\FuTS$ given below is
well-defined.

\begin{defi}
  \label{df-2-FuTS-bis}
  Let $\calS = ( \, S ,\, \mtrans{} \, )$ be a two-level nested
  $\FuTS$ over the label set~$\calL$ and semirings $\calR_1$
  and~$\calR_2$. An equivalence relation $R \subseteq {S \times S}$ is
  a two-level bisimulation for~$\calS$ if and only if $R(s_1,s_2)$
  implies
  \begin{equation}
    \tssum_{\varphi' \in \lftRclass{\varphi}} \; \theta \mkern2mu (s_1)(\ell
    \mkern2mu )(\varphi' \mkern1mu ) 
    =
    \tssum_{\varphi' \in \lftRclass{\varphi}} \; \theta \mkern2mu (s_2)(\ell
    \mkern2mu )(\varphi' \mkern1mu ) 
    \raisebox{-8pt}{\rule{0pt}{12pt}}
    \label{eq-two-ltfs-bisim} 
  \end{equation}%
  for all $\ell \in \calL$ and $\varphi \in \fsfn{\mkern2mu
    S}{\calR_1}$.  Two elements $s_1, s_2 \in S$ are called bisimilar
  for~$\calS$ if $R(s_1,s_2)$ for some two-level bisimulation~$R$
  for~$\calS$.  Notation $s_1 \tFuTSbis{\calS} s_2$.
\end{defi}

\noindent
In Section~\ref{sec-mal} we will show that, in the setting of Markov
Automata, the notion of a two-level bisimulation for a suitable
two-level nested $\FuTS$ (having $\calR_1=\nnreals$ and
$\calR_2=\bools$) coincides with the notion of strong bisimulation for
Markov Automata.

As is to be expected, a total and deterministic two-level $\FuTS$ can
be considered as a coalgebra of a suitable functor on~sets.

\begin{defi}
  \label{df-2-futs-functor}
  Let $\calL$ be a label set and $\vcalR = \langle \mkern2mu \calR_1
  ,\, \calR_2 \mkern2mu \rangle$ be an pair of semirings. The functor
  $\twocalVvcalLR : \Set \to \Set$ assigns to a set~$X$ the function
  space $\fsfn{\mkern2mu \fsfn{\mkern2mu X}{\calR_1}}{\calR_2}^{\,
    \calL}$ of all functions $\psi : \calL \to \fsfn{\mkern2mu
    \fsfn{\mkern2mu X}{\calR_1}}{\calR_2}$ and assigns to a mapping
  $f: X \to Y$ the mapping $\twocalVvcalLR(f) : \fsfn{\mkern2mu
    \fsfn{\mkern2mu X}{\calR_1}}{\calR_2}^{\, \calL} \to
  \fsfn{\mkern2mu \fsfn{\mkern2mu Y}{\calR_1}}{\calR_2}^{\, \calL}$
  where 
  \begin{displaymath}
    \twocalVvcalLR(f \mkern1mu)(\Phi)(\myell)(\psi)= \tssum_{\,
    \varphi \myin \fsfn{f}{\calR_1}^{-1}(\psi)} \; \Phi(\myell)(\varphi) 
  \end{displaymath}
  for all $\Phi : \calL \to \fsfn{\mkern2mu \fsfn{\mkern2mu
      X}{\calR_1}}{\calR_2}$, $\ell \in \calL$, $\psi \in
  \fsfn{\mkern2mu Y}{\calR_1}$, where we use the function
  $\fsfn{f}{\calR_1} :\fsfn{X}{\calR_1} \rightarrow \fsfn{Y}{\calR_1}$
  with $\fsfn{f}{\calR_1}(\varphi)(y) = \tssum_{x \in f^{-1}(y)} \:
  \varphi(x)$ for $\varphi \in \fsfn{X}{\calR_1}$ and $y\in Y$.
\end{defi}

\noindent
Note that in the definition above the sums exist since
$\Phi$ and $\varphi$ have finite support.

For readability we use $\calW$ as shorthand for~$\twocalVvcalLR$, when
the label set~$\calL$ and the pair of semirings~$R$ are clear from the
context. It is readily checked that each $\calW$ is a functor, in fact
an accessible one being a composition of accessible functors. Thus,
$\calW$ possesses a final coalgebra. The associated notion of
behavioural equivalence is denoted by~$\approx_{\calW}$. As before, we
have for nested $\FuTS$ a correspondence result as well.


\begin{thm}
  \label{th-nested-correspondence}
  Let $\calS = ( \, S ,\, \theta \, )$ be a two-level nested $\FuTS$
  over the label set~$\calL$ and the semirings $\calR_1$
  and~$\calR_2$. Let the functor~$\calW$ be as in
  Definition~\ref{df-2-futs-functor}. Then $s_1 \FuTSbis{\calS} s_2
  \IFF s_1 \beq{\calW}{} s_2$, for all $s_1, s_2 \in S$.
\end{thm}

\begin{proof}
  Let $s_1, s_2 \in S$.  We first prove ${s_1 \FuTSbis{\calS} s_2} \,
  \IMPL \, {s_1 \beq{\calW}{} s_2}$.  So, assume $s_1 \FuTSbis{\calS}
  s_2$.  Let $R \subseteq S \times S$ be a two-level bisimulation
  with $R(s_1,s_2)$.  We turn the collection of equivalence
  classes~$S/R$ into a $\calW$-coalgebra $\calS_R = ( S/R, \theta_R )$
  by putting
  \begin{displaymath}
    \theta_R ( \, \Rclass{s} \, )(\myell)( \,
    \varphibar \, ) = \textstyle{\sum}_{\, \varphi \myin \fsfn{
        \varepsilon_R }{\calR_1}^{\mkern1mu -1}( \mkern1mu \varphibar
      \mkern1mu )} \; 
    \theta(s)(\myell)(\varphi)
  \end{displaymath}
  for $s \in S$, $\ell \in \calL \mkern1mu$, and $\varphibar \in
  \fsfn{S/R}{\calR_1}$ and $\varepsilon : S \to S/R$ the canonical
  mapping. This is well-defined since $R$ is a two-level bisimulation
  and $\fsfn{ \varepsilon_R }{\calR_1}^{\mkern1mu -1}( \mkern1mu
  \varphibar \mkern1mu )$ is an equivalence class of~$R$, for all
  $\varphibar \in \fsfn{ \mkern1mu S/R \mkern1mu }{\calR_1}$. For, if
  $\varphi_1,\varphi_2 \in \fsfn{ \varepsilon_R }{\calR_1}^{\mkern1mu
    -1}( \mkern1mu \varphibar \mkern1mu )$, $t \in S$ then $\fsfn{
    \varepsilon_R }{\calR_1}(\varphi_1)( \Rclass{t} ) = \fsfn{
    \varepsilon_R }{\calR_1}(\varphi_2)( \Rclass{t} )$. Thus
  ${\sum_{t' \in \Rclass{t}} \: \varphi_1 (t')} \ = {\sum_{t' \in
      \Rclass{t}} \varphi_2(t')}$ for all $t \in S$,
  hence $R( \varphi_1, \varphi_2 )$.  Therefore, $\varepsilon_R : S
  \to S/R$ is a $\calW$-homomorphism: for $\ell \in \calL$ and
  $\varphibar \in \fsfn{S/R}{\calR_1}$, we have 
  \pagebreak[3]
  \begin{displaymath}
    \def\arraystretch{1.2}
    \begin{array}{rcll}
      \multicolumn{4}{l}{ \calW \mkern1mu
        ( \varepsilon_R )( \, \theta(s)\,)(\myell)( \mkern1mu
        \varphibar \mkern1mu )}
      \\ & = &
      \tssum_{\varphi \myin \fsfn{ \varepsilon_R }{\calR_1}^{-1}(
        \varphibar )} \; 
      \theta( s )( \myell )( \varphi ) 
      & \text{by definition of~$\calW$}
      \displaybreak[3]
      \\ & = &
      \theta_R \mkern1mu ( \, \Rclass{s} \, )(\myell)( \, \varphibar \, ) 
      & \text{by definition of~$\theta_R$}
      \\ & = &
      \theta_R \mkern1mu ( \, \varepsilon_R(s) \,
      )(\myell)( \, \varphibar \, ) 
      & \text{by definition of~$\varepsilon_R$}
   \end{array}
 \end{displaymath}%
 Thus, $\calW(\varepsilon_R) \compose \theta = \theta_R \compose
 \varepsilon_R$ and $\varepsilon_R : \calS \to \calS_R$ is a
 $\calW$-homomorphism as claimed. Now, by uniqueness of a final
 morphism, we have $\fmorph{\calW}{\calS}{\cdot} =
 \fmorph{\calW}{{\calS_R}}{\cdot} \compose \, \varepsilon_R$.  In
 particular, with respect to~$\calS$, this implies that
 $\fmorph{\calW}{}{s_1} = \fmorph{\calW}{}{s_2}$ since
 $\varepsilon_R(s_1) = \varepsilon_R(s_2)$.  Thus, $s_1 \beq{\calW}{}
 s_2$ as was to be shown.
 
 For the reverse, ${s_1 \beq{\calW}{} s_2 } \, \IMPL \, { s_1
   \FuTSbis{\calS} s_2 }$, assume $s_1 \beq{\calW}{} s_2$, i.e.
 $\fmorph{\calW}{}{s_1} = \fmorph{\calW}{}{s_2}$, for $s_1,s_2 \myin
 S$.  Since the map $\fmorph{\calW}{}{\cdot} : \twotuple{S}{\mytheta}
 \to \twotuple{\Omega}{\omega}$ is a $\calW$-homomorphism, the
 equivalence relation~$R_{\calS}$ given by $R_{\calS} \mkern1mu
 (s',s'') \IFF \fmorph{\calW}{}{s'} = \fmorph{\calW}{}{s''}$ is a
 two-level bisimulation: Suppose $R_{\calS} \mkern1mu (s',s'')$, i.e.\
 $s' \beq{\calW}{} s''$, for some $s',s'' \in S$. Pick $\ell \in
 \calL$, $t \in S$ and assume $\fmorph{\calW}{}{t} = w \in \Omega$.
 For $\calW$ we have $\omega \, \compose \, \fmorph{\calW}{}{\cdot} \,
 = \, \calW(\fmorph{\calW}{}{\cdot})\, \compose \, \theta$. Hence, for
 $s \in S$, $\ell \in \calL$, $\chi \in \fsfn{ \Omega }{\calR_1}$, it
 holds that
 \begin{equation}
   \omega \mkern1mu ( \,\fmorph{\calW}{}{s} \,)( \myell )( \chi ) 
   = 
   \calW( \, \fmorph{\calW}{}{\cdot} \, ) (\theta(s) )( \myell )( \chi ) 
   = 
   \tssum_{\varphi \myin \fsfn{ \, \fmorph{\calW}{}{\cdot} \, }{\calR_1}^{-1}(
     \chi )} \; \theta( s )( \myell )( \varphi ) 
   \label{eq-theta-OmegaS-nested}
 \end{equation}
 Moreover, we have, for $\varphi_1, \varphi_2 \in \fsfn{S}{\calR_1}$,
 that
 \begin{displaymath}
   R_S(  \,\varphi_1 , \, \varphi_2 \, ) 
   \iff
   \fsfn{ \, \fmorph{\calW}{}{\cdot} \,}{\calR_1}( \varphi_1 )
   =
   \fsfn{ \, \fmorph{\calW}{}{\cdot} \,}{\calR_1}( \varphi_2 )
 \end{displaymath}
 since we observe that
 \begin{displaymath}
   \def\arraystretch{1.2}
   \begin{array}{rcl}
     \multicolumn{3}{l}{\fsfn{ \, \fmorph{\calW}{}{\cdot}
         \,}{\calR_1}( \varphi_1 )  = \fsfn{ \,
         \fmorph{\calW}{}{\cdot} \,}{\calR_1}( \varphi_2 )}   
     \\ & \IFF &
     \forall w \in \fmorph{\calW}{}{S} \colon
     \fsfn{ \, \fmorph{\calW}{}{\cdot} \,}{\calR_1}( \varphi_1 )(w)  
     = 
     \fsfn{ \, \fmorph{\calW}{}{\cdot} \,}{\calR_1}( \varphi_2 )(w)
     \\ & &
     \text{since both $\fsfn{ \, \fmorph{\calW}{}{\cdot} \,}{\calR_1}(
       \varphi_1 )(w), \fsfn{ \, \fmorph{\calW}{}{\cdot} \,}{\calR_1}(
       \varphi_2 )(w) = 0$ if ${\fmorph{\calW}{}{\cdot}}^{\mkern1mu
         -1}(w) = \emptyset$}
     \\ & \IFF &
     \forall t \in S \colon
     \fsfn{ \, \fmorph{\calW}{}{\cdot} \,}{\calR_1}( \varphi_1 )(
     \fmorph{\calW}{}{t} )  
     = 
     \fsfn{ \, \fmorph{\calW}{}{\cdot} \,}{\calR_1}( \varphi_2 )(
     \fmorph{\calW}{}{t} )  
     \\ & \IFF &
     \forall t \in S \colon
     \tssum_{t' \in {\fmorph{\calW}{}{\cdot}}^{\mkern1mu
         -1}( \fmorph{\calW}{}{t} )} \; \varphi_1(t') 
     =
     \tssum_{t' \in {\fmorph{\calW}{}{\cdot}}^{\mkern1mu
         -1}( \fmorph{\calW}{}{t} )} \; \varphi_2(t') 
     \\ & & \text{by definition of $\fsfn{\cdot}{\calR_1}$}
     \\ & \IFF &
     \forall t \in S \colon
     \tssum_{t' \in \RSclass{t}} \; \varphi_1(t') 
     =
     \tssum_{t' \in \RSclass{t}} \; \varphi_2(t') 
     \\ & & \text{since $t' \in \RSclass{t}$ iff $\fmorph{\calW}{}{t'} =
       \fmorph{\calW}{}{t}$} 
     \\ & \IFF &
     R_\calS( \mkern1mu \varphi_1 , \mkern1mu \varphi_2 \mkern1mu )
     \\ & & \text{by definition of $R_\calS$ on $\fsfn{S}{\calR_1}$}
   \end{array}
   \def\arraystretch{1.0}
 \end{displaymath}
 Therefore,
 \begin{equation}
   \label{eq-RS-vs-final-image}
   \varphi' \in \RSclass{\varphi} 
   \iff
   \varphi' \in {\fsfn{ \, \fmorph{\calW}{}{\cdot}
       \,}{\calR_1}}^{\mkern1mu -1} ( \chi ) 
   \qquad
   \text{for $\chi = \fsfn{ \mkern1mu \fmorph{\calW}{}{\cdot} \mkern1mu
     }{\calR_1}(\varphi)$}
 \end{equation}
 Now, let $s', s'' \in S$ such that $R_\calS(s',s'')$, and choose any
 $\ell \in \calL$ and $\varphi \in \fsfn{S}{\calR_1}$. Put $\chi =
 \fsfn{ \mkern1mu \fmorph{\calW}{}{\cdot} \mkern1mu
 }{\calR_1}(\varphi)$. Then we have
 \begin{displaymath}
   \def\arraystretch{1.2}
   \begin{array}{rcll}
     \multicolumn{4}{l}{\tssum_{\varphi' \in \RSclass{\varphi}} \;
       \theta \mkern2mu (s')(\myell)(\varphi' \mkern1mu )}  
     \\ & = &
     \tssum_{\varphi' \myin  \, {\fsfn{ \, \fmorph{\calW}{}{\cdot} \, 
         }{\calR_1}^{\mkern1mu -1} ( \, \chi \, )} } \;
       \theta(s')(\myell)(\varphi')  
     & \text{by Equation~(\ref{eq-RS-vs-final-image} and definition~$\chi$)}
     \\ & = &
     \omega \mkern1mu ( \, \fmorph{\calW}{}{s'} \,)( \myell )( \,
     \chi \, ) 
     & \text{by Equation~(\ref{eq-theta-OmegaS-nested})}
     \\ & = &
     \omega \mkern1mu ( \, \fmorph{\calW}{}{s''} \,)( \myell )( \,
     \chi \, ) 
     & \text{$s' \beq{\calW}{} s''$ by assumption}
     \\ & = &
     \tssum_{\varphi' \myin  \, {\fsfn{ \, \fmorph{\calW}{}{\cdot} \, 
         }{\calR_1}^{\mkern1mu -1}( \, \chi \, )} } \;
       \theta(s'')(\myell)(\varphi')  
     & \text{by Equation~(\ref{eq-theta-OmegaS-nested})}
     \\ & = &
     \tssum_{\varphi' \in \RSclass{\varphi}} \;
       \theta \mkern2mu (s'')(\myell)(\varphi' \mkern1mu )
     & \text{by Equation~(\ref{eq-RS-vs-final-image} and definition~$\chi$)}
   \end{array}
   \def\arraystretch{1.0}
 \end{displaymath}
 Thus, if $R_{\calS} \mkern1mu (s',s'')$ then $\tssum_{\varphi' \in
   \RSclass{\varphi}} \; \theta \mkern2mu (s')(\myell)(\varphi'
 \mkern1mu ) = \tssum_{\varphi' \in \RSclass{\varphi}} \; \theta
 \mkern2mu (s'')(\myell)(\varphi' \mkern1mu )$ for all $\ell \in
 \calL$ and $\varphi \in \fsfn{S}{\calR_1}$. Therefore, $R_{\calS}$ is
 a two-level bisimulation according to Definition~\ref{df-2-FuTS-bis}.
 Since $\fmorph{\calW}{}{s_1} = \fmorph{\calW}{}{s_2}$, it follows
 that $R_{\calS} \mkern1mu (s_1,s_2)$. Thus $R_{\calS}$ is a two-level
 bisimulation relating $s_1$ and~$s_2$. Conclusion, it holds that $s_1
 \FuTSbis{\calS} s_2$.
\end{proof}

\noindent
Above we introduced the notion of a two-level nested $\FuTS$ and an
associated notion of bisimulation. Also in the case of such nested
$\FuTS$, $\FuTS$-bisimulation and behavioral equivalence of the
corresponding functor coincides. Combination of nested $\FuTS$, or
combination of nested and simple $\FuTS$, over the same set of states,
is a straightforward generalization along the lines of
Section~\ref{sec-combined}. 
We will not pursue unfolding of the details here.
In the next section we will encounter
an example of such a construction.

\section{\texorpdfstring{$\FuTS$}{FuTS} Semantics of a language for Markov Automata}
\label{sec-mal}

As a final application of the $\FuTS$ approach to modeling quantitative
behaviour we consider so-called Markov automata~(MA). A Markov
automaton, as proposed in~\cite{EHZ10:concur,EHZ10:lics,Ti+12},
combines non-deterministic and probabilistic behaviour, on the one
hand, with stochastic time behaviour, on the other hand. Therefore, we
need a combination of a nested and a simple $\FuTS$ to model the
respective behaviour.

The definition of an MA here follows~\cite{Ti+12}. We first recall
some definitions from~\cite{Ti+12,DH13:ic} with $\Distr(S) \subseteq
\fsfn{S}{\nnreals}$ denoting the class of (finitely
supported)
probability distributions
over~$S$.

The superposition of non-deterministic and probabilistic behaviour is
provided in Markov automata by means of a combination of a standard
choice operator~`$\cho$' together with the probabilistic extension of
action prefix $a.\SET{ \, \four{p_1}{P_1} \, \Box \cdots \Box \,
  \four{p_h}{P_h} \,}$ for~$a \in \calA$, $h > 0$, and $p_1, \ldots,
p_h \in (0,1]$ such that $p_1 + \cdots + p_h = 1$. The syntactic
construct $\SETponePonephPh$ denotes the distribution
$\mu_{\SETponePonephPh}$ over processes defined by
\begin{displaymath}
  \mu_{\SETponePonephPh} = \tssum_{i=1}^{h} \: [P_i \mapsto p_i]
\end{displaymath}
The intuitive meaning is then obvious: process $a.{\SETponePonephPh}$
performs action~$a$ and then behaves as process~$P$ with probability
$\mu_{\SETponePonephPh} (P)$.

A process language for Markov Automata called MAPA (Markov Automata
Process Algebra) has been proposed in~\cite{Ti+12,Tim12:ctit,Tim13phd}. MAPA includes a rich data system
and is equipped with restrictions to facilitate state space generation
of relatively small models. Below, we consider $\MAL$ as introduced
in~\cite{De+14}. $\MAL$ constitutes a simplified fragment of MAPA
which highlights how nested non-deterministic and probabilistic
behaviour combined with Markovian behaviour can be modeled in the
$\FuTS$ framework.

\begin{defi}
  The set $\prcMAL$ of $\MA$ processes is given by the grammar
  \begin{displaymath}
    P \bnfeq \nil \mid 
    a. \SETponePonephPh \mid 
    \lambda.P \mid P + P \mid P \prlA P \mid X
  \end{displaymath}
  where $a$~ranges over the set of actions~$\calA$, $p_i$~over the
  interval~$(0,1]$, $\lambda$~over~$\poreals$, $A$~over the set of
  finite subsets of~$\calA$ and $X$~over the set of
  constants~$\calX$. For an probabilistic action-prefix
  $a.\SETponePonephPh$ it is required that $h > 0$ and $p_1 + \cdots
  + p_h =1$.  
\end{defi}

\noindent
We assume the same notation, guardedness requirements and
conventions for constant definitions as in Section~\ref{sec-pepa}
for $\PEPA$, $\IML$ and~$\TPC$.

\blankline

\noindent
In the setting of~$\prcMAL$ we use $\amset{P}, \amset{Q}$ to range
over $\fsfn{\, \fsfn{\prcMAL}{\nnreals} \,}{\bools}$ and $\mycal{P},
\mycal{Q}$ to range over $\fsfn{\prcMAL}{\nnreals}$. We use $\mu, \nu$
to range over $\Distr(\prcMAL) \subseteq \fsfn{\prcMAL}{\nnreals}$.
As before, we let $\mycalP_1 + \mycalP_2$ be the pointwise sum of
$\mycalP_1$ and~$\mycalP_2$. (Note, we are adding rates here.) We put
$\chut_{\mkern1mu P} = [ \mkern1mu P \mapsto 1 \mkern1mu ]$
in~$\fsfn{\prcMAL}{\nnreals}$ and define $\mycalP_1 \prlA \mycalP_2 :
\prcMAL \to \nnreals$, for $\mycalP_1, \mycalP_2 \in
\fsfn{\prcMAL}{\nnreals}$ and $A \subseteq \calA$, by
\begin{displaymath}
  (\mycalP_1 \prlA  \mycalP_2)(R) =
  \left\{
    \begin{array}{cl}
      \mycalP_1(R_1)\cdot \mycalP_2(R_2) &
      \mbox{if $R = R_1 \prlA R_2$ for some $R_1, R_2 \in \prcMAL$} \\
      0 & \text{otherwise}
\end{array}
\right.
\end{displaymath}
Note $\mycalP_1 \prlA  \mycalP_2 \in
\fsfn{\prcMAL}{\nnreals}$. Moreover, if $\mu_1, \mu_2 \in
\Distr(\prcMAL)$ then $\mu_1 \prlA \mu_2 \in \Distr(\prcMAL)$ too,
since
$
  \tssum_{R \in \prcMAL} \; ( \mu_1 \prlA \mu_2 )(R)
  =
  \tssum_{R_1, R_2 \in \prcMAL} \; \mu_1(R_1) \cdot \mu_2(R_2)
$
and the latter summation is equal to 
$
  \bigl ( \mkern1mu \tssum_{R_1 \in \prcMAL} \; \mu_1(R_1) \mkern1mu \bigr ) 
  \cdot
  \bigl ( \mkern1mu \tssum_{R_2 \in \prcMAL} \; \mu_2(R_2) \mkern1mu \bigr ) 
$
while $\tssum_{R_1 \in \prcMAL} \; \mu_1(R_1) \mkern1mu$ and
$\tssum_{R_2 \in \prcMAL} \; \mu_2(R_2)$ are both equal to $1$.
For $\amset{P}_1, \amset{P}_2 \in \fsfn{\, \fsfn{\prcMAL}{\nnreals}
  \,}{\bools}$ and $A \subseteq \calA$, we also use constructs
$\amset{P}_1 + \amset{P}_2$ and $\amset{P}_1 \prlA \amset{P}_2$ where
$(\amset{P}_1 + \amset{P}_2)( \mkern1mu \mu) = \amset{P}_1( \mkern1mu
\mu) \vee \amset{P}_2( \mkern1mu \mu)$ is pointwise disjunction, and
$\amset{P}_1 \prlA \amset{P}_2$ is defined by
\begin{displaymath}
  \tssum_{\mu_1, \mkern2mu \mu_2 \colon \amset{P}_1( \mkern1mu \mu_1) = \TRUE 
    \, \land \, \amset{P}_2( \mkern1mu \mu_2) = \TRUE} \;
  [ \mkern2mu \mu_1 \prlA \mu_2 \mapsto \TRUE \mkern1mu ]
\end{displaymath}
Thus $(\amsetP_1 \prlA \amsetP_2)( \mkern1mu \mu) = \TRUE$ iff $\mu =
\mu_1 \prlA \mu_2$, for $\mu_1$ such that $\amsetP_1( \mkern1mu \mu_1)
= \TRUE$ and~$\mu_2$ such that $\amsetP_2( \mkern1mu \mu_2) = \TRUE$.
We overload $\chut_{\mkern-1mu P}$ for~$P \in \prcMAL$; with respect
to $\fsfn{\, \fsfn{\prcMAL}{\nnreals} \,}{\bools}$ we have
$\chut_{\mkern-1mu P} = [ \mkern2mu [ \mkern1mu P \mapsto 1 \mkern1mu
] \mapsto \TRUE \mkern2mu ]$. Because of the contexts no confusion
arises whether to interpret~$\chut_{\mkern-1mu P}$ with respect to
$\fsfn{\, \fsfn{\prcMAL}{\nnreals} \,}{\bools}$ or with respect
to~$\fsfn{\prcMAL}{\nnreals}$.

With the operators defined above in place, and a combined treatment of
actions and probabilities vs.\ stochastic delays, it is
straightforward to capture the semantics of~$\MAL$ with $\FuTS$, cf.~\cite{De+14}.

\begin{defi}
\label{df-mal-sem}
The formal semantics of $\prcMAL$ is given by the $\FuTS$ $\calSmal =
\threetuple{\, \prcMAL}{\mtrans{}_1}{\mtrans{}_2}$, a general $\FuTS$
over the label sets~$\calA$ and $\Delta = \SET{ \mkern1mu \delta
  \mkern1mu }$ and the semirings $\nnreals$,~$\bools$ and~$\nnreals$
again with transition relations ${\mtrans{}_1}$ and ${\mtrans{}_2}$, where 
${\mtrans{}_1} \subseteq {\prcMAL
  \times \calA \times \fsfn{\, \fsfn{\prcMAL}{\nnreals} \,}{\bools}}$
and ${\mtrans{}_2} \subseteq {\prcMAL \times \Delta \times
  \fsfn{\prcMAL}{\nnreals}}$, defined as the least relations satisfying
the rules of Figure~\ref{fig-mal-rules}.  
\end{defi}

\begin{figure}
\begin{displaymath}
  \def\arraystretch{1.3}
\scalebox{0.85}{$
  \begin{array}{@{}c@{}}
  \sosrn{NIL1}{}
    \sosrule
    {a \in \calA}
    {\nil \, \mtrans{a}_1 \, \zerof_\bools} 
 \quad
   \sosrn{NIL2}{}
    \sosrule
    {}
    {\nil \, \mtrans{\delta}_2 \, \zerof_{\nnreals}}
 \quad
   \sosrn{RPF1}{}
    \sosrule
    {a \in \calA}
    {\lambda.P \, \mtrans{a}_1 \, \zerof_\bools}
 \quad   
  \sosrn{RPF2}{}
    \sosrule
    {}
    {\lambda.P  \, \mtrans{\delta}_2 \, [P \mapsto \lambda]}
    \bigskip \\
    \sosrn{APF1}{}
    \sosrule
    {}
    {a.\SETponePonephPh 
      \, \mtrans{a}_1 \, 
      [ \mkern2mu \mu_{\SETponePonephPh} \mapsto \TRUE \mkern1mu ]}
    \bigskip \\
    \sosrn{APF2}{}
    \sosrule
    {b \neq a}
    {a.\SETponePonephPh \, \mtrans{b}_1 \, \zerof_\bools}
    \quad
    \sosrn{APF3}{}
    \sosrule
    {}
    {a.\SETponePonephPh \, \mtrans{\delta}_2 \, \zerof_{\nnreals}} 
    \bigskip \\
\sosrn{CHO1}{}
\sosrule
  {P \, \mtrans{a}_1 \, \amset{P} \quad 
   Q \, \mtrans{a}_1 \, \amset{Q}}
  {P \cho Q \  \mtrans{a}_1 \  \amset{P} \cho \amset{Q}} 
\qquad
\sosrn{CHO2}{}
\sosrule
  {P \, \mtrans{\delta}_2 \, \mycal{P} \quad 
   Q \, \mtrans{\delta}_2 \, \mycal{Q}}
  {P \cho Q \  \mtrans{\delta}_2 \  \mycal{P} \cho \mycal{Q}} 
\bigskip \\
\sosrn{PAR1}{}
\sosrule
  {P \, \mtrans{a}_1 \, \amset{P} \quad
   Q \, \mtrans{a}_1 \, \amset{Q} \quad 
   a \notin A}
  {P \prlA Q \  \mtrans{a}_1 \ 
   ( \, \amset{P} \prlA \chut_{\mkern-1mu Q} \, ) 
   \, + \,
   ( \, \chut_{\mkern-1mu P} \prlA \amset{Q} \, )}
\qquad
\sosrn{PAR2}{}
\sosrule
  {P \, \mtrans{a}_1 \, \amset{P} \quad 
   Q \, \mtrans{a}_1 \, \amset{Q} \quad 
   a \in A}
  {P \prlA Q \  \mtrans{a}_1 \ 
   \amset{P} \prlA \amset{Q}}
\bigskip \\
\sosrn{PAR3}{}
\sosrule
  {P \, \mtrans{\delta}_2 \, \mycal{P} \quad
   Q \, \mtrans{\delta}_2 \, \mycal{Q}}
  {P \prlA Q \  \mtrans{\delta}_2 \ 
   ( \, \mycal{P} \prlA \chut_{\mkern-1mu Q} \, ) 
   \, + \,
   ( \, \chut_{\mkern-1mu P} \prlA \mycal{Q} \, )}
\bigskip \\
\sosrn{CON1}{}
\sosrule
  {P \, \mtrans{a}_1 \, \amset{P} \quad 
   X \dfas P}
  {X \, \mtrans{a}_1 \, \amset{P}}
\qquad
\sosrn{CON2}{}
\sosrule
  {P \, \mtrans{\delta}_2 \, \mycal{P} \quad 
   X \dfas P}
  {X \, \mtrans{\delta}_2 \, \mycal{P}}
   \end{array}
$} 
   \def\arraystretch{1.0}
\end{displaymath}
\caption{$\FuTS$ Transition Deduction System for $\MAL$.}
\label{fig-mal-rules}
\end{figure}

\noindent
By guarded induction we obtain that the finitely supported functions
involved in the definition of~$\mtrans{}_1$ are indeed probability
distributions. Ultimately this relies on the restriction on the
extended prefix, for the process $a. \SETponePonephPh$ the finite sum
$p_1 + \cdots + p_h$ must be equal to~$1$.

\begin{lem}
  \label{lm-distr}
  For all $P \in \prcMAL, a \in \calA$,  functions  $\amset{P}\in \fsfn{\,
    \fsfn{\prcMAL}{\nnreals} \,}{\bools}$ and $\mycalP \in
  \fsfn{\prcMAL}{\nnreals}$, if $P \mtrans{a}_1 \amset{P}$ and
  $\amset{P}(\mycalP) = \TRUE$, then $\mycalP \in \Distr(\prcMAL)$.
  \qed
\end{lem}




\noindent
It is not difficult either to verify that $\calSmal$ is a total and
deterministic combined $\FuTS$, i.e.\ for $P \in \prcMAL$, $a \in
\calA$ we have $P \mtrans{a}_1 \amset{P}$ for exactly one $\amset{P}
\in \fsfn{\, \fsfn{\prcMAL}{\nnreals} \,}{\bools}$ and $P
\mtrans{\delta}_2 \mycal{P}$ for exactly one $\mycal{P} \in
\fsfn{\prcMAL}{\nnreals}$.
 
\begin{lem}
  \label{lm-mal-total-det}
  The general $\FuTS$ $\calSmal$ is total and deterministic.  
  \qed
\end{lem}


\noindent
Below we use $\calSmal = \threetuple {\prcMAL}{\theta_1}{\theta_2}$
with $\theta_1 : \prcMAL \to \fsfn{ \fsfn{\prcMAL}{\reals} }{\bools}$
and $\theta_2 : \prcMAL \to \fsfn{\prcMAL}{\reals}$ induced by
$\mtrans{}_1$ and~$\mtrans{}_2$, respectively. We write
$\twoFuTSbis{\mal}$ for the associated notion of bisimilarity. Recall,
for~$\theta_1$ the relevant definition is
Definition~\ref{df-2-FuTS-bis} on page~\pageref{df-2-FuTS-bis}, while
for~$\theta_2$ we of course refer to Definition~\ref{df-ltfs-bisim} of
page~\pageref{df-ltfs-bisim}, as shown below, for clarity.

\begin{defi}
\label{df-Smal-bis}
An equivalence relation $R \subseteq \prcMAL \times \prcMAL$ is an
$\calSmal$-bisimulation if and only if $R$ is a nested bisimulation
with respect to~$\theta_1$ and a simple bisimulation with respect
to~$\theta_2$.
\end{defi}

\noindent
If we unfold the definitions for the two types of $\FuTS$ bisimulation
we obtain that an equivalence relation $R \subseteq \prcMAL \times
\prcMAL$ is an $\calSmal$-bisimulation, if for all $P_1, P_2 \in \prcMAL$
such that $R(P_1,P_2)$, it holds that
\begin{itemize}
\item for all $a \in \calA$ and $\mu \in \Distr( \prcMAL
  )$: $\tssum_{\mu' \in \Rclass{\mu}} \; \theta_1(P_1)(a)( \mkern1mu
  \mu') \; = \; \tssum_{\mu' \in \Rclass{\mu}} \; \theta_1(P_2)(a)(
  \mkern1mu \mu')$, and
\item for all $Q \in \prcMAL$: $\tssum_{Q' \in \Rclass{Q}} \;
  \theta_2(P_1)(\Edelay)(Q') \; = \; \tssum_{Q' \in \Rclass{Q}} \;
  \theta_2(P_2)(\Edelay)(Q')$
\end{itemize}
with~$R$ on~$\Distr( \prcMAL )$ induced by~$R$ on~$\prcMAL$.
Recall that, for $\mu_1,\mu_2 \in \Distr( \prcMAL )$,
$R(\mu_1,\mu_2)$ if and only if
$
\sum_{Q' \in \Rclass{Q}} \mu_1(Q') =
\sum_{Q' \in \Rclass{Q}} \mu_2(Q')
$
for all $Q \in  \prcMAL$.

\begin{figure}
\begin{displaymath}
\scalebox{0.85}{$
\begin{array}{c}
    \sosrn{ACT}{}
    \sosrule
    {}
    {a.\pCponePonephPh \, \trans{a} \, \mu_{\mkern1mu \pCponePonephPh}}
\qquad
    \sosrn{DELAY}{}
    \sosrule
    {}
    {\lambda.P  \, \Mtrans{\lambda} \, P}
\bigskip \\
    \sosrn{CHO1}{}
    \sosrule
    {P \trans{a} \mu}
    {P + Q \trans{a} \mu}
\qquad
    \sosrn{CHO2}{}
    \sosrule
    {Q \trans{a} \nu}
    {P + Q \trans{a} \nu}
\bigskip \\
    \sosrn{CHO3}{}
    \sosrule
    {P \Mtrans{\lambda}_2 P'}
    {P + Q \Mtrans{\lambda}_2 P'}
\qquad
    \sosrn{CHO4}{}
    \sosrule
    {Q \Mtrans{\lambda}_2 Q'}
    {P + Q \Mtrans{\lambda}_2 Q'}
\bigskip \\
    \sosrn{PAR1}{}
    \sosrule
    {P \, \trans{a} \, \mu \quad a \notin A}
    {P \prlA Q \, \trans{a} \, \mu \prlA \chutQ}
\qquad
    \sosrn{PAR2}{}
    \sosrule
    {Q \, \trans{a} \, \nu \quad a \notin A}
    {P \prlA Q \, \trans{a} \, \chutP \prlA \nu}
\bigskip \\
    \sosrn{PAR3}{}
    \sosrule
    {P \, \trans{a} \, \mu \quad 
      Q \, \trans{a} \, \nu \quad a \in A}
    {P \prlA Q \, \trans{a} \, \mu \prlA \nu}
\bigskip \\
    \sosrn{PAR4}{}
    \sosrule
    {P \, \Mtrans{\lambda} \, P'}
    {P \prlA Q \, \Mtrans{\lambda} \, P' \prlA Q}
\qquad
    \sosrn{PAR5}{}
    \sosrule
    {Q \, \Mtrans{\lambda} \, Q'}
    {P \prlA Q \, \Mtrans{\lambda} \, P \prlA Q'}
\bigskip \\
    \sosrn{REC1}{}
    \sosrule
    {P \trans{a} \mu \quad X:=P}
    {X \trans{a} \mu}
\qquad
    \sosrn{REC2}{}
    \sosrule
    {P \Mtrans{\lambda} P' \quad X:=P}
    {X \Mtrans{\lambda} P'}
\end{array}
$} 
\end{displaymath}
\halflineup
\caption{Standard Transition Deduction System for $\MAL$.}
\label{fig-standard-mal-revised}
\end{figure}

\blankline

\noindent
A standard LTS-based operational semantics of~$\MAL$ is given by the
SOS rules of Figure~\ref{fig-standard-mal-revised}. The semantics is
the similar to the one reported in~\cite{Ti+12,Tim12:ctit}. Here,
however, the technical overhead of decorations on transitions as used
in the above mentioned papers is avoided at the expense of implicit
multiplicities, in line with the treatment of $\PEPA$ and $\IML$ in
Sections \ref{sec-pepa} and~\ref{sec-iml}, respectively. Note, as
$\MAL$ extends~$\IML$, there are separate rules for interactive
transitions (ACT, CHO1--2, PAR1--3 and~REC1) captured by the transition
relation~$\trans{}$, and for Markovian transitions (DELAY, CHO3--4,
PAR4--5, REC2) captured by the transition relation~$\Mtrans{}$.

\begin{defi}
  \label{df-concrete-semantics-mal}
  The semantics of the process language~$\MAL$ 
  is the tuple $(\prcMAL ,\, \calA
  ,\, {\trans{}} ,\, {\Mtrans{}})$ where the probabilistic
  transition relation ${\trans{}} \subseteq {\prcMAL \times \calA
    \times \Distr(\prcMAL)}$ and the standard transition relation
  ${\Mtrans{}} \subseteq {\prcMAL \times \poreals \times \prcMAL}$ are
  given by the SOS rules of Figure~\ref{fig-standard-mal-revised}.
\end{defi}

\noindent
Similar to our treatment of~$\prcIML$ in Section~\ref{sec-iml}, we
introduce the functions $\bfI$ and~$\bfM$ based on the transition
relations $\trans{}$ and~$\Mtrans{}$ of
Definition~\ref{df-concrete-semantics-mal} for~$\prcMAL$. Now, for the
interactive part of~$\MAL$, we have $\bfI \colon \prcMAL \times \calA
\times {\textbf{2}}^{\mkern2mu \Distr(\prcMAL)} \to \bools$ given by
$\bfI \mkern1mu ( P, a, \amset{C} ) = \TRUE$ if the set $\ZSET{ \mu
  \in \amset{C} }{ P \trans{a} \mu }$ is non-empty, for all $P \in
\prcMAL$, $a \in \calA$ and any subset~$\amset{C} \subseteq
\Distr(\prcMAL)$. The Markovian part of~$\MAL$ is similar to that
of~$\IML$. We define for~$\MAL$ the function $\bfM \colon \prcMAL
\times \prcMAL \to \nnreals$ by $\bfM(P,P') = \tssum \MSET{ \lambda }{
  P \Mtrans{\lambda} P' }$. Because of the implicit multiplicities of
the SOS of Definition~\ref{df-concrete-semantics-mal}, the
comprehension is over the multiset of transitions leading from~$P$
to~$P'$ with label~$\lambda$. We also extend $\bfM$, now to $\prcMAL
\times {\textbf{2}}^{\mkern2mu \prcMAL}$, by $\bfM(P,C) = \tssum_{P'
  \myin C} \; \tssum \MSET{ \lambda }{ P \Mtrans{\lambda} P' }$, for
$P \in \prcMAL$ and $C \subseteq \prcMAL$. With the adapted functions
$\bfI$ and~$\bfM$ in place, the notion of strong bisimulation
for~$\MAL$ is defined as follows.

\begin{defi}
\label{df-mal-strong-bisimulation}
  An equivalence relation $R \subseteq \prcMAL \times \prcMAL$ is
  called a strong bisimulation for~$\MAL$ if, for all $P_1, P_2 \in
  \prcMAL$ such that $R(P_1,P_2)$, it holds that
  \begin{itemize}
  \item for all $a \in \calA$ and $\mu \in \Distr( \prcMAL$ ) : $\bfI
    \threetuple {P_1} a {\Rclass{\mu}} \iff \bfI \threetuple {P_2} a
    {\Rclass{\mu}}$ \smallskip
  \item for all $Q \in \prcMAL$: $\bfM \twotuple {P_1}{\Rclass{Q}} \;
    = \; \bfM \twotuple {P_2}{\Rclass{Q}}$
\end{itemize}
with the relation~$R$ on~$\Distr(\prcMAL)$ induced by the relation~$R$
on~$\prcMAL$.  Two processes $P_1, P_2 \in \prcMAL$ are called
strongly bisimilar if it holds that $R(P_1,P_2)$ for a strong
bisimulation~$R$ for~$\MAL$, notation~\hbox{$P_1 \MALsbis
  P_2$}.
\end{defi}

\noindent
Recall, again, that the relation $R \subseteq \prcMAL \times \prcMAL$
induces relation $R \subseteq \Distr( \prcMAL ) \times \Distr( \prcMAL )$ by
$R(\mu_1,\mu_2)$ if and only if
$
\sum_{Q' \in \Rclass{Q}} \mu_1(Q') =
\sum_{Q' \in \Rclass{Q}} \mu_2(Q')
$
for all $Q \in  \prcMAL$.
In line with what we have seen in the previous sections, the crux for
relating the notion of $\calSmal$-bisimulation 
and the notion of strong bisimulation of
Definition~\ref{df-mal-strong-bisimulation} is the following result.

\begin{lem}
  \label{lm-futs-vs-sos-mal}
  \hfill
\begin{enumerate}[label=\({\alph*}]
\item 
  Let $P \in \prcMAL$ and $a \in \calA$. 
  If $P \mtrans{a}_1 \amsetP$ then $P \trans{a} \mu \iff \amsetP( \mu
  ) = \TRUE$. 
\item 
  Let $P \in \prcMAL$. 
  If $P \mtrans{\Edelay}_2 \mycal{P}$ then $\tssum \MSET{ \lambda }{ P
    \Mtrans{\lambda} P' } = \mycal{P}( P' )$. 
\end{enumerate}
\end{lem} 

\begin{proof} \hfill\smallskip

\noindent\(a\ Guarded induction. Let $a \in \calA$.  We treat the
  cases $a.\SETponePonephPh$ and $P_1 \prlA P_2$ for $a \in A$.
 
  Case~$a.\SETponePonephPh$.  
  $a.\SETponePonephPh \mtrans{a}_1 [
  \mkern2mu \mu_{\SETponePonephPh} \mapsto \TRUE \mkern1mu ]$, 
  while 
  $a.\SETponePonephPh \trans{a} \mu_{\SETponePonephPh}$ 
  is the only transition for process
  $a.\SETponePonephPh$.

  Case~$P_1 \prlA P_2$, $a \in A$. Assume $P_1 \prlA P_2 \mtrans{a}_1
  \amsetP$. Then $\amsetP = \amsetP_1 \prlA \amsetP_2$ for $\amsetP_1,
  \amsetP_2 : \fsfn{\prcMAL}{\nnreals} \to \bools$ such that $P_1
  \mtrans{a}_1 \amsetP_1$, $P_2 \mtrans{a}_1 \amsetP_2$. Suppose $P_1
  \prlA P_2 \trans{a} \mu$. Then there exist $\mu_1, \mu_2 \in
  \Distr(\prcMAL)$ such that $P_1 \trans{a} \mu_1$, $P_2 \trans{a}
  \mu_2$ and $\mu = \mu_1 \prlA \mu_2$, since only rule~(PAR3) of
  Figure~\ref{fig-standard-mal-revised} applies. By induction
  hypothesis, $\amsetP_1( \mkern1mu \mu_1) = \TRUE$ and $\amsetP_2(
  \mkern1mu \mu_2) = \TRUE$. Hence $\amsetP( \mkern1mu \mu) = (
  \amsetP_1 \prlA \amsetP_2 )( \mkern1mu \mu_1 \prlA \mu_2 ) = \TRUE$
  by definition of~$\prlA$ on $\fsfn{ \fsfn{\prcMAL}{\nnreals}
  }{\bools}$. Reversely, suppose $\amsetP( \mkern1mu \mu) =
  \TRUE$. Then $\mu = \mu_1 \prlA \mu_2$ for $\mu_1,\mu_2 \in
  \Distr(\prcMAL)$ such that $\amsetP_1( \mkern1mu \mu_1) = \TRUE$ and
  $\amsetP_2( \mkern1mu \mu_2) = \TRUE$. By induction hypothesis, $P_1
  \trans{a} \mu_1$ and $P_2 \trans{a} \mu_2$. Hence $P_1 \prlA P_2
  \trans{a} \mu_1 \prlA \mu_2$ by rule~(PAR3), i.e.\ $P_1 \prlA P_2
  \trans{a} \mu$.
 
  The other cases are left to the reader.\medskip

  \noindent
  \(b\ Guarded induction. Compared to the proof of
  Lemma~\ref{lm-mtrans-trans} there is only one new case, viz.\ for
  processes of the form~$a.\SETponePonephPh$. This case is straightforward,
  since, on the one hand, $a.\SETponePonephPh \mtrans{\Edelay}_2
  \zerof_{\nnreals}$ by definition of~$\mtrans{\Edelay}_2$ and, on the
  other hand, we have that 
  $a.\SETponePonephPh \Mtrans{\lambda} P'$ for \emph{no}~$P' \in
  \prcMAL$ by definition of~$\Mtrans{}$.

  The remaining cases are similar to the proof for the corresponding
  lemma for~$\IML$ and left to the reader.
\end{proof}

\noindent
We are now in a position to relate the notions of $\FuTS$
bisimilarity~$\FuTSbis{\mal}$ and standard strong
bisimilarity~$\MALsbis$ for~$\MAL$.

\begin{thm}
\label{th-correpondence-mal}
For any two processes $P_1, P_2 \in \prcMAL$ it holds that $P_1
\FuTSbis{\mal} P_2$ iff $P_1 \MALsbis P_2$.
\end{thm} 

\begin{proof}
  Let~$R$ be an equivalence relation on~$\prcMAL$.  Pick $P \in
  \prcMAL$, $a \in \calA$ and choose any~$\amsetP \in \fsfn{
    \fsfn{\prcMAL}{\nnreals} }{\bools}$.  Suppose $P \mtrans{a}_1
  \amsetP$.  Thus $\theta_1(P)(a) = \amsetP$. Then we have, for any
  $\mu \in \Distr( \prcMAL )$,
  \begin{displaymath} 
    \def\arraystretch{1.2}
    \begin{array}{rcll}
      \bfI \threetuple {P} a {\Rclass{\mu}}  
      & \Leftrightarrow &
      \exists \, \mu' \in \Rclass{\mu} \colon P \trans{a} \mu'
      & \text{by definition of $\bfI$}
      \\ & \Leftrightarrow &
      \exists \, \mu' \in \Rclass{\mu} \colon \amsetP(\mu') = \TRUE
      & \text{by Lemma~\ref{lm-futs-vs-sos-mal}a}
      \\ & \Leftrightarrow &
      \tssum_{\mu' \myin \Rclass{\mu}} \; \theta_1(P)(a)(\mu')
      & \text{by definition of $\theta_1$}
    \end{array}
    \def\arraystretch{1.0}
  \end{displaymath}
  Note, summation in~$\bools$ is disjunction. Likewise, on the
  Markovian side, we have, for any $Q \in \prcMAL$,
  \begin{displaymath} 
    \def\arraystretch{1.2}
    \begin{array}{rcll}
      \bfM \twotuple {P}{\Rclass{Q}}  
      & = &
      \tssum_{Q' \myin \Rclass{Q}} \; \tssum \MSET{ \lambda }{ P
        \Mtrans{\lambda} Q' } 
      & \text{by definition of $\bfM$}
      \\ & = &
      \tssum_{Q' \myin \Rclass{Q}} \; \amsetP(Q' \mkern1mu ) 
      & \text{by Lemma~\ref{lm-futs-vs-sos-mal}b}
      \\ & = &
      \tssum_{Q' \myin \Rclass{Q}} \; \theta_2(P)(\Edelay)(Q)
      & \text{by definition of $\theta_2$}
    \end{array}
    \def\arraystretch{1.0}
  \end{displaymath}
  Comparing the equations following Definition~\ref{df-Smal-bis} and
  the equations of Definition~\ref{df-mal-strong-bisimulation}, we
  conclude that a strong bisimulation for~$\MAL$ is also an
  $\calSmal$-bisimulation for the $\FuTS$~$\calSmal$, and vice versa.
  From this the theorem follows.
\end{proof}

\noindent
As a corollary of the theorem we obtain that also for~$\MAL$ the
concrete notion of strong bisimilarity $\MALsbis$ is coalgebraically
underpinned, as it coincides, with the behavioral equivalence
$\FuTSbis{\mal}$ that comes with the corresponding $\FuTS$~$\calSmal$.


\section{Concluding remarks}
\label{sec-conclusions}

Total and deterministic state-to-function labeled transition systems,
$\FuTS$s, are a convenient instrument to express the operational
semantics of both qualitative and quantitative process languages.  In
this paper we have discussed the notion of bisimilarity that arises
from a $\FuTS$, possibly involving multiple transition relations, from
a coalgebraic perspective. For $\FuTS$ models of prominent process languages
based on prominent stochastic process algebras we related the induced
notion of bisimulation to the standard equivalences, thus providing
these equivalence with a coalgebraic underpinning.
The main technical contributions of our paper include correspondence
results, viz.\ Theorem~\ref{th-correspondence},
Theorem~\ref{th-combined-correspondence} and
Theorem~\ref{th-nested-correspondence}, that relate in the simple,
combined and the new nested case, bisimilarity of a $\FuTS$~$\calS$ to
behavioural equivalence of the functor associated with~$\calS$. The
result extends to general $\FuTS$ as well.

It is noted in~\cite{BBBRS12}, in the context of weighted automata,
that in general the type of functors~$\fsfn{{\cdot}}{\calR}$ may not
preserve weak pullbacks and, therefore, the notions of coalgebraic
bisimilarity and of behavioural equivalence may not coincide.  A
counter example is provided, cf.~\cite[Section~2.2]{BBBRS12}.
Essential for the construction of the counter-example, in their
setting, is the fact that the sum of non-zero weights may add to
weight~$0$.  The same phenomenon prevents a general proof, along the
lines of~\cite{DeR99}, for coalgebraic bisimilarity and $\FuTS$
bisimilarity to coincide.  In the construction of a mediating
morphism, going from $\FuTS$ bisimulation to coalgebraic bisimulation
a denominator may be zero, hence a division undefined, in case the sum
over an equivalence class cancels out.  In the concrete case
for~\cite{KS08:fossacs}, although no detailed proof is provided there,
this will not happen with~$\nnreals$ as underlying semiring. 
In~\cite[Theorem~5.13]{GS01:entcs} for $\fsfn{{\cdot}}{\calM}$, with
$\calM$ a monoid, a characterization is given for weak preservation
of pullbacks: $\calM$~should be positive and refinable, i.e.\ (i)~$m_1
+ m_2 = 0$ iff $m_1,m_2 = 0$, and (ii)~if $m_1 + m_2 = n_1 + n_2$
there exist $p_{ij}$ such that 
$p_{i1} + p_{i2} = m_i$ and $p_{1j} +p_{2j} = n_j$ for 
$1 \leq i,j \leq 2$. 
The latter condition
is also referred to as the row-column property for $2 \times 2$
matrices over~$\calM$, a property going back to~\cite{Mos99:apal}.
In~\cite{LMV13} we propose to consider semirings which admit a (right)
multiplicative inverse for non-zero elements, and satisfy the
so-called zero-sum property, stating that for a sum $x = x_1 + \cdots
+ x_n$ it holds that $x = 0$ iff $x_i = 0$ for all $i = 1 \ldots
n$. 
The proof follows the set-up of~\cite{DeR99}, hence is different
from~\cite{GS01:au}; we see that zero-sum coincides with positivity,
while the existence of multiplicative inverses guarantees
refinability.
Thus, for semirings involved enjoying these properties, pullbacks
are weakly preserved by $\fsfn{{\cdot}}{\calR}$. Therefore, coalgebraic
bisimilarity and behavioural equivalence are the same. As a
consequence, under conditions which are met by the $\SPC$s discussed
in the preceding, we have that concrete bisimulation,
$\FuTS$-bisimilarity, behavioural equivalence and coalgebraic
bisimilarity coincide.

For typical stochastic process languages based on $\PEPA$ and~$\IMC$
we have shown that the notion of strong equivalence and strong
bisimilarity associated with these calculi, coincides with the notion
of bisimilarity of the corresponding~$\FuTS$. Using these $\FuTS$ as a
stepping stone, the correspondence results bridge between the concrete
notion of bisimulation for $\PEPA$ and $\IML$, and the associated
coalgebraic notions of behavioural equivalence. Hence, from this
perspective, the concrete notions are seen as the natural strong
equivalence to consider. Obviously, classical strong
bisimilarity~\cite{Mil80:lncs,Par81:gi} and bisimilarity for $\FuTS$
over~$\bools$ coincide (see~\cite{KS08:fossacs} or~\cite{LMV13} for
details). Also, strong bisimulation of~\cite{Hil96:phd}, an
alternative to Hillston's notion of strong equivalence covered here,
involving apart from the usual transfer conditions the comparison of
state information, viz.\ the apparent rates, can be treated
with~$\FuTS$. Again the two notions of equivalence coincide. Finally,
we gave an account of how languages based on discrete deterministic
time, $\TPC$, as well as those where stochastic time is integrated
with discrete probability and with non-determinism, $\MAL$, can be
treated in the $\FuTS$ framework. A similar mediating role for~$\FuTS$
applies to these calculi too: the concrete notion of bisimulation
coincides with $\FuTS$ bisimulation, hence coincides with the
corresponding behavioral equivalence.

As mentioned in Section~\ref{sec:introduction}, related work in the area of 
systematic approaches to frameworks for the semantics of $\SPC$---and
quantitative extensions of process calculi in general---includes the study of 
abstract quantitative GSOS, with its application to Weighted Transition Systems~($\WTS$)
\cite{KS08:fossacs,Kli09:mosses,MiP14}. Stochastic GSOS (SGSOS) and Weighted GSOS 
appear to be a special case of  Miculan and Peressotti's \emph{weight function} GSOS.
In~\cite{KS08:fossacs,MiP14}  a treatment is given for $\PEPA$, in line with Section~\ref{sec-pepa} of the present paper.
The formats mentioned above arise from the abstract theory of SOS\@.
A noteworthy result, shown in~\cite{KS08:fossacs}, 
is that stochastic bisimilarity of $\SPC$ defined using the SGSOS format 
is guaranteed to be a congruence. The result is generalized to $\WTS$ in~\cite{Kli09:mosses}.
We did not address the issue of congruences for $\FuTS$ in the present paper.
Nevertheless, we note that {\em Rated Transition Systems}---the semantic 
model used in~\cite{KS08:fossacs}---are very similar to $\RTS$ of
Latella, Massink et al.~\cite{De+08,De+09,De+09a}, which are the instantiation
of simple $\FuTS$ on non-negative real numbers, and that
$\WTS$ are very similar to  simple $\FuTS$. Thus, it is to be expected  
that simple $\FuTS$ can be represented as $\WTS$ using the SGSOS,  which would extend the congruence result to  $\FuTS$. The issue of the relationship with $\WTS$ remains, though, for 
the richer class of combined, nested, and general $\FuTS$, which we leave for further study.

\blankline

\noindent
\emph{Acknowledgments} The authors are grateful to Rocco De~Nicola,
Fabio Gadducci, Daniel Gebler, Michele Loreti, Jan Rutten, and Ana
Sokolova for fruitful discussions on the subject and useful
suggestions. The constructive comments by the reviewers have been of
help and are much appreciated.
DL and MM acknowledge support by EU Project n.~600708 \emph{A
  Quantitative Approach to Management and Design of Collective and
  Adaptive Behaviours} (QUANTICOL).  This research has been partially
conducted while EV was spending a sabbatical leave at the CNR/ISTI\@.
EV gratefully acknowledges the hospitality and support during his stay
in Pisa.

\bibliographystyle{plain}
\bibliography{lmv-lmcs}

\end{document}